\documentclass[12pt,a4paper]{article}
\usepackage[letterpaper,margin=1.in]{geometry}
\usepackage{graphicx}
\usepackage{hyperref}
\usepackage{amsmath}
\usepackage{amssymb}
\usepackage{subcaption}
\usepackage{float}
\usepackage[english]{babel}
\usepackage{amsthm}
\usepackage{mathtools}
\usepackage{enumitem}
\usepackage{mathrsfs}
\usepackage{tikz}
\usepackage{algorithm}
\usepackage{algpseudocode}
\usepackage{diagbox}
\usepackage{amsmath}
\usepackage{indentfirst}
\usepackage{authblk}
\usepackage[T1]{fontenc}

\DeclareMathAlphabet\mathbfcal{OMS}{cmsy}{b}{n}
\DeclareMathOperator{\SW}{SW}

\newcommand{\R}{\mathcal{R}}
\newcommand{\X}{\mathcal{X}}
\newcommand{\p}{\prime}
\newcommand{\C}{\mathcal{C}}
\newcommand{\K}{\mathcal{K}}
\newcommand{\real}{\mathbb{R}}
\newcommand{\bcO}{\mathbfcal{O}}
\newcommand{\bcA}{\mathbfcal{A}}
\newcommand{\bcB}{\mathbfcal{B}}
\newcommand{\cO}{\mathcal{O}}
\newcommand{\legal}{\mathscr{L}}

\newcommand{\va}{\mathbf{a}}
\newcommand{\vp}{\mathbf{p}}

\newtheorem{theorem}{Theorem}[section]
\newtheorem{definition}[theorem]{Definition}
\newtheorem{remark}[theorem]{Remark}
\newtheorem{lemma}[theorem]{Lemma}
\newtheorem{corollary}[theorem]{Corollary}
\newtheorem{claim}[theorem]{Claim}

\date{\today}
\title{A Two-Step Approach to Optimal Dynamic Pricing in Multi-Demand Combinatorial Markets}
\author{Kanstantsin Pashkovich and Xinyue Xie}

\affil{University of Waterloo\\
 Department of Combinatorics and Optimization\\
200 University Avenue West,
Waterloo, ON, Canada  N2L 3G1\\
	\emph{kpashkovich@uwaterloo.ca}, \emph{xinyue.xie@uwaterloo.ca}}

\begin{document}

\maketitle

\begin{abstract}
Online markets are a part of everyday life, and their rules are governed by algorithms. Assuming participants are inherently self-interested, well designed rules can help to increase social welfare. Many algorithms for online markets are based on prices: the seller is responsible for posting prices while buyers make purchases which are most profitable given the posted prices. To make adjustments to the market the seller is allowed to update prices at certain timepoints.  

Posted prices are an intuitive way to design a market. Despite the fact that each buyer acts selfishly, the seller's goal is often assumed to be that of social welfare maximization. Berger, Eden and Feldman recently considered the case of a market with only three buyers where each buyer has  a fixed number of goods to buy and the profit of a bought bundle of items is the sum of profits of the items in the bundle. For such markets, Berger et.\ al.\ showed that the seller can maximize social welfare by dynamically updating posted prices before arrival of each buyer. B\'{e}rczi, B\'{e}rczi-Kov\'{a}cs and Sz\"{o}gi showed that the social welfare can be maximized also when each buyer is ready to buy at most two items.

We study the power of posted prices with dynamical updates in more general cases. First, we show that the result of Berger et.\ al.\ can be generalized from three to four buyers. Then we show that the result of B\'{e}rczi, B\'{e}rczi-Kov\'{a}cs and Sz\"{o}gi  can be generalized to the case when each buyer is ready to buy up to three items. We also show that a dynamic pricing is possible whenever there are at most two allocations maximizing social welfare.\end{abstract}

\section{Introduction}
We study the dynamic pricing that optimizes social welfare in a multi-demand combinatorial market. A \emph{combinatorial market} consists of three components, namely, \emph{ items}, \emph{ buyers}, and \emph{valuation functions}. In particular, we consider a market with a set $\mathcal{X}$ of $m$ items and a set $I$ of $n$ \emph{buyers}, also called \emph{players}. The items are assumed to be heterogeneous and indivisible so that the buyers are not allowed to buy a fractional amount of items. Each buyer $i\in I$ is associated with a valuation function $v_i:2^{\X}\to\R_{\ge0}$, which maps subsets of items, also called \emph{bundles}, to non-negative values. In particular, we focus on multi-demand valuation functions. In this framework, every player $i\in I$ wants $k_i$ items and the value of a bundle is the sum of the values of the $k_i$ most valued items in the bundle.

For example, we can consider a market with four items $\{\alpha,\beta,\gamma,\delta\}$ and three players $\{1,2,3\}$, with the value of each individual item given in Table~\ref{example: valuations}. Suppose that player 1 is bi-demand and players 2 and 3 are unit-demand, that is $k_1=2$ and $k_2=k_3=1$.

\begin{table}[h!]
\begin{center}
\begin{tabular}{|c|c|c|c|c|}
\hline  & $\alpha$ & $\beta$ & $\gamma$ & $\delta$ \\
\hline 1 & 2 & 1 & 1 & 1\\
\hline 2 & 2 & 0 & 1 & 0\\
\hline 3 & 1 & 0 & 0 & 1\\
\hline
\end{tabular}
\end{center}
\caption{An example of valuations in a multi-demand combinatorial market.}\label{example: valuations}
\end{table}

Given a particular pricing of the items, we define the \emph{utility} of a bundle with respect to a player to be the difference between its valuation and price. Of course, a player always picks a bundle that maximizes the utility. Such a bundle is referred to as a \emph{bundle in demand}. 
Suppose the four items in the example above are respectively assigned prices $1.5$, $0.1$, $0.5$ and $0.9$. Then, for player 1, there are two bundles in demand, namely, $\{\alpha,\beta\}$ and $\{\beta,\gamma\}$, that achieve the maximal utility of $1.4$. 

The \emph{social welfare} of a market is defined to be the sum of the valuations of the received bundles across all buyers. In the above example, there are also two allocations that maximize the social welfare. We can either assign $\{\alpha,\beta\}$ to player 1, $\gamma$ to player 2 and $\delta$ to player 3; or assign $\{\beta,\gamma\}$ to player 1, $\alpha$ to player 2 and $\delta$ to player 3. Both these allocations yield a social welfare of $5$.

Our goal is to find an optimal pricing so that while the buyers are free to pick any bundle that maximizes their utility, their choices simultaneously maximize the social welfare. We make the following assumptions. The valuation function of each buyer is publicly known. However, there are two unknowns in the problem that increase the difficulty for finding an optimal pricing. Firstly, the buyers arrive one at a time, but the order of their arrival is unknown. Secondly, when two different bundles are both in demand for a player, that player can break the tie arbitrarily.

Due to the two unknowns mentioned above, the well-known Walrasian prices~\cite{Walras} are insufficient to achieve an optimal allocation~\cite{Hsu2016}. It can be shown that in general no static pricing can achieve more than $2/3$ of the optimal social welfare~\cite{Cohen2016}. Instead, a \emph{dynamic pricing} is necessary. The dynamic pricing is more powerful since after each buyer has picked a bundle in demand, we are allowed to adjust the prices of the remaining items in the market before the arrival of the next buyer. Since the prices can be updated, to show that a dynamic pricing is optimal, it suffices to check that no matter which player $i\in I$ arrives first, any bundle in demand for player $i$ can be extended to an optimal allocation. That is, any choice of any player can always be accommodated afterwards by changing the prices of the remaining items.

In our example, one can see that with the assigned prices the two bundles in demand for player 1 can be extend to an optimal allocation. Similarly, one can check that $\{\alpha\}$ and $\{\gamma\}$ are the only two bundles in demand for player 2 and both extend to an optimal allocation. Similarly, $\{\delta\}$ is the only bundle that is in demand for player 3 and indeed $\delta$ is assigned to player 3 in all optimal allocations.

Ezra et.\ al.\ ~\cite{Ezra2019} proved that dynamic pricing allows to achieve an optimal allocation in markets with homogeneous items and submodular valuations, which is a superclass of multi-demand valuations. As for markets with heterogeneous items, Cohen-Addad et.\ al.\ ~\cite{Cohen2016} developed algorithms for computing the dynamic pricing in unit-demand markets and in gross substitutes markets with a unique optimal allocation. A generalization of this result led to another algorithm by Berger et.\ al.\ which applies to multi-demand markets with at most three players ~\cite{Berger2020}. B{\'{e}}rczi et.\ al.\ showed that a dynamic pricing is possible in combinatorial multi-demand markets where each player is a bi-demand player~\cite{Berczi2021}. One can achieve the optimal social welfare by dynamic pricing in the case of two players where both players have valuation functions induced by certain classes of matroids~\cite{Kakimura2021}. The existence of dynamic pricing in a general multi-demand market remains an open problem.  In this paper, we show that a dynamic pricing exists in multi-demand markets with four players, multi-demand markets with at most two optimal allocations, and tri-demand markets. 

Our technique heavily relies on the concept of legality. An item $x\in\X$ is said to be legal to a player $i\in I$ if there exists an optimal allocation in which $x$ is assigned to $i$. A legal allocation refers to an allocation in which every player $i$ receives $k_i$ legal items. Berger et.\ al.\ proved that optimal allocations and legal allocations are equivalent~\cite{Berger2020}. 

Section~\ref{Combinatorial Market} provides the precise definitions of relevant concepts, including valuation functions, utility, social welfare, dynamic pricing and legality. Sections~\ref{Auxiliary Graph} and~\ref{Rough Price} describe first step of the process, where we impose rough prices in order to limit the choices of each player to their legal items. The algorithm for finding rough prices involves weighted directed graphs, which is a generalization of the technique used in~\cite{Cohen2016}   for unit-demand markets.

Section~\ref{Fine Price} and onward describes second step of the process, where we define fine prices such that the sum of rough prices and fine prices is a dynamic pricing. We do not have an algorithm that finds fine prices in all multi-demand markets. Instead, we focus on three particular types of multi-demand markets. Sections~\ref{Multi-Demand Market with Four Players}, \ref{Multi-Demand Market with Two Optimal Allocations}, and \ref{Market with 3-Demand Players} present algorithms for computing dynamic pricing in multi-demand markets with four players, in multi-demand markets with at most two optimal allocations and in tri-demand markets, respectively. The first two algorithms involve induction on the number of items in the market. We use direct graphs that are constructed based on information about legality to find a subset of items whose price can be easily determined and thus can be removed from the market. The last algorithm involves induction on the number of players and items. 

\subsection{Comparison to Previous Techniques}

Our results on rough prices were independently obtained by B{\'{e}}rczi et.\ al.\ ~\cite{Berczi2021}, in particular the results guaranteeing the existence of rough prices, see Theorem~\ref{thm:rough prices}, and positive utility of legal items for rough prices, see Lemma~\ref{positive utility}. Our techniques to obtain rough prices results are different from  the techniques in~\cite{Berczi2021}. While the paper of B{\'{e}}rczi et.\ al.\ ~\cite{Berczi2021} provides a novel approach based on the duality and careful modifications of an optimal dual solution, we do a construction through graphs generalizing the approach of~\cite{Cohen2016}.

Optimal allocations in a multi-demand market naturally correspond to perfect matchings in some bipartite graph.
In~\cite{Berczi2021} the techniques for finding fine prices in the case of a bi-demand market are inspired by the ideas from Hall's marriage theorem~\cite{Hall}. Similarly our approach for tri-demand market uses similar ideas inspired by Hall's marriage theorem. However, for our proof we need to combine fine prices in a careful way and for this we need to strengthen our statements.

\subsubsection*{Acknowledgement.} We would like to thank Jochen K\"onemann for many fruitful discussions. The first author would also like to thank Chaitanya Swamy for discussions that led to ideas used in design of rough prices.

\section{Combinatorial Market}\label{Combinatorial Market}

In this section, we provide a formal definition of the studied problem together with the central notions.

We consider a market with a set $\X=[m]$ of heterogeneous and indivisible \emph{items} and a set $I=[n]$ of \emph{buyers}, also called players. For all $i\in I$, the buyer $i$ is $k_i$-demand and has a \emph{multi-demand valuation function} $v_i: 2^{\X}\to\real_{\ge0}$, i.e.\ for all $\X^{\p}\subseteq\X$ we have
\[v_i(\mathcal{\X^{\p}}):=\max\left\{\sum\limits_{x\in\mathcal{\hat{\X}}}v_i(x):\mathcal{\hat{\X}}\subseteq\mathcal{\X^{\p}},|\mathcal{\hat{\X}}|\le k_i\right\}\,.\]
In the case when $k_i=1$, 2 or 3, the buyer $i$ is called a unit-demand, bi-demand or tri-demand buyer, respectively. Without loss of generality we can assume that $m=\sum_{i=1}^n k_i$.

Since we work with multi-demand buyers, we encode every valuation function $v_i$, $i\in I$ as an $m$-dimensional vector $\mathbf{v}^i=(v_x^i)_{x\in \X}\in\real_{\ge0}^m$ indexed by the items, where $v_x^i=v_i(x)$ for all $x\in\X$. Then for any subset $\X^{\p}\subseteq\X$ with $|\X^{\p}|\le k_i$, we have the relation $v_i(\X^{\p})=\langle\mathbf{v}_i,\chi(\X^{\p})\rangle$, where $\chi$ is the characteristic vector of $\X^{\p}$ on $\X$ and $\langle \mathbf{a}, \mathbf{b}\rangle$ is the scalar product of vectors $\mathbf{a}$ and $\mathbf{b}$.

\subsection{Utility and Optimal Allocation}

We consider the \emph{prices of the items} to be a vector with positive entries. To encode prices of items, we use the \emph{price vector}, which is a positive $m$-dimensional vector $\vp=(p_x)_{x\in\X}\in\real_{>0}^m$ indexed by the items in the market.  

The \emph{price of a bundle} is the sum of the prices of all items in the bundle. Given a price vector~$\vp\in\real_{>0}^m$, for every player we define the \emph{utility} of a bundle to be the difference between its valuation and price. Formally, the utility of a buyer $i\in I$ from a bundle $\X^{\p}\subseteq\X$ is defined as follows
\[u_i(\X^{\p},\vp)\coloneqq v_i(\X^{\p})-\langle\vp,\chi(\X^{\p})\rangle.\]
For simplicity, when $\X^{\p}=\{x\}$ is a singleton, we may omit the set brackets and write $u_i(x,\vp)$.

Every player prefers to pick a \emph{bundle in demand}, i.e.\ a bundle maximizing the utility. Formally speaking, given a price vector $\vp\in\real_{>0}^m$, the bundles in demand for a player $i\in I$ are exactly the ones in the set
\[D_{\vp}(v_i)\coloneqq\arg\max\limits_{\X^{\p}\subseteq\X}\{u_i(\X^{\p},\vp)\}\,.\]
In case where multiple bundles are in demand for a given player, the player can break the tie arbitrarily and pick any bundle in demand. 

An \emph{allocation} is a collection $\bcO=(\cO_i)_{i\in I}$ of disjoint subsets of items in $\X$, where the bundle~$\cO_i$ is allocated to player $i$. For each allocation, we can calculate the \emph{social welfare} by summing the valuations corresponding to each player and the bundle assigned to this player. In other words, the social welfare of an allocation $\bcO$ is defined as
\[\SW(\bcO)\coloneqq\sum\limits_{i=1}^n v_i(\cO_i).\]
An \emph{optimal allocation} is an allocation that achieves the maximal social welfare.

\textbf{We assume that in every optimal allocation, every item is assigned to some player.} Thus, we also assume that if an item $x\in\X$ is assigned to player $i\in I$ in some optimal allocation then $v_i(x)>0$. This ensures that it is possible for $x$ to contribute positively to the utility of player $i$.

\subsection{Legality and Dynamic Pricing}

Our goal is to find a \emph{dynamic pricing} that guarantees social welfare optimization, despite the fact that each individual buyer always maximizes their own utility. We introduce the useful concepts of legal items, legal bundles, and legal allocations as in~\cite{Berger2020}.

Since we are allowed to update prices after each purchased bundle, to achieve the maximal social welfare, it suffices to ensure that at any point in time, for each buyer  every bundle in demand can be completed to an optimal allocation.

\begin{definition}\label{dynamic pricing definition}
A price vector $\vp\in\real_{>0}^m$ is an (optimal) dynamic pricing if for any player $i\in I$ and any bundle $\X^{\p}\subseteq \X$ that is in demand for player $i$, there exists an optimal allocation in which the bundle $\X^{\p}$ is allocated to player $i$.
\end{definition}

To formalize the main results we use the notion of \emph{legality} introduced in~\cite{Berger2020}. An item $x$ is \emph{legal} for a player $i$ if there exists an optimal allocation in which $x$ is allocated to $i$. A bundle $\X^{\p}$ is \emph{legal}  for a player $i$ if $|\X^{\p}|=k_i$ and every item $x\in\X^{\p}$ is legal for player $i$. A \emph{legal allocation} $\bcO$ is an allocation in which $\cO_i$ is legal for player $i$ for all $i\in I$.

A very powerful Theorem~\ref{legality-optimality} by Berger et al. states that optimal allocations and legal allocations are equivalent~\cite{Berger2020}. Consequently, we use these terms interchangeably. 

\begin{theorem}\label{legality-optimality}
An allocation is legal if and only if it is optimal.
\end{theorem}

Theorem~\ref{legality-optimality} allows us to restate the definition of dynamic pricing in terms of legality. Namely, we can rewrite the definition of dynamic pricing as follows. A price $\vp\in\real_{>0}^m$ is a dynamic pricing if for any player $i\in I$ and any bundle $\X^{\p}\subseteq \X$ that is in demand for player $i$, there exists a legal allocation in which the bundle $\X^{\p}$ is allocated to player $i$.

\section{Auxiliary Graph}\label{Auxiliary Graph}

In this section, we define the \emph{auxiliary graph of a market}. Later, we also use existence of Walrasian prices to prove certain properties of the auxiliary graph. Indeed, since multi-demand valuations is a subclass of gross-substitutes valuations, Walrasian prices always exist~\cite{Gul1999}. In the next sections, we use auxiliary graph  to compute prices.

\subsection{Construction of the  Auxiliary Graph}\label{Construction of auxiliary graph}

The auxiliary graph can be seen as a generalization of the relation graph in Cohen-Addad et.\ al.\ algorithm~\cite{Cohen2016}. The relation graph is constructed based on a single optimal allocation, while our auxiliary graph captures all optimal allocations.

We list all optimal allocations in the market as $\bcO^1,\bcO^2,...,\bcO^k$ and denote the maximum social welfare by $\widehat{\SW}$. For each player $i\in I$, we define a vector
\[\va^i\coloneqq\frac{1}{k}\sum\limits_{t=1}^k\chi(\cO_i^t)\,,\]
which averages the bundles assigned to player $i$ in all listed optimal allocations. Intuitively, this corresponds to a ``fractional bundle'' assigned to the player $i$; and thus $\bcA\coloneqq(\va^i)_{i\in I}$ is an average of all optimal allocations. In fact, had we allowed items to be divisible, $\bcA$ itself would be a valid ``fractional allocation'' that achieves the social welfare $\widehat{\SW}$. This intuition is captured by the following lemma. Note that by ``fractional allocation'', we mean that all items are assigned completely to the players, i.e. $\sum_{i=1}^n\va^i=\mathbf{1}$, and that the total amount of items received by each player is equal to their demand, i.e. $\langle\va^i,\mathbf{1}\rangle=k_i$. We define the value of a ``fractional bundle'' to be the linear combination of the values of the items in the bundle, i.e.\ for all players $i\in I$ and ``fractional bundle'' $\va$ we have
\[v_i(\va)=\langle\mathbf{v}^i,\va\rangle\] whenever $\langle\va,\mathbf{1}\rangle \leq k_i$, i.e.\ whenever the ``fractional number'' of items allocated to player $i$ by $\va$ is at most the demand.

\begin{lemma}\label{A is a fractional allocation}
$\bcA$ is a fractional allocation with social welfare $\widehat{\SW}$.
\end{lemma}

\begin{proof}
Each optimal allocation $\bcO^t$, $t=1,\ldots, k$, allocates every item to exactly one player, i.e.\ $\bcO^t$ is a partition of $\X$. Thus we have
$$\sum\limits_{i=1}^n\chi(\cO_i^t)=\mathbf{1}\,.$$

Summing over all $k$ optimal allocations, we have
$$\sum\limits_{i=1}^n\va^i=\sum\limits_{i=1}^n\frac{1}{k}\sum\limits_{t=1}^k\chi(\cO_i^t)=\frac{1}{k}\sum\limits_{t=1}^k\sum\limits_{i=1}^n\chi(\cO_i^t)=\frac{1}{k} \sum\limits_{t=1}^k\mathbf{1}=\mathbf{1}.$$

Since every player $i$ receives exactly $k_i$ items in any optimal allocation, for every player $i\in I$ we have
$$
\langle\va^i,\mathbf{1}\rangle=\langle\frac{1}{k}\sum\limits_{t=1}^k\chi(\cO_i^t),\mathbf{1}\rangle=\frac{1}{k}\sum\limits_{t=1}^k\langle\chi(\cO_i^t),\mathbf{1}\rangle=\frac{1}{k}\sum\limits_{t=1}^kk_i= k_i.
$$

So, $\bcA$ is a fractional allocation and we can calculate the social welfare of $\bcA$ as follows
\begin{align*}
\SW(\bcA)&=\sum\limits_{i=1}^n\langle\mathbf{v}^i,\va^i\rangle=\sum\limits_{i=1}^n\langle\mathbf{v}^i,\frac{1}{k}\sum\limits_{t=1}^k\chi(\cO_i^t)\rangle=\frac{1}{k}\sum\limits_{t=1}^k\sum\limits_{i=1}^n\langle\mathbf{v}^i,\chi(\cO_i^t)\rangle\\&=\frac{1}{k}\sum\limits_{t=1}^k\SW(\bcO^t)
=\frac{1}{k}\sum\limits_{t=1}^k\widehat{\SW}=\widehat{\SW}\,,
\end{align*}

finishing the proof.
\end{proof}

To facilitate the discussion of legal items, we introduce the following notations. For each player $i\in I$, let us define
\[\K_i:=\bigcup\limits_{t=1}^k\cO_i^t\subseteq\X\quad\text{and}\quad\R_i:=\bigcap\limits_{t=1}^k\cO_i^t\subseteq\X\,.\]

Then, $\K_i$ denotes the set of all items that are legal to player $i$ and $\R_i$ denotes the set of all items that are legal only  to player $i$.

The fractional allocation $\bcA$ is closely related to legality. We can learn whether an item $x$ is legal to a player $i$ and whether $x$ is legal to players other than player $i$ by comparing the value of $\va_x^i$ with $0$ and $1$.

\begin{remark}\label{R-K-a}
For every item $x\in\X$ and every player $i \in I$, we have the following equivalences.
\begin{enumerate}
\item$x\in \R_i$ if and only if $\va_x^i=1$.
\item$x\in \K_i$ if and only if $ 0<\va_x^i\le1$.
\item$x\notin \K_i$ if and only if $\va_x^i=0$.
\end{enumerate}
\end{remark}

Finally, we define the auxiliary graph using the fractional allocation $\bcA$.  Intuitively speaking, the edges in the graph represent preferences that we aim to enforce by rough prices. The desired properties of rough prices are discussed in Section~$\ref{Rough Price}$.

\begin{definition}\label{auxiliary graph definition}
An auxiliary graph is a graph $H=(\X,E)$, such that for all player $i\in I$ and items $x,y\in\X$, whenever $\va_x^i>0$ and $\va_y^i<1$, we have an edge $e_i=x\to y\in E$ induced by player $i$ with weight
\[w(x\to y)=v_i(x)-v_i(y)\,.\]
\end{definition}

Note that for $x,y\in\X$ the auxiliary graph $H$ can potentially have multiple edges $x\to y$.

\subsection{Cycles in the Auxiliary Graph}

In this section, we prove that there are no negative-weight cycles in the auxiliary graph. This fact ensures that shortest paths in the auxiliary graph are well defined. Such paths are studied later. We also find the necessary and sufficient condition for the occurrence of a $0$-weight cycle in the auxiliary graph.

To prove that negative-weight cycles do not exist, we consider the utility of bundles and ``fractional bundles'' with respect to the \emph{Walrasian prices}. The concept of \emph{Walrasian equilibrium} is closely related to optimal allocations and bundles in demand. In particular, Paes Leme and Wong proved the Second Welfare Theorem, which implies that in any optimal allocation the bundle received by each player is a bundle in demand with respect to Walrasian prices~\cite{PaesLeme2020}. 

\begin{definition}
A Walrasian equilibrium consists of a price vector $\vp^W\in\real^m$ and an allocation $\bcO$ such that for any player $i$ the bundle $\cO_i$ is in demand. $\vp^W$ is called a Walrasian price and $\bcO$ is called a Walrasian allocation induced by $\vp^W$.
\end{definition}

As mentioned before, for a market with multi-demand players, Walrasian equilibrium always exists~\cite{Gul1999}. In the remaining parts of the paper, we denote Walrasian prices by $\vp^W\in\real^m$.

\begin{theorem}[Second Welfare Theorem]
Any Walrasian price and any optimal allocation form a Walrasian equilibrium.
\end{theorem}

 We first show that any small enough compatible perturbation of the ``fractional bundle'' assigned by $\bcA$ to a certain player is a ``convex combination'' of integral bundles, so that later we can apply the Second Welfare Theorem to compare the utility of these integral bundles with the utility of the bundles that are in demand for the same player.

\begin{lemma}\label{convex combination}
If $a_x^i>0$ and $a_y^i<1$ for some player $i\in I$ and items $x,y\in\X$, then $\va^i-\delta\boldsymbol{\alpha}$ is a convex combination of characteristic vectors of sets with size $k_i$, where $\boldsymbol{\alpha}=\mathbf{\hat{e}}_x-\mathbf{\hat{e}}_y\in\real^m$ and $0\le\delta\le\frac{1}{k}$. Here, $\mathbf{\hat{e}}_x$ and $\mathbf{\hat{e}}_y$ are vectors in $\real^m$ with the entry corresponding to $x$ and $y$ equal to $1$, respectively, and all other entries equal to $0$.
\end{lemma}

\begin{proof}
Consider a player $i\in I$ such that $a_x^i>0$ and $a_y^i<1$ and consider $0\le\delta\le\frac{1}{k}$. Consider the polytope
\[P:=\{\mathbf{b}\in\real^m:\mathbf{0}\le\mathbf{b}\le\mathbf{1},\,\sum_{x=1}^m b_x= k_i\}\,.\]

From Lemma~\ref{A is a fractional allocation}, we have $\langle\va^i,\mathbf{1}\rangle=k_i$ and so we have
\[\langle\va^i-\delta\boldsymbol{\alpha},\mathbf{1}\rangle=\langle\va^i,\mathbf{1}\rangle-\delta\langle\boldsymbol{\alpha},\mathbf{1}\rangle= k_i-\delta(\langle\mathbf{\hat{e}}_x,\mathbf{1}\rangle-\langle\mathbf{\hat{e}}_y,\mathbf{1}\rangle)=k_i\,.\]

By definition of $\va^i$, we have that $\va^i-\delta\boldsymbol{\alpha}$ satisfies the constraints defining the polytope $P$ whenever  $0\le\delta\le\frac{1}{k}$.
Since the polytope $P$ is defined by a totally unimodular matrix, all extreme points of $P$ are integral. This implies that $\va^i-\delta\boldsymbol{\alpha}$ can be written as a convex combination of the extreme points of $P$, but the extreme points of $P$ are the vectors in $\mathbb{Z}^m$ with exactly $k_i$ ones. These extreme points are exactly the characteristic vectors of sets with size $k_i$.
\end{proof}

When there is a cycle in the auxiliary graph, we can reallocate items along the cycle to obtain another fractional allocation by assigning a larger fraction of item $y$ and a smaller fraction of item $x$ to player $i$ for every edge $e_i=x\to y$ in the cycle. The total change in utility is given by the weight of the cycle. In particular, in the following lemma, we show that if there is a negative-weight cycle, then we can perturb the fractional allocation $\bcA$ to obtain an allocation with a strictly higher utility. We later combine the results from Lemmas $\ref{convex combination}$ and $\ref{negative cycle}$ to show that the auxiliary graph does not contain any negative-weight cycle.

\begin{lemma}\label{negative cycle}
If the auxiliary graph $H$ has a negative-weight cycle $\C=x_0e_{i_0}x_1e_{i_1}x_2...x_le_{i_l}x_0$, where $x_0,x_1,...,x_l\in\X$ are distinct and $e_{i_0},...,e_{i_l}\in E$ are induced by players $i_0,...,i_l$ respectively, then
$$\sum\limits_{j=0}^l\langle\mathbf{v}^{i_j}-\vp^W,\va^{i_j}-\delta\boldsymbol{\alpha}^j\rangle>\sum\limits_{j=0}^l\langle\mathbf{v}^{i_j}-\vp^W,\va^{i_j}\rangle,$$
where $0<\delta\le\frac{1}{k}$ and $\boldsymbol{\alpha}^j=\mathbf{\hat{e}}_{x_j}-\mathbf{\hat{e}}_{x_{j+1}}$ for all $j=0,\ldots, l$. Here, we identify $x_{l+1}$ with $x_0$.
\end{lemma}

\begin{proof}
Consider the total change in valuation after the allocation is perturbed by $\boldsymbol{\alpha}^j$ for $j=0,\ldots, l$,
\[
\sum\limits_{j=0}^l\langle\mathbf{v}^{i_j},\delta\boldsymbol{\alpha}^j\rangle=\delta\sum\limits_{j=0}^l\langle\mathbf{v}^{i_j},\mathbf{\hat{e}}_{x_j}-\mathbf{\hat{e}}_{x_{j+1}}\rangle=\delta\sum\limits_{j=0}^lw(e_{i_j})=\delta w(\C)<0\,.
\]

Since $\sum\limits_{j=0}^l\va^j=0$, we get
\[
\sum\limits_{j=0}^l\langle\vp^W,\delta\boldsymbol{\alpha}^j\rangle=0\,.
\]

Putting these equalities together, we get
\[\sum\limits_{j=0}^l\langle\mathbf{v}^{i_j}-\vp^W,\va^{i_j}-\delta\boldsymbol{\alpha}^j\rangle=\sum\limits_{j=0}^l\langle\mathbf{v}^{i_j}-\vp^W,\va^{i_j}\rangle-\sum\limits_{j=0}^l\langle\mathbf{v}^{i_j},\delta\boldsymbol{\alpha}^j\rangle+\sum\limits_{j=0}^l\langle\vp^W,\delta\boldsymbol{\alpha}^j\rangle>\sum\limits_{j=0}^l\langle\mathbf{v}^{i_j}-\vp^W,\va^{i_j}\rangle\,,\]
finishing the proof.
\end{proof}

Using Lemma $\ref{convex combination}$ and the Second Welfare Theorem, we can show that for any player $i_j$, $j=0,\ldots, l$, perturbations by $\delta\boldsymbol{\alpha}_j$ cannot achieve higher utility with respect to the prices $\vp^W$ than every integral bundle that is in demand for player $i_j$. Together with the result of Lemma~\ref{negative cycle}, this proves that no negative-weight cycles can exist in the auxiliary graph $H$. This fact is captured in the following theorem.

\begin{theorem}\label{no negative cycle in auxiliary graph}
There is no negative-weight cycle in the auxiliary graph $H$.
\end{theorem}

\begin{proof}
Suppose, for contradiction, that there is a negative-weight cycle $\C=x_0e_{i_0}x_1e_{i_1}x_2...x_le_{i_l}x_0$. We use the same notation as in Lemma $\ref{negative cycle}$. For $j=0,\ldots, l$, let us consider $\va^{i_j}-\delta\boldsymbol{\alpha}^j$. By Lemma $\ref{convex combination}$, we can write
$$\va^{i_j}-\delta\boldsymbol{\alpha}^j=\sum\limits_{q=1}^uc_q\chi(\X_q),$$
for some $u>1$ and $c_q\ge0$, such that $\sum_{q=1}^uc_q=1$ and for each $1\le q\le u$, $|\X_q|= k_{i_j}$. For all $t=1,\ldots, k$, by the Second Welfare Theorem the bundle $\cO_{i_j}^t$ is in demand for player $i_j$. So,
$$\langle\mathbf{v}^{i_j}-\vp^W,\va^{i_j}-\delta\boldsymbol{\alpha}^j\rangle=\langle\mathbf{v}^{i_j}-\vp^W,\sum\limits_{q=1}^uc_q\chi(\X_q)\rangle=\sum\limits_{q=1}^uc_q\langle\mathbf{v}^{i_j}-\vp^W,\chi(\X_q)\rangle\le\langle\mathbf{v}^{i_j}-\vp^W,\chi(O_{i_j}^t)\rangle$$

Summing over all $k$ optimal allocations, we have
\begin{align*}
\langle\mathbf{v}^{i_j}-\vp^W,\va^{i_j}-\delta\boldsymbol{\alpha}^j\rangle=\frac{1}{k}\sum\limits_{t=1}^k\langle\mathbf{v}^{i_j}-\vp^W,\va^{i_j}-\delta\boldsymbol{\alpha}^j\rangle\leq\frac{1}{k}\sum\limits_{t=1}^k\langle\mathbf{v}^{i_j}-\vp^W,\chi(O_{i_j}^t)\rangle=\langle\mathbf{v}^{i_j}-\vp^W,\va^{i_j}\rangle\,.
\end{align*}

Since the above inequality holds for all $j=1,\ldots, l$, summing over all  $j=1,\ldots, l$ we get
\[\sum\limits_{j=1}^l\langle\mathbf{v}^{i_j}-\vp^W,\va^{i_j}-\delta\boldsymbol{\alpha}^j\rangle\le\sum\limits_{j=1}^l\langle\mathbf{v}^{i_j}-\vp^W,\va^{i_j}\rangle\,,\]
contradicting Lemma~\ref{negative cycle} and finishing the proof.
\end{proof}

Let us consider a cycle with weight 0. By a similar argument as in Lemma~\ref{negative cycle}, if we apply similar perturbations to a 0-weight cycle, there should be no change in the utility or the social welfare. Since ``fractional bundles'' can be considered as convex combinations of integral bundles, this indicates that for every edge $e_i=x\to y$ in the cycle, both items $x$ and $y$ are allocated to the player $i$ in some optimal allocation, and hence are legal to $i$. This motivates the following theorem, where we show that for edge $e_i=x\to y$ to be in a 0-weight cycle, $x,y\in\K_i\setminus\R_i$ is a necessary and sufficient condition.

\begin{lemma}\label{a and 0-weight cycle}
For any two items $x,y\in\X$, if in the auxiliary graph $H$ there is an edge $e_{\hat{i}}$ from $x$ to $y$ induced by player $\hat{i}$, then $e_{\hat{i}}$ is in a 0-weight cycle if and only if $0<\va_x^{\hat{i}}<1$ and $0<\va_y^{\hat{i}}<1$.
\end{lemma}

\begin{proof}

Let us do a direct proof for both directions. Consider two items $x,y\in\X$ and suppose there is an edge $e_{\hat{i}}=x\to y$ with weight $v_{\hat{i}}(x)-v_{\hat{i}}(y)$.

Suppose that the edge $e_{\hat{i}}$ is in a 0-weight cycle $\C$. Let $\widehat{\SW}$ and $\widehat{\SW}_2$ be the maximal and the second maximal social welfare, respectively. Define
$
\Delta_{\SW}:=\frac{1}{2}(\widehat{\SW}-\widehat{\SW}_2)$. Note that if $\widehat{\SW}_2$ is not well defined then the statement of the lemma holds, so we can assume that $\widehat{\SW}_2$ is well defined and $\Delta_{\SW}>0$.
By definition of the auxiliary graph $H$, the existence of $e_{\hat{i}}$ implies $\va_x^{\hat{i}}>0$ and $\va_y^{\hat{i}}<1$, so it suffices to show $\va_x^{\hat{i}}\neq1$ and $\va_y^{\hat{i}}\neq0$. Suppose, for contradiction, that $\va_x^{\hat{i}}=1$ or $\va_y^{\hat{i}}=0$. We consider the following two cases.

\textbf{Case $\va_x^{\hat{i}}=1$.} 
Consider a new set of valuation functions $(v_i^{\p})_{i\in I}$ derived from $(v_i)_{i\in I}$ by defining $v_{\hat{i}}^{\p}(x):=v_{\hat{i}}(x)-\Delta_{\SW}$ and $v_i^{\p}:=v_i$ otherwise. By the definition of $\Delta_{\SW}$ and by the assumption that $\va_x^{\hat{i}}=1$, we can prove that the set of optimal allocations is the same for the original and the new valuation functions. So, $\{\bcO^t\}_{1\le t\le k}$ is exactly the set of all optimal allocations corresponding to the new valuation functions $(v_i^{\p})_{i\in I}$. Hence, the new auxiliary graph $H^{\p}$ is the same as $H$ and $\mathcal{C}$ is a cycle in $H^{\p}$. The weights of all edges in $\mathcal{C}$ remain unchanged, except that the weight on edge $e_{\hat{i}}$ becomes
\[w^{\p}(e_{\hat{i}})=v_{\hat{i}}^{\p}(x)-v_{\hat{i}}^{\p}(y)=v_{\hat{i}}(x)-\Delta_{\SW}-v_{\hat{i}}(y)=w(e_{\hat{i}})-\Delta_{\SW}\,,\]
where $w^{\p}$ denotes the weights in $H^{\p}$. So,
\[w^{\p}(\C)=w(\C)-\Delta_{\SW}=0-\Delta_{\SW}<0\,.\]
So $\C$ is a negative-weight cycle in the new auxiliary graph $H^{\p}$, contradicting Theorem~\ref{no negative cycle in auxiliary graph}.

\textbf{Case $\va_y^{\hat{i}}=0$.}  
This case is analogous to the above case.
We can consider a new set of valuation functions $(v_i^{\p\p})_{i\in I}$ derived from $(v_i)_{i\in I}$ by defining $v_{\hat{i}}^{\p\p}(y):=v_{\hat{i}}(y)+\Delta_{\SW}$ and $v_i^{\p\p}:=v_i$ otherwise to obtain a similar contradiction.

Conversely, suppose $0<\va_x^{\hat{i}}<1$ and $0<\va_y^{\hat{i}}<1$ and let us show that the edge $e_{\hat{i}}$ is in a 0-weight cycle $\C$. By the definition of the auxiliary graph, there is an edge from $y$ to $x$ induced by player $\hat{i}$. Let us denote it by $e_{\hat{i}}^{\p}$. The cycle $xe_{\hat{i}}ye_{\hat{i}}^{\p}x$ is a 0-weight cycle since
\[w(xe_{\hat{i}}ye_{\hat{i}}^{\p}x)=w(e_{\hat{i}})+w(e_{\hat{i}}^{\p})=v_{\hat{i}}(x)-v_{\hat{i}}(y)+v_{\hat{i}}(y)-v_{\hat{i}}(x)=0\,.\]
\end{proof}

\section{Rough Price}\label{Rough Price}

In this section, we provide the exact definition of \emph{rough prices} and describe an algorithm for finding rough prices. The algorithm is based on the auxiliary graph discussed in the previous section. 

\begin{definition}\label{rough price definition}
A price $\vp^R\in\real_{>0}^m$ is called a rough price if for every player $i\in I$, it satisfies the following conditions:
\begin{enumerate}[ref={\arabic*}]
\item For all $x\in \R_i$, $y\notin \R_i$, $u_i(x,\vp^R)>u_i(y,\vp^R)$.\label{p^R-prop.1}
\item For all $x,y\in \K_i\setminus \R_i$, $u_i(x,\vp^R)=u_i(y,\vp^R)$.\label{p^R-prop.2}
\item For all $x\in \K_i$, $y\notin \K_i$, $u_i(x,\vp^R)>u_i(y,\vp^R)$.\label{p^R-prop.3}
\end{enumerate}
\end{definition}

 From this definition, it is more apparent why the edges in the auxiliary graph represent preference constraints. Let us take two items $x,y\in\X$ and a player $i\in I$. If $x,y\in\K_i\setminus\R_i$ then we want player $i$ to be indifferent between $x$ and $y$; and so we include both edges $x\to y$ and $y\to x$ in $H$. If $x\in\R_i$ and $y\notin\R_i$, or $x\in\K_i$ and $y\notin\K_i$, then we want player $i$ to strongly prefer $x$ to $y$; and so we only include the edge $x\to y$ in $H$. Using Remark $\ref{R-K-a}$, one can check that this is equivalent to Definition~\ref{auxiliary graph definition} of the auxiliary graph.

The weight of every edge $x\to y$ in $H$ is defined to be the difference between the values of $x$ and $y$ for the corresponding player. For each item, we define the rough price to be the negative of the weight of a shortest path from some artificial source vertex to that item. Such a price counterbalances all the difference in the valuation functions. We also introduce a small perturbation of $\epsilon$ on certain edges to ensure that the inequalities in~\ref{p^R-prop.1} and~\ref{p^R-prop.3} in Definition~\ref{rough price definition} are strict. The details of our construction are outlined in the following algorithm.

\begin{algorithm}[H]
\caption{Finding rough prices in multi-demand markets.}\label{rough price}
\begin{algorithmic}[1]
\State Construct the auxiliary graph $H=(\X,E)$.
\State Define $E^{\p}:=\{e\in E:e\text{ is in some 0-weight cycle in } H\}$.
\State Define $\epsilon_c:=\min\{w(\C):\text{$\C$ is a cycle in $H$ and }w(\C)>0\}>0$.
\State Define $\epsilon_v:=\min\{v_i(x):i\in I,x\in\X\}>0$.
\State Define $\epsilon:=\frac{\min\{\epsilon_c,\epsilon_v\}}{m+1}>0$.
\State Define a new graph $H^{\p}$ based on $H$ by adding a source vertex $s$ and adding an edge $s\to x$ for each item $x\in\X$.
\State For every edge $e=x\to y$ in $H^{\p}$, define the weight $w^{\p}$ in $H^{\p}$ by
\[w^{\p}(e)=\begin{cases}
w(e)\quad&\text{if }e\in E^{\p}\\
w(e)-\epsilon\quad&\text{if }e\in E\setminus E^{\p}\\
0\quad&\text{if }x=s\,.
\end{cases}
\]
\State For every $x\in\X$, define $\delta(s,x)$ to be the weight of a shortest path from $s$ to $x$ in $H^{\p}$.
\State Define $\vp^R\in\real^m$ as $p^R_x:=-\delta(s,x)+\epsilon$ for every $x\in\X$.
\State \Return $\vp^R$.
\end{algorithmic}
\end{algorithm}

Through the choice of $\epsilon_c$, we ensure that $\epsilon$ is small enough so that all cycles in $H^{\p}$ have nonnegative weight as shown in the following lemma. This guarantees the existence of a shortest path from $s$ to every item. Through the choice of $\epsilon_v$, we ensure that the utility of any legal item remains positive as shown in Lemma~\ref{positive utility}. We see in the next section that this guarantees that a bundle in demand for player $i$ always contains $k_i$ items.

\begin{lemma}\label{cycle in H'}
There is no negative-weight cycle in the  graph $H^{\p}$.
\end{lemma}

\begin{proof}
Let $\C$ be a simple cycle in $H^{\p}$. Since the source vertex $s$ has no incoming edge, $\C$ does not include $s$ and hence is also a cycle in $H$. Since there are $m$ items, $\C$ has length at most $m$. 

If $\C$ is a $0$-weight cycle in $H$, then it is also a $0$-weight cycle in $H^{\p}$, because the weights of edges in~$\C$ are the same in $H$ and $H^{\p}$. Otherwise, $\C$ is a positive-weight cycle in $H$. By construction, we have $\epsilon\le\frac{\epsilon_c}{m+1}$ where $\epsilon_c$ is the minimum weight among all positive-weight cycles in~$H$. Thus, we have
\[w^{\p}(\C)\geq w(\C)-m\epsilon>w(\C)-\epsilon_c\geq0\,.\]
\end{proof}

From Lemma~\ref{cycle in H'} we know that the prices $\vp^R$ in Algorithm~\ref{rough price} are well-defined. We now show that these are indeed the desired rough prices.

\begin{theorem}\label{thm:rough prices}
The prices $\vp^R\in\real^m$ given by Algorithm $\ref{rough price}$ are rough prices.
\end{theorem}

\begin{proof}
For all $x\in\X$, $s\to x$ is a path from $s$ to $x$, so
\[p_x=-\delta(s,x)+\epsilon\ge-w^{\p}(s\to x)+\epsilon=0+\epsilon>0\]
So $\vp^R$ is indeed a vector in $\real_{>0}^m$.

Now let us show that $\vp^R$ satisfies Definition~\ref{rough price definition}. Consider a player $i\in I$ and two items $x,y\in\X$. We consider the following three cases, each corresponding to a property in Definition~\ref{rough price definition}.

\textbf{Case $x\in \R_i$ and $y\notin \R_i$.}
By Remark $\ref{R-K-a}$, $\va_x^i=1$ and $\va_y^i<1$. So, there is an edge $e_i=x\to y$ with weight $w(e_i)=v_i(x)-v_i(y)$ in $H$. By Lemma~\ref{a and 0-weight cycle} we have $e_i\notin E^{\p}$, so
\[w^{\p}(e_i)=w(e_i)-\epsilon=v_i(x)-v_i(y)-\epsilon\,.\] 
Notice that a shortest path from $s$ to $x$ together with the edge $e_i$ is a path from $s$ to $y$, so we have
\[\left(-p_x^R+\epsilon\right)+(v_i(x)-v_i(y)-\epsilon)=\delta(s,x)+w^{\p}(e_i)\ge\delta(s,y)=-p_y^R+\epsilon\,,\]
and thus
\[u_i(x,\vp^R)=v_i(x)-p_x^R\ge v_i(y)-p_y^R+\epsilon=u_i(y,\vp^R)+\epsilon>u_i(y,\vp^R)\,.
\]

\textbf{Case $x,y\in \K_i\setminus\R_i$.}
By Remark~\ref{R-K-a}, $0<\va_x^i<1$ and $0<\va_y^i<1$. So, there is an edge $e_i=x\to y$ with weight $w(e_i)=v_i(x)-v_i(y)$ in $H$. By Lemma~\ref{a and 0-weight cycle} we have $e_i\in E^{\p}$, so
\[w^{\p}(e_i)=w(e_i)=v_i(x)-v_i(y).\]
By the same argument as in the above case, we have
\[\left(-p_x^R+\epsilon\right)+(v_i(x)-v_i(y))=\delta(s,x)+w^{\p}(e_i)\ge\delta(s,y)=-p_y^R+\epsilon,\]
and thus
\[u_i(x,\vp^R)=v_i(x)-p_x^R\ge v_i(y)-p_y^R=u_i(y,\vp^R)\,.
\]
Similarly, we have $u_i(y,\vp^R)\ge u_i(x,\vp^R)$; and  therefore, $u_i(x,\vp^R)=u_i(y,\vp^R)$.

\textbf{Case $x\in \K_i$ and $y\notin \K_i$.}
By Remark~\ref{R-K-a}, $0<\va_x^i\le1$ and $\va_y^i=0$. This case is similar to the case $x\in \R_i$ and $y\notin \R_i$.
\end{proof}

Now let us show that the found rough prices guarantee positive utilities for all legal items.

\begin{lemma}\label{positive utility}
For every player $i\in I$ and for every legal item $x\in \K_i$, we have $u_i(x,\vp^R)>0$.
\end{lemma}

\begin{proof}
Take player $i\in I$ and item $x\in \K_i$. Let $se_sx_1e_{i_1}x_2...x_{l-1}e_{i_{l-1}}x_l=x$ be a shortest path from $s$ to $x$.
By Lemma~\ref{cycle in H'}, all cycles in $H^{\p}$ have nonnegative weight. So, any cycle in the above path must be a 0-weight cycle. Removing such cycles does not change the weight of the path, so we assume that $x_1,x_2,...x_l$ are all distinct. Therefore, we have $l\le m$.

For all $j=1,\ldots, l$, we have $\va_{x_j}^{i_j}>0$ since $e_{i_j}$ exists, so there exists $t_j\in \{1,\ldots, k\}$, such that~$x_j\in\cO_{i_j}^{t_j}$. We consider the following lower bound on the utility of item $x$ with respect to player $i$. Note that we identify player $i_l$ with player $i$ and item $x_{l+1}$ with item $x_1$.
\begin{align*}
u_i(x,\vp^R)&=v_{i}(x)-p_x^R=v_{i_l}(x_l)+\delta(s,x)-\epsilon=v_{i_l}(x_l)+w^{\p}(e_s)+\sum\limits_{j=1}^{l-1}w^{\p}(e_{i_j})-\epsilon\\
&\ge v_{i_l}(x_l)+0+\sum\limits_{j=1}^{l-1}(v_{i_j}(x_j)-v_{i_j}(x_{j+1})-\epsilon)-\epsilon\\
&=\sum\limits_{j=1}^{l}v_{i_j}(x_j)-\sum\limits_{j=1}^{l-1}v_{i_j}(x_{j+1})-l\epsilon\\
&=k\left(\frac{1}{k}\sum\limits_{j=1}^{l}v_{i_j}(x_j)-\frac{1}{k}\sum\limits_{j=1}^{l}v_{i_j}(x_{j+1})\right)+v_{i_l}(x_1)-l\epsilon\\
&\ge k(\SW(\bcA)-\SW(\bcB))+v_{i_l}(x_1)-m\epsilon,
\end{align*}
where $\bcB=(\mathbf{b}^i)_{i\in I}$ is the fractional allocation obtained from $\bcA$ by reallocating $\frac{1}{k}$ of item $x_{j+1}$ to player $i_j$ instead of player $i_{j+1}$ for $j=1,\ldots, l$. By construction, we have $\va_{x_j}^{i_j}\ge\frac{1}{k}$ for all $j=1,\ldots, l$. So $\bcB$ is indeed a valid fractional allocation. Formally speaking, 
for all $j=1,\ldots, l$ we have
\[\mathbf{b}^{i_j}=\va^{i_j}-\frac{1}{k}\boldsymbol{\alpha}^j\quad\text{where}\quad\boldsymbol{\alpha}^j=\mathbf{\hat{e}}_{x_j}-\mathbf{\hat{e}}_{x_{j+1}},\]
and for all other players $i$ we have $\mathbf{b}^i=\va^i$.

Using the same arguments as in Theorem~\ref{no negative cycle in auxiliary graph}, we get 
\[\sum\limits_{j=1}^l\langle\mathbf{v}^{i_j}-\vp^W,\mathbf{b}^{i_j}\rangle=\sum\limits_{j=1}^l\langle\mathbf{v}^{i_j}-\vp^W,\va^{i_j}-\frac{1}{k}\boldsymbol{\alpha}^j\rangle\le\sum\limits_{j=1}^l\langle\mathbf{v}^{i_j}-\vp^W,\va^{i_j}\rangle\,.\]

Using $\sum\limits_{j=1}^l\boldsymbol{\alpha}^j=\mathbf0$, we get
\[\sum\limits_{j=1}^l\langle\vp^W,\frac{1}{k}\boldsymbol{\alpha}^j\rangle=0\,.\]

Thus, we get
\[\sum\limits_{j=1}^l\langle\mathbf{v}^{i_j},\mathbf{b}^{i_j}\rangle=\sum\limits_{j=1}^l\langle\mathbf{v}^{i_j},\va^{i_j}-\frac{1}{k}\boldsymbol{\alpha}^j\rangle\le\sum\limits_{j=1}^l\langle\mathbf{v}^{i_j},\va^{i_j}\rangle\,.\]

Since $\SW(\bcB)\le\SW(\bcA)$, we get
\begin{align*}
u_i(x,\vp^R)\ge k(\SW(\bcA)-\SW(\bcB))+v_{i_l}(x_1)-m\epsilon\ge v_{i_l}(x_1)-m\epsilon\ge v_{i_l}(x_1)-\frac{m}{m+1}\epsilon_v>0\,.
\end{align*}
\end{proof}

\section{Fine Price}\label{Fine Price}
We have introduced Algorithm $\ref{rough price}$ for finding rough prices $\vp^R$ in a multi-demand market. In this section, we define \emph{fine prices} and show that the sum of rough and fine prices is a dynamic pricing. We call the market that we have been working with the original market, or \emph{Market \#1}. We now define a market induced by Market \#1, which we call \emph{Market \#2}.

Market \#2 has items $\X^2:=\X\setminus\left(\bigcup\limits_{i\in I}\R_i\right)$ and buyers $I_2:=\{i\in I\,:\,\R_i\neq \K_i\}$, where each player $i\in I_2$ is a $\hat{k}_i$-demand player with $\hat{k}_i\coloneqq k_i-|\R_i|$ and valuation function
\[\hat{v}_i(x):=\begin{cases}
u_i^*\quad&\text{if }x\in  \widehat{\K}_i:=\K_i\setminus \R_i\\
0\quad&\text{if }x\in\X^2\setminus \widehat{\K}_i\,,
\end{cases}\]
where $u_i^*\coloneqq u_i(x,\vp^R)$ for $x\in \widehat{\K}_i$. Note that by~\ref{p^R-prop.2} in Definition~\ref{rough price definition} of rough prices, $u_i^*$ is well-defined; and by Lemma~\ref{positive utility}, we have $u_i^*>0$ for every player $i\in I_2$.

For $i\in I_2$, let us define
\[\Delta_i\coloneqq u_i^*-\max(\{0\}\cup \{u_i(x,\vp^R)\,:\,x\notin \K_i\}).\]

Note that $\Delta_i>0$ for each  $i\in I_2$ by Lemma~\ref{positive utility} and~\ref{p^R-prop.3} in Definition~\ref{rough price definition} of rough prices. Let us also define $\Delta:=\min\{\Delta_i:i\in I_2\}>0$, where we assume that $I_2$ is not empty since otherwise Market~\#2 is trivial.

\begin{definition}\label{definition fine prices}
A price $\vp^F\in\real_{>0}^{|\X^2|}$ is called a fine price of Market \#1 if it satisfies the following conditions.
\begin{enumerate}[ref={\arabic*}]
\item $\vp^F$ is a dynamic pricing in Market \#2.
\item For all $x\in\X^2$ and for all $i\in I_2$, $0<p_x^F<\Delta\le\Delta_i$ .\label{range of p^F}
\end{enumerate}
\end{definition}

Let us assume that we are given a fine price $\vp^F$ and  rough price $\vp^R$ for Market \#1, then we can combine them into a dynamic pricing $\vp^D\in\real_{>0}^m$ for Market \#1 as follows
\[
p_x^D:=\begin{cases}
p_x^R+p_x^F\quad&\text{if }x\in \X^2\\
p_x^R\quad&\text{if }x\in\X\setminus \X^2\,.
\end{cases}
\]

We first note that $\vp^D$ is within the desired range. That is, $\vp^D$ is positive; and from~\ref{range of p^F} in Definition~\ref{definition fine prices} and the construction of $\Delta$, one can show that for each player every legal item yields a positive utility. We now show that $\vp^D$ is indeed a dynamic pricing in Market \#1. In the next section, we proceed by finding fine prices in certain types of markets.

Intuitively, the purpose of setting $\vp^F$ is to motivate each player $i$ to differentiate the items in $\K_i\setminus\R_i$. We further note that as long as $\vp^F$ is a dynamic pricing in Market \#2 that satisfies ~\ref{range of p^F} in Definition~\ref{definition fine prices}, the numerical values of $\vp^F$ are irrelevant. By construction of the rough prices, a player $i\in I_2$ only picks items from $\widehat{\K}_i$ and all such items have the same value in Market \#2. So the player $i$ always picks the $\hat{k}_i$ items in $\widehat{\K}_i$ with the lowest prices, regardless of the numerical value of the prices. As a result, it suffices to consider the ordering of the prices in $\vp^F$. Most importantly, the values of $u_i^*$, $i\in I_2$, are irrelevant, we can simply consider a  valuation function that maps any legal item to $2$ and any illegal item to $1$. We call such Market \#2 a \emph{simplified} multi-demand market. The exact definition  is provided in the next section in Definition~\ref{simplified multi-demand market}.

Hereafter, we focus on Market \#2 and construct an algorithm for finding the dynamic pricing $\vp^F$ in Market \#2. In particular, we consider three types of markets, namely, multi-demand markets with four players, multi-demand markets with two optimal allocations, and tri-demand markets.

\section{Legality Graph}\label{Legality Graph}

In this section, we define a weighted directed graph called \emph{legality graph}. As its name suggests, this graph is constructed based on legality information about the items. We also discuss the properties of this graph, especially the existence of cycles. Later, we use legality graphs to find a dynamic pricing in markets with four players and in markets with two optimal allocations.

Firstly, we formalize our observation regarding the Market \#2 in the previous section into the following definition. To make the notation clearer, we also define a \emph{legality} function which replaces the role of valuation functions.

\begin{definition}\label{simplified multi-demand market}
Consider a set $\X$ of items and a set $I$ of players such that for all $i\in I$, player $i$ is $k_i$-demand, where $k_i\ge1$. Let the vector of demands be defined as $\mathbf{k}:=(k_i)_{i\in I}$. Let $\legal:I\to2^\X$ be the function given by $\legal(i)=\{x\in\X\,:\,x\text{ is legal to }i\}$ for all $i\in I$. Then, $A:=(\X,I,\mathbf{k},\legal)$ is called a simplified multi-demand market. We refer to the function $\legal$ as the legality function of this simplified market. If $|I|\ge 2$ and $|\X|=\sum_{i\in I}k_i$, then $A$ is said to be saturated.
\end{definition}

Hereafter, we assume that the market  is a simplified multi-demand market. We consider a saturated market $A=(\X_A,I_A,\mathbf{k},\legal)$. Let $m=|\X_A|$ and $n=|I_A|$. We assume that every item is legal to at least two players. Using the notation from previous sections, we can write $\R_i=\varnothing$ and $\K_i=\legal(i)$ for all $i\in I$. Whenever necessary, for $i\in I$, we consider every item in $\legal(i)$ as having valuation 2 and every item in $\X_A\setminus\legal(i)$ as having valuation 1. Notice that these definitions and assumptions are consistent with the properties of Market \#2 as discussed in the previous section. Let $\bcO$ be a legal allocation in market $A$.

Our approach to finding a dynamic pricing in multi-demand market with four players or two optimal allocation is by induction on the number of items. So, it is useful to make the following definition of a \emph{submarket}.

\begin{definition}
Given a subset $\X_B\subseteq\X_A$. For all $i\in I_A$, let $k_i^{\p}:=|O_i\cap\X_B|$. Let $I_B:=\{i\in I_A:k_i^{\p}>0\}\subseteq I_A$. Define a new market $B=(\X_B,I_B,(k_i^{\p})_{i\in I_B},\legal)$. Note that here $\legal$  represents the legality with respect to market $A$. Suppose market $B$ has the same valuation functions as market $A$ with the domain restricted to $\X_B$. Then we say that market $B$ is a submarket of market $A$ induced by the legal allocation $\bcO$.
\end{definition}

Whenever it is not relevant for our argumentation, we refer to $B$ as a submarket of $A$ without specifying a particular legal allocation.

We now define a legality graph $G$ with respect to a legal allocation $\bcO$. This is related to the preference graph in Cohen-Addad et.\ al.\ paper~\cite{Cohen2016}, but we use legality, instead of valuation functions while constructing this graph.

\begin{definition}
The legality graph $G$ with respect to $\bcO$ is defined to be the directed graph $(\X_A,E)$, such that for any two items $x,y\in\X_A$ the edge $x\to y$ is in $E$ if and only if $x\in\cO_i$ and $y\in\legal(i)\setminus\cO_i$ for some player $i\in I_A$.
\end{definition}

Now that we have introduced the concepts of submarkets and legality graphs, we show how submarkets of $A$ relate to market $A$ in terms of legality. Remark~\ref{legality in B implies legality in A} states that if an item is legal to a player in some submarket of $A$, then that item is also legal to the same player in market $A$. Consequently, Remark~\ref{submarket implies subgraph} states that the legality graph with respect to the restriction of $\bcO$ to a submarket $A$ induced by $\bcO$ must be a subgraph of $G$.

\begin{remark}\label{legality in B implies legality in A}
Let $B$ be a submarket of market $A$ with items $\X_B$ and players $I_B$. For each player $i\in I_B$, if an item $x\in\X_B$ is legal to player $i$ in market $B$ then $x$ is also legal to player $i$ in market~$A$.
\end{remark}

\begin{remark}\label{submarket implies subgraph}
Let market $B$ be the submarket of $A$ induced by $\bcO$ with items $\X_B$. Let $\bcO^B$ be the restriction of $\bcO$ to $B$. Let $G^B$ be the legality graph corresponding to $\bcO^B$. Then, $G_B$ is a subgraph of $G$.
\end{remark}

Next, we provide two main results about the cycles in the legality graph. Firstly, we define the \emph{reallocation of items with respect to a cycle} in Definition~\ref{reallocation along a cycle}, which always yields another legal allocation. We use this result to move between any two legal allocations and to move between any two legality graphs. Secondly, every edge and vertex in the legality graph is in some cycle. This guarantees the existence of cycles containing a certain edge or a certain vertex. These properties are useful in the next two sections.

\begin{definition}\label{reallocation along a cycle}
Suppose $\C=x_1\to x_2\to...\to x_k=x_1$ is a cycle in the legality graph $G$  with respect to $\bcO$, and $x_j\in\cO_{i_j}$ for $j=1,\ldots, k-1$. The allocation obtained from $\bcO$ by assigning item $x_{j+1}$ to player $i_j$ for $j=1,\ldots, k-1$, is called the reallocation of $\bcO$ with respect to $\C$.
\end{definition}

Using Theorem~\ref{legality-optimality}, we see that optimal allocations are related to each other through disjoint simple cycles in the legality graph. As a consequence, we can show that every edge and every vertex in the legality graph are in a cycle.

\begin{remark}\label{optimal allocations are related by disjoint cycles}
Let $\bcO^{\p}$ be an optimal allocation in market $A$ different from the optimal allocation~$\bcO$. Then, there exists a set of disjoint simple cycles in $G$
such that we can obtain $\bcO^{\p}$ from $\bcO$ by reallocating items with respect to each of these cycles. 
\end{remark}

\begin{remark}\label{reallocation along cycle}
An allocation is legal if and only if it can be obtained from $\bcO$ by reallocations  with respect to a disjoint set of simple cycles in $G$.
\end{remark}

Using Remark $\ref{optimal allocations are related by disjoint cycles}$, we now show that given an edge $x\to y$, we can either define a bijection on the items that maps $x$ to $y$ or find two intersecting cycles in the legality graph such that one contains $x$ and the other contains $y$. We then conclude that every edge in $G$ is in a cycle. It follows directly that every vertex in $G$ is in a cycle.

\begin{corollary}\label{every edge is in a cycle}
Every edge in the legality graph $G$ is in some cycle.
\end{corollary}

\begin{proof}
Consider an edge $x\to y$ in the legality graph $G$. Suppose $x\in\cO_i$ for some player $i$, then $y\in\legal(i)\setminus\cO_i$.
If there exists a legal allocation $\bcO^{\p}$ such that $x\in \cO_i\setminus \cO_i^{\p}$ and $y\in \cO_i^{\p}\setminus \cO_i$, then we can let $f_i:\cO_i\setminus \cO_i^{\p}\to \cO_i^{\p}\setminus \cO_i$ be a bijection such that $f_i(x)=y$. Similarly we can fix an arbitrary bijection $f_j:\cO_j\setminus \cO_j^{\p}\to \cO_j^{\p}\setminus \cO_j$ for every player $j$ different from $i$.  Consider bijections $f_j$ for every player $j$, we find a set of simple cycles as in Remark~\ref{optimal allocations are related by disjoint cycles} such that  $x\to y$ is in  one of these simple cycles.

Otherwise, for all legal allocations $\bcO^{\p}$, we have that $y\in \cO_i^{\p}\setminus \cO_i$ implies $x\in \cO_i\cap \cO_i^{\p}$. Suppose $y\in \cO_i^1\setminus \cO_i$ for some legal allocation $\bcO^1$. Consider an item $x_1\in\cO_i\setminus \cO_i^1$, then $x_1\to y$ is an edge in $G$ and is in some simple cycle $\C_1$, by arguments above. Let $\bcO^2$ be a legal allocation in which $x\in\cO_i\setminus\cO_i^2$. Such $\bcO^2$ exists since $x$ is legal to at least two different players. Consider an item $y_2\in\cO_i^2\setminus\cO_i$, then $x\to y_2$ is an edge in $G$ and is in some simple cycle $\C_2$.

Suppose, for contradiction that $\C_1$ and $\C_2$ do not intersect. Let $\bcO^{\p}$ be the reallocation of $\bcO$ with respect to both cycles $\C_1$ and $\C_2$. By Remark~\ref{reallocation along cycle}, the allocation $\bcO^{\p}$ is legal. In the allocation $\bcO^{\p}$, the item $x$ is no longer allocated to $i$, but $y$ is allocated to $i$. This contradicts the hypothesis that $y\in\cO_i^{\p}\setminus\cO_i$ implies $ x\in\cO_i\cap\cO_i^{\p}$ for every legal allocation $\cO_i^{\p}$.
So, we have that $\C_1$ and $\C_2$ intersect at some item $\hat{x}$, then the path from $\hat{x}$ to $x$ in $\C_2$, edge $x\to y$ and the path from $y$ to $\hat{x}$ in $\C_1$ form a cycle which contains the edge $x\to y$.
\end{proof}

\begin{corollary}\label{every vertex is in a cycle}
Every vertex in the legality graph $G$ is in some cycle.
\end{corollary}

Finally, for the sake of exposition we usually consider a cycle in which no two items are assigned to the same player. Therefore, we make the following definition of a \emph{uniquely assiged cycle} and prove in Lemma~\ref{cycle of items allocated to different players} that every item is contained in a cycle with this property.

\begin{definition}
Let $\C=x_0\to x_1\to x_2...x_{l-1}\to x_l=x_0$ be a simple cycle in $G$. We say that $\C$ is a uniquely assigned cycle if for all $i\in I_A$ we have
\[|\{x_j\,:\,j=1,\ldots, l\}\cap\cO_i|\in\{0,1\}\,.\]
\end{definition}

\begin{lemma}\label{cycle of items allocated to different players}
Every vertex in the legality graph $G$ is contained in some uniquely assigned cycle.
\end{lemma}

\begin{proof}
Consider an item $x_0\in\X_A$. By Corollary~\ref{every vertex is in a cycle}, $x_0$ is in some cycle $\C=x_0\to x_1\to x_2...x_{l-1}\to x_l=x_0$ in $G$. Let us select such cycle $\C$ of minimum possible length. Suppose $x_j\in O_{i_j}$ for $j=0,\ldots,l$. Then, for $j=0,\ldots, l-1$, we have that $x_{j+1}$ is legal but not allocated to player $i_j$. If all $i_j$ for $j=0,\ldots, l-1$ are distinct, then $\C$ is a desired uniquely assigned cycle. Otherwise, suppose $i_j=i_{j^{\p}}$ for some $j$ and $j^\p$ such that $0\le j<j^{\p}\le l-1$. Then $x_{j+1}$ and $x_{j^{\p}+1}$ are legal but not allocated to $i_j=i_{j^{\p}}$. This means that the edge $x_j\to x_{j^{\p}+1}$ exists in $G$ and $x_0\to x_1\to\ldots\to x_j\to x_{j^{\p}+1}\to x_{j^{\p}+2}\to\ldots\to x_{l-1}\to x_l=x_0$ is a cycle in $G$. This cycle with length $l-j^{\p}+j$ is strictly smaller than the cycle $\C$ we started with and it contains the item $x_0$, contradicting the selection of $\C$. 
\end{proof}

\section{Multi-Demand Market with Four Players}\label{Multi-Demand Market with Four Players}

To prove that dynamic prices exist for markets with four buyers, we would like to use induction on the number of items in the market. For this, we would like to find an appropriate set of items $\X_r\subseteq\X_A$ to remove from the market. By inductive hypothesis, a market with items $\X_A\setminus\X_r$ has a dynamic pricing; so by assigning either the ``lowest'' or the ``highest'' prices to the items in $\X_r$, we plan to obtain the dynamic pricing for the original market $A$. 

As mentioned at the end of Section~\ref{Fine Price}, we are not concerned with the numerical values of the prices but rather with their ordering. In particular, we assume that the original market $A$ is a \emph{simplified multi-demand market}. We assume that all prices are within the appropriate range given by~\ref{range of p^F} in Definition~\ref{definition fine prices} and focus on the ordering of the prices. Under these assumptions, a bundle in demand is a legal bundle with the lowest fine prices. So, it suffices to find the proper ordering of the items in terms of prices, so that whenever a player picks a lowest priced legal bundle, it is extendable to a legal allocation.

We first consider a special case that is easily solvable by induction on the number of items in the market.

\begin{lemma}\label{item legal to all players}
Suppose an item $\hat{x}\in\X_A$ is legal to all players. Let $B$ be the submarket of $A$ induced by $\bcO$ with items $\X_A\setminus\{\hat{x}\}$. Then any pricing $\vp\in\real_m$ satisfying the following conditions is a dynamic pricing in market $A$.
\begin{enumerate}
\item For all $x\in\X_A\setminus\{\hat{x}\}$, $p_{\hat{x}}>p_x$.\label{item legal to all players.1}
\item $(p_x)_{x\in\X_A\setminus\{\hat{x}\}}$ is a dynamic pricing in market $B$. \label{item legal to all players.2}
\end{enumerate}
\end{lemma}

\begin{proof}
Suppose that $\hat{x}\in\cO_{\hat{i}}$ for some $\hat{i}\in I_A$. Every player $i$ different from $\hat{i}$ is a $k_i$-demand player  in market $B$. Condition~\ref{item legal to all players.1} ensures that player $i$ only picks items from market $B$ and condition~\ref{item legal to all players.2} ensures that such choices can be extended to an optimal allocation in market $B$. This guarantees that any bundle in demand for player $i$ in the market $B$ with respect to the price $\vp$ can be extended to a legal allocation in the market $A$ by assigning $\hat{x}$ to $\hat{i}$. 

On the other hand, any bundle in demand for player $\hat{i}$ with respect to the price $\vp$ has size $k_{\hat{i}}$. Conditions~\ref{item legal to all players.1} and~\ref{item legal to all players.2}  guarantee that $k_{\hat{i}}^{\p}=k_{\hat{i}}-1$ of these items can be extended to some legal allocation $\bcO^B$ in market~$B$; and the $k_{\hat{i}}$-th item, call it $x^{\p}$, can be either $\hat{x}$ or some item in market~$B$. If $x^{\p}=\hat{x}$, then we finish the proof since this bundle can be extended  to a legal allocation in the market $A$. Otherwise, suppose that $x^{\p}\in\cO_{i^{\p}}^B$ for some $i^{\p}$ different from $\hat{i}$. We can compensate the player $i^{\p}$ by assigning $\hat{x}$ to $i^{\p}$ in order to obtain a legal allocation in market $A$.
\end{proof}

Lemma~\ref{item legal to all players} allows us to assume that, in market $A$, there does not exist an item that is legal to all players. We define the following two types of \emph{removable subsets} of items. We show that when a removable subset exists, we can always find a dynamic pricing by assuming that a dynamic pricing can be found in a submarket with less than $m$ items.

\begin{definition}
Suppose there is an item $x_c\in\cO_{i_c}$ for some player $i_c$ such that for every player $ i\in I_A$ who is different from $i_c$ and for whom $x_c$ is legal, there exists an item $x_i\in\X_A\setminus\{x_c\}$ such that $x_i\in\cO_i\cap\legal(i_c)$. That is, for every player $i$ such that $x_c$ is legal but not allocated to $i$, there is an item $x_i$ that is allocated to player $i$ in $\bcO$ and is legal to player $i_c$. Let $I_r:=\{i\in I_A:x_c\in\legal(i)\}$. Then, the set $\X_r=\{x_c\}\cup\{x_i:i\in I_r\setminus\{i_c\}\}$ is called a removable set of type I with respect to legal allocation $\bcO$; for this removable set the item $x_c$ is called a central item.
\end{definition}

Intuitively, when a removable set of type I exists, we can set the price of $x_c$ to be ``low'' and the prices of all other items in $\X_r$ to be ``high'' so that no player picks items from $\X_r\setminus\{x_c\}$. For players not in $I_r$ the item $x_c$ is illegal, so their behaviour is solely determined by the dynamic pricing in the submarket with items $\X_A\setminus\X_r$. For the players in $I_r$ the item $x_c$ is legal. By constructing $X_r$ to contain items that are allocated to $I_r$ and legal for $i_c$, we make sure if a player picks $x_c$, we can always compensate by assigning some other item in $\X_r$ to player $i_c$ if needed. This leads to the following lemma.

\begin{lemma}\label{prices of removable set of type I}
Let $\X_r\subseteq\X_A$ be a removable set of type I with a central item $x_c\in\X_r$. Let $B=(\X_B,I_B,(k_i^{\p})_{i\in I_B},\legal)$ be the submarket of $A$ induced by $\bcO$ with $\X_B=\X\setminus\X_r$. Then any pricing $\vp\in\real^m$ satisfying the following conditions is a dynamic pricing in market $A$.
\begin{enumerate}
\item For all items $x\in\X_r\setminus\{x_c\}$ and $y\in\X_B$, $p_x>p_y>p_{x_c}$.
\item $(p_x)_{x\in\X_B}$ is a dynamic pricing in market $B$.\label{prices of removable set of type I.2}
\end{enumerate}
\end{lemma}

\begin{proof}
Let $I_r:=\{i\in I_A:x_c\in\legal(i)\}$. Consider a player $\hat{i}\in I_A$, by Definition~\ref{dynamic pricing definition} it suffices to show that any bundle in demand for this player can be extended to a legal allocation in market $A$.

\textbf{Case  $\hat{i}\notin I_r$.}
By construction, no item in $\X_r$ is allocated to $\hat{i}$ by the legal allocation $\bcO$. So, we have $\cO_{\hat{i}}\subseteq\X_B\cap\legal(\hat{i})$, which implies $ k_{\hat{i}}^{\p}=|\cO_{\hat{i}}\cap\X_B|=|\cO_{\hat{i}}|=k_{\hat{i}}$.

Since items in $\X_B$ have strictly lower prices then items in $\X_r\setminus\{x_c\}$ and $x_c$ is not legal to player~$\hat{i}$, a bundle in demand for player $\hat{i}$ is a bundle $\X_{\hat{i}}$ of $k_{\hat{i}}$ items from market $B$. By condition~\ref{prices of removable set of type I.2}, there exists a legal allocation $\bcO^B$ in market $B$ such that $\cO_{\hat{i}}^B=\X_{\hat{i}}$. Thus, this bundle $\X_{\hat{i}}$ extends to a legal allocation $\bcO^A$ in market $A$, where $\cO_i^A=\cO_i^B\cup(\cO_i\cap\X_r)$ for all $i\in I_A$.\\

\textbf{Case $\hat{i}\in I_r$.}
Since $|\cO_{\hat{i}}\cap\X_r|=1$ we have $k_{\hat{i}}^{\p}=k_{\hat{i}}-1$. Since $x_c$ is legal to $\hat{i}$ and has the lowest price among all items, a bundle in demand for player $\hat{i}$ consists of $x_c$ and a bundle $\X_{\hat{i}}$ of size $k_{\hat{i}}^{\p}$ that is in demand in market $B$. Since $(p_x)_{x\in\X_B}$ a dynamic pricing in market $B$, there exists a legal allocation $\bcO^B$ in market $B$ such that $\cO_{\hat{i}}^B=\X_{\hat{i}}$.

If $\hat{i}=i_c$, then we can define $\bcO^A$ in the same way as in the first considered case. If $\hat{i}\neq i_c$, let $\hat{x}$ be the unique item in $\X_r\cap\cO_{\hat{i}}$. Then we obtain a legal allocation $\bcO^A$ in market $A$ such that $\cO_{\hat{i}}^A=\X_{\hat{i}}\cup\{x_c\}$ by defining
\begin{align*}
\cO_i^A=\begin{cases}
\cO_i^B\cup(\cO_i\cap\X_r)\quad&\text{if }i\neq\hat{i}\text{ and }i\neq i_c\\
\cO_i^B\cup\{x_c\}\quad&\text{if }i=\hat{i}\\
\cO_i^B\cup\{\hat{x}\}\quad&\text{if }i=i_c
\end{cases}\,.
\end{align*}
\end{proof}

 Now, we consider a different removable set $\X_r$ in which there is a unique item assigned to every player, according to the allocation $\bcO$. We make sure that the items in $\X_r$ are ``well-connected'', as in~\ref{remove-II-prop.2} in Definition $\ref{remove-II-prop}$, so that when some player takes an extra item in the submarket with items $\X_A\setminus\X_r$, the player who consequently loses an item can be compensated by taking two items from~$\X_r$.

\begin{definition}\label{remove-II-prop}
Let $\X_r\subseteq\X_A$ be a subset of items and let $G_r$ be the subgraph of $G$ induced by $\X_r$. The set $\X_r$ is called a removable set of type II with respect to the legal allocation $\bcO$ if it satisfies the following conditions.
\begin{enumerate}
\item For all $i\in I_A$, $|\cO_i\cap\X_r|=1$.\label{remove-II-prop.1}
\item For all $x,y\in\X_r$, there is a path from $x$ to $y$ in $G_r$.\label{remove-II-prop.2}
\end{enumerate}
\end{definition}

\begin{lemma}\label{prices of removable set of type II}
Let $\X_r\subseteq\X_A$ be a removable set of type II and let $G_r$ be the subgraph of $G$ induced by $\X_r$. Let $B=(\X_B,I_B,(k_i^{\p})_{i\in I_B},\legal)$ be a submarket of $A$ induced by $\bcO$ with $\X_B=\X\setminus\X_r$. Then, any pricing $\vp\in\real^m$ satisfying the following conditions is a dynamic pricing in market $A$.
\begin{enumerate}
\item For all items $x\in\X_r$ and $y\in\X_B$, $p_x>p_y$.\label{prices of removable set of type II.1}
\item $(p_x)_{x\in\X_B}$ is a dynamic pricing in market $B$.\label{prices of removable set of type II.2}
\end{enumerate}
\end{lemma}

\begin{proof}
Consider a player $\hat{i}\in I_A$. From \ref{remove-II-prop.1} in Definition~\ref{remove-II-prop}, we see that $k_{\hat{i}}^{\p}=k_{\hat{i}}-1$ and we can let~$\hat{x}$ be the unique item in $\cO_{\hat{i}}\cap\X_r$.  By condition~\ref{prices of removable set of type II.1}, a bundle in demand for player $\hat{i}$ in $A$ consists of a set $\X_{\hat{i}}$ of size $k_{\hat{i}}^{\p}$ that is in demand for $\hat{i}$ in market $B$ and an item $x^{\p}$ from $\X_A$. By condition~\ref{prices of removable set of type II.2}, there exists a legal allocation $\bcO^B$ in market $B$, such that $\cO_{\hat{i}}^B=\X_{\hat{i}}$. Let us find a legal allocation~$\bcO^A$ in market $A$ such that $\cO_{\hat{i}}^A=B_{\hat{i}}\cup\{x^{\p}\}$. We consider the following three cases.

\textbf{Case $x^{\p}=\hat{x}$.}
The allocation $\bcO^A$, where $\cO_i^A:=\cO_i^B\cup(\cO_i\cap\X_r)$ for every player $i$, is legal in market $A$.

\textbf{Case  $x^{\p}\in\cO_{i^*}\cap\X_r$ for some $i^*$ such that $i^*\neq\hat{i}$.}
By ~\ref{remove-II-prop.2} of Definition~\ref{remove-II-prop}, there is a simple path in the graph $G_r$
\[x^{\p}=x_0\to x_1\to x_2\to\ldots\to x_{k-1}\to x_k=\hat{x}\]
 from $x^{\p}$ to $\hat{x}$ in $G_r$. Let us suppose that $x_j\in\cO_{i_j}\cap\X_r$ for $j=0,\ldots, k$.
Then, we can define the legal allocation $\bcO^A$ in market $A$ as follows
\begin{align*}
\cO_i^A:=\begin{cases}
\cO_i^B\cup(\cO_i\cap\X_r)\quad&\text{if }i\neq i_j \text{ for all }j=0,\ldots,k\\
\cO_i^B\cup\{x_{j+1}\}\quad&\text{if }i=i_j\text{ for some } j=1,\ldots, k-1\\
\cO_i^B\cup\{x^{\p}\}\quad&\text{if }i=i_k=\hat{i}\\
\cO_i^B\cup\{x_1\}\quad&\text{if }i=i_0=i^*\,.
\end{cases}
\end{align*}

\textbf{Case  $x^{\p}\in\X_B$.}
Then $x^{\p}\in\cO_{i^*}^B$ for some $i^*$  such that $i^*\neq\hat{i}$. Let $x^*$ be the unique item in $\cO_{i^*}\cap\X_r$. Then there is a simple path in $G_r$
\[x^*=x_0\to x_1\to x_2\to\ldots\to x_{k-1}\to x_k=\hat{x}\]
from $x^*$ to $\hat{x}$. Let us suppose $x_j\in\cO_{i_j}\cap\X_r$ for $j=0,\ldots, k$.
Then, we can define the legal allocation $\cO_i^A$ in the same way as in the previous case, except that for the player $i^*$ we have $O_{i^*}^A:=(\cO_{i^*}^B\setminus\{x^{\p}\})\cup\{x^*,x_1\}$.
\end{proof}

From Lemmas~\ref{prices of removable set of type I} and~\ref{prices of removable set of type II}, we see that whenever a removable set exists, we can reduce the problem to finding a dynamic pricing in a market with less items. The following theorem shows that it is always possible to find a removable set in a market with four players. We  use the properties of legality graphs as discussed in the previous section. In particular, we construct removable sets from uniquely assigned cycles and we find new legal allocations by reallocating $\bcO$ with respect to some cycle in $G$, if necessary.

\begin{theorem}\label{removable set in market with 4 players}
In a multi-demand market with four players and no item being legal to all of them, there always exists a removable set of either type I or type II in the legality graph for some optimal allocation. 

\end{theorem}

\begin{proof}

Suppose there are four players in the market $A$, i.e. $n=4$. By Lemma~\ref{cycle of items allocated to different players}, there exists a uniquely assigned cycle $\C$ in the graph~$G$. Since there are only four players, $\C$ has length $2$, $3$, or $4$.

Suppose $\C$ is a uniquely assigned cycle of length 4. Then, exactly one item in $\C$ is allocated to each player. Moreover, since $\C$ is a cycle, there is a path between any two items in $\C$. Therefore, the four items in $\C$ form a removable set of type II.

Suppose $\C$ is a uniquely assigned cycle of length 3. Let us write $\C$ as $x_3\to x_2\to x_1\to x_3$. Without loss of generality, let us suppose that item $x_i$ is allocated to player $i$ for $i=1,2,3$. Consider an item $x_4$ such that $x_4\in\cO_4$. By our assumption, every item is legal to more than one player; so $x_4$ is legal to some player other than player $4$. Without loss of generality, let us assume that $x_4$ is legal to player $3$, so the edge $x_3\to x_4$ exists in $G$. If there is an item in $\C$ that is legal to player $4$, then there is an edge from $x_4$ to that item. So, there is a path between any two vertices in $\{x_1,x_2,x_3,x_4\}$. By definition, in such a case $\{x_1,x_2,x_3,x_4\}$ is a removable set of type II. It remains to consider the case where none of the items in $\C$ is legal to player $4$. By Corollary~\ref{every vertex is in a cycle}, there exists an edge from $x_4$ to some other item $x_5$ in $G$. Then, the graph in the top of Figure~\ref{C has length 3} is a subgraph of $G$. Each vertex $x$ in the top of Figure~\ref{C has length 3}  is labeled with a player $i$ and a set $I$ of players, indicating that $x$ is allocated to player $i$ and is legal but not allocated to the players in $I$. Note that $x$ could also be legal to players other than $\{i\}\cup I$. By construction of $G$, $x_5$ is legal but not allocated to player $4$. Let us consider the following three cases regarding the allocation of $x_5$, all three cases are illustrated in Figure~\ref{C has length 3}.

\textbf{Case 1: $x_5\in\cO_1$.}
Let $\bcO^{\p}$ be the reallocation of $\bcO$ with respect to the cycle $x_3\to x_4\to x_5\to x_3$. By Remark~$\ref{reallocation along cycle}$, $\bcO^{\p}$ is a legal allocation in market $A$. So, $\{x_1,x_2,x_4,x_5\}$ is a removable set of type~II in the legality graph associated with $\bcO^{\p}$, see Figure~\ref{C has length 3}.

\textbf{Case 2: $x_5\in\cO_2$.}
In this case $\{x_1,x_3,x_4,x_5\}$ is a removable set of type II.

\textbf{Case 3: $x_5\in\cO_3$.}
If both items $x_4$ and $x_5$ are only legal to players $3$ and $4$, then $\{x_4,x_5\}$ is a removable set of type I. Otherwise, at least one of $x_4$ and $x_5$ is legal to player $1$ or $2$. By possibly reallocating with respect to the cycle $x_4\to x_5\to x_4$, we can assume without loss of generality that $x_5$ is legal to a player other than players $3$ and $4$.

If $x_5$ is legal to player $1$, then we can let $\bcO^{\p}$ be the reallocation of $\bcO$ with respect to the cycle $x_1\to x_5\to x_2\to x_1$. By Remark $\ref{reallocation along cycle}$, $\bcO^{\p}$ is a legal allocation in market $A$. So, $\{x_1,x_2,x_4,x_5\}$ is a removable set of type II in the legality graph associated with $\bcO^{\p}$, see Figure~\ref{C has length 3}.
If $x_5$ is legal to player $2$, then it cannot be legal to player $1$ because we assumed that an item cannot be legal to all players. So, $\{x_2,x_4,x_5\}$ is a removable set of type I with $x_5$ being  the central item.

Finally, suppose $\C$ is a uniquely assigned cycle of length 2. Let us write $\C$ as $x_1\to x_2\to x_1$. Without loss of generality, let us suppose that item $x_i$ is allocated to player $i$ for $i=1,2$. If $x_1$ and $x_2$ are only legal to players $1$ and $2$, then $\{x_1,x_2\}$ is a removable set of type I. Otherwise, at least one of $x_1$ and $x_2$ is legal to player $3$ or $4$. Without loss of generality, assume $x_2$ is legal to player $3$. Consider an item $x_3\in\cO_3$, then the edge $x_3\to x_2$ exists in $G$. So, the graph in the top of Figure~\ref{C has length 2} is a subgraph of $G$. The item $x_3$ is also legal to some player other than $3$. We consider the following three cases regarding the legality of $x_3$, all three cases are illustrated in Figure~\ref{C has length 2}.

\textbf{Case 4: $x_3\in\legal(1)$.}
$x_1\to x_3\to x_2\to x_1$ is a uniquely assigned cycle of length $3$, see Figure~\ref{C has length 2}. So we can use the argumentation from the previous cases.

\textbf{Case 5: $x_3\in\legal(2)$.} Note that
$x_2$ is not legal to player $4$ because we assumed that no item is legal to all players. So, $\{x_1,x_2,x_3\}$ is a removable set of type I with $x_2$ being a central item.

\textbf{Case 6: $x_3\in\legal(4)$.}
Consider $x_4\in\cO_4$, then $x_4$ is legal to some player other than player $4$. If $x_4$ is legal to player $1$ or $2$, then  $\{x_1,x_2,x_3,x_4\}$ is a removable set of type II. If $x_4$ is only legal to players $3$ and $4$, then $\{x_3,x_4\}$ is a removable set of type I with a central item $x_4$.
\end{proof}

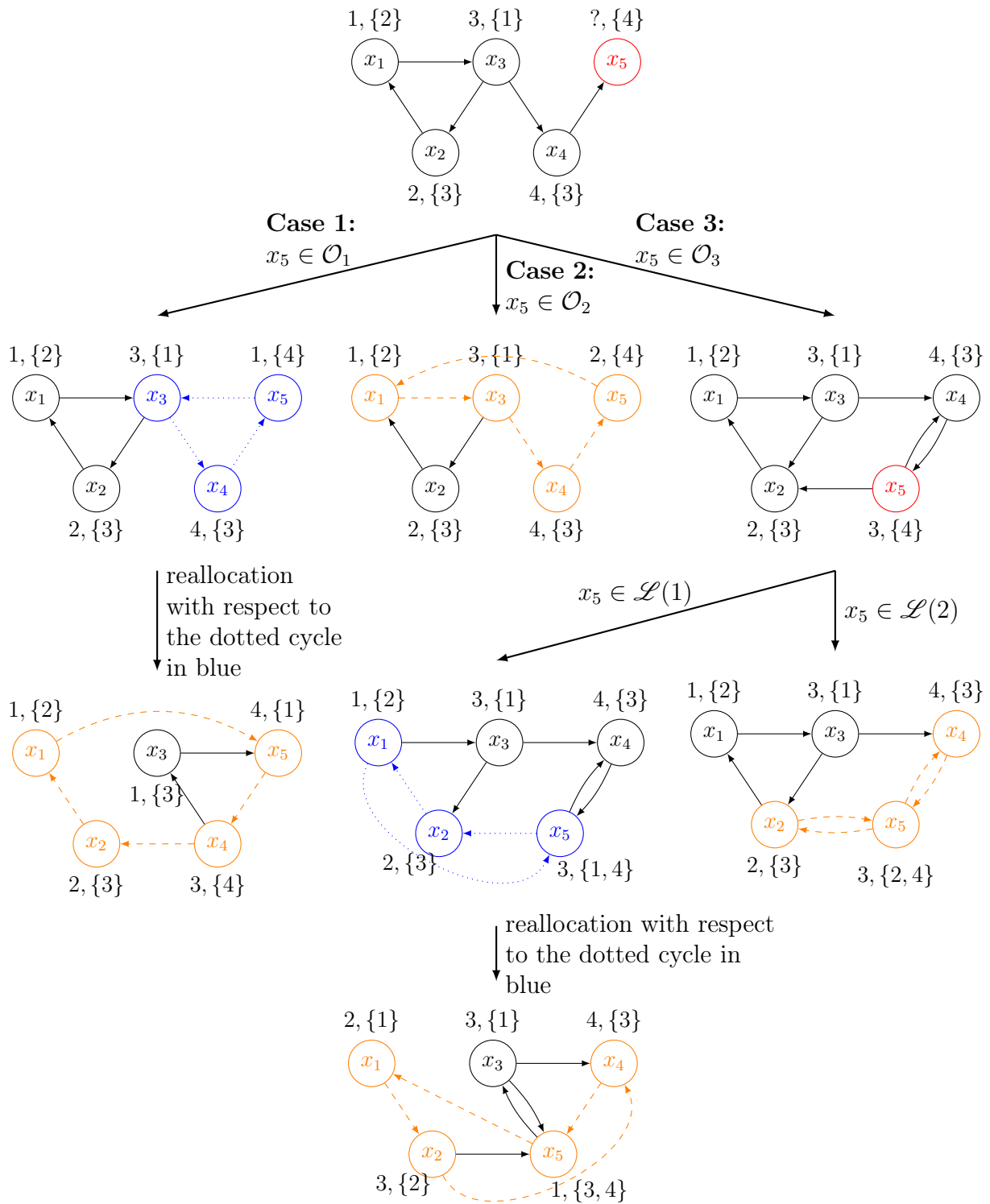
\begin{figure}
    \centering

\begin{subfigure}{\textwidth}
\centering
\begin{tikzpicture}[remember picture, every node/.style={draw,circle,label distance=-4mm,scale=0.9}]

\node[label=above:{$1,\{2\}$}] (a) at (-2,1.5) {$x_1$};
\node[label=below:{$2,\{3\}$}] (b) at (-1,0) {$x_2$};
\node[label=above:{$3,\{1\}$}] (c) at (0,1.5) {$x_3$};
\node[label=below:{$4,\{3\}$}] (d) at (1,0) {$x_4$};
\node[label=above:{$?,\{4\}$},color=red] (e) at (2,1.5) {$x_5$};

\draw[->,>=latex] (a) to (c);
\draw[->,>=latex] (c) to (b);
\draw[->,>=latex] (b) to (a);
\draw[->,>=latex] (c) to (d);
\draw[->,>=latex] (d) to (e);

\path (current bounding box.south west) -- (current bounding box.south east) coordinate[midway] (initial);
\end{tikzpicture}
\end{subfigure}
\bigskip\bigskip\bigskip

\begin{subfigure}{\textwidth}
\centering
\begin{tikzpicture}[remember picture, every node/.style={draw,circle,label distance=-4mm,scale=0.9}]

\node[label=above:{$1,\{2\}$}] (a) at (-2,1.5) {$x_1$};
\node[label=below:{$2,\{3\}$}] (b) at (-1,0) {$x_2$};
\node[color=blue,label=above:{$3,\{1\}$}] (c) at (0,1.5) {$x_3$};
\node[color=blue,label=below:{$4,\{3\}$}] (d) at (1,0) {$x_4$};
\node[color=blue,label=above:{$1,\{4\}$}] (e) at (2,1.5) {$x_5$};

\draw[->,>=latex] (a) to (c);
\draw[->,>=latex] (c) to (b);
\draw[->,>=latex] (b) to (a);
\draw[->,>=latex,dotted,color=blue] (c) to (d);
\draw[->,>=latex,dotted,color=blue] (d) to (e);
\draw[->,>=latex,dotted,color=blue] (e) to (c);

\path (current bounding box.north west) -- (current bounding box.north east) coordinate[midway] (case 1 end);
\path (current bounding box.south west) -- (current bounding box.south east) coordinate[midway] (case 1 start);
\end{tikzpicture}
\hfill
\begin{tikzpicture}[remember picture, every node/.style={draw,circle,label distance=-4mm,scale=0.9}]

\node[color=orange,label=above:{$1,\{2\}$}] (a) at (-2,1.5) {$x_1$};
\node[label=below:{$2,\{3\}$}] (b) at (-1,0) {$x_2$};
\node[color=orange,label=above:{$3,\{1\}$}] (c) at (0,1.5) {$x_3$};
\node[color=orange,label=below:{$4,\{3\}$}] (d) at (1,0) {$x_4$};
\node[color=orange,label=above:{$2,\{4\}$}] (e) at (2,1.5) {$x_5$};

\draw[->,>=latex,dashed,color=orange] (a) to (c);
\draw[->,>=latex] (c) to (b);
\draw[->,>=latex] (b) to (a);
\draw[->,>=latex,dashed,color=orange] (c) to (d);
\draw[->,>=latex,dashed,color=orange] (d) to (e);
\draw[->,>=latex,bend right,dashed,color=orange] (e) to (a);

\path (current bounding box.north west) -- (current bounding box.north east) coordinate[midway] (case 2 end);
\end{tikzpicture}
\hfill
\begin{tikzpicture}[remember picture, every node/.style={draw,circle,label distance=-4mm,scale=0.9}]

\node[label=above:{$1,\{2\}$}] (a) at (-2,1.5) {$x_1$};
\node[label=below:{$2,\{3\}$}] (b) at (-1,0) {$x_2$};
\node[label=above:{$3,\{1\}$}] (c) at (0,1.5) {$x_3$};
\node[color=red,label=below:{$3,\{4\}$}] (d) at (1,0) {$x_5$};
\node[label=above:{$4,\{3\}$}] (e) at (2,1.5) {$x_4$};

\draw[->,>=latex] (a) to (c);
\draw[->,>=latex] (c) to (b);
\draw[->,>=latex] (b) to (a);
\draw[->,>=latex] (c) to (e);
\draw[->,>=latex] (d) to (b);
\draw[->,>=latex,bend left=10] (e) to (d);
\draw[->,>=latex,bend left=10] (d) to (e);

\path (current bounding box.north west) -- (current bounding box.north east) coordinate[midway] (case 3 end);
\path (current bounding box.south west) -- (current bounding box.south east) coordinate[midway] (case 3 start);
\end{tikzpicture}
\end{subfigure}
\bigskip\bigskip\bigskip

\begin{subfigure}{\textwidth}
\centering
\begin{tikzpicture}[remember picture, every node/.style={draw,circle,label distance=-4mm,scale=0.9}]

\node[color=orange,label=above:{$1,\{2\}$}] (a) at  (-2,1.5) {$x_1$};
\node[color=orange,label=below:{$2,\{3\}$}] (b) at  (-1,0) {$x_2$};
\node[label=below:{$1,\{3\}$}] (c) at  (0,1.5) {$x_3$};
\node[color=orange,label=below:{$3,\{4\}$}] (d) at  (1,0) {$x_4$};
\node[color=orange,label=above:{$4,\{1\}$}] (e) at  (2,1.5) {$x_5$};

\draw[->,>=latex,bend left,dashed,color=orange] (a) to (e);
\draw[->,>=latex,dashed,color=orange] (b) to (a);
\draw[->,>=latex,dashed,color=orange] (d) to (b);
\draw[->,>=latex] (c) to (e);
\draw[->,>=latex,dashed,color=orange] (e) to (d);
\draw[->,>=latex] (d) to (c);

\path (current bounding box.north west) -- (current bounding box.north east) coordinate[midway] (case 1 reallocated);
\end{tikzpicture}
\hfill
\begin{tikzpicture}[remember picture, every node/.style={draw,circle,label distance=-4mm,scale=0.9}]

\node[color=blue,label=above:{$1,\{2\}$}] (a) at (-2,1.5) {$x_1$};
\node[color=blue,label=265:{$2,\{3\}$}] (b) at (-1,0) {$x_2$};
\node[label=above:{$3,\{1\}$}] (c) at (0,1.5) {$x_3$};
\node[color=blue,label=275:{$3,\{1,4\}$}] (d) at (1,0) {$x_5$};
\node[label=above:{$4,\{3\}$}] (e) at (2,1.5) {$x_4$};

\draw[->,>=latex] (a) to (c);
\draw[->,>=latex] (c) to (b);
\draw[->,>=latex,dotted,color=blue] (b) to (a);
\draw[->,>=latex] (c) to (e);
\draw[->,>=latex,dotted,color=blue] (d) to (b);
\draw[->,>=latex,bend left=10] (e) to (d);
\draw[->,>=latex,bend left=10] (d) to (e);
\draw[->,>=latex,bend right=90,dotted,color=blue] (a) to (d);

\path (current bounding box.north west) -- (current bounding box.north east) coordinate[midway] (case 3 subcase 1 end);
\path (current bounding box.south west) -- (current bounding box.south east) coordinate[midway] (case 3 subcase 1 start);
\end{tikzpicture}
\hfill
\begin{tikzpicture}[remember picture, every node/.style={draw,circle,label distance=-4mm,scale=0.9}]

\node[label=above:{$1,\{2\}$}] (a) at (-2,1.5) {$x_1$};
\node[color=orange,label=below:{$2,\{3\}$}] (b) at (-1,0) {$x_2$};
\node[label=above:{$3,\{1\}$}] (c) at (0,1.5) {$x_3$};
\node[color=orange,label=below:{$3,\{2,4\}$}] (d) at (1,0) {$x_5$};
\node[color=orange,label=above:{$4,\{3\}$}] (e) at (2,1.5) {$x_4$};

\draw[->,>=latex] (a) to (c);
\draw[->,>=latex] (c) to (b);
\draw[->,>=latex] (b) to (a);
\draw[->,>=latex] (c) to (e);
\draw[->,>=latex,dashed,color=orange,bend left=10] (d) to (b);
\draw[->,>=latex,dashed,color=orange,bend left=10] (b) to (d);
\draw[->,>=latex,dashed,color=orange,bend left=10] (e) to (d);
\draw[->,>=latex,dashed,color=orange,bend left=10] (d) to (e);

\path (current bounding box.north west) -- (current bounding box.north east) coordinate[midway] (case 3 subcase 2 end);
\end{tikzpicture}
\end{subfigure}
\bigskip\bigskip

\begin{subfigure}{\textwidth}
\centering
\begin{tikzpicture}[remember picture, every node/.style={draw,circle,label distance=-4mm,scale=0.9}]

\node[color=orange,label=above:{$2,\{1\}$}] (a) at (-2,1.5) {$x_1$};
\node[color=orange,label=265:{$3,\{2\}$}] (b) at (-1,0) {$x_2$};
\node[label=above:{$3,\{1\}$}] (c) at (0,1.5) {$x_3$};
\node[color=orange,label=275:{$1,\{3,4\}$}] (d) at (1,0) {$x_5$};
\node[color=orange,label=above:{$4,\{3\}$}] (e) at (2,1.5) {$x_4$};

\draw[->,>=latex,bend left=10] (c) to (d);
\draw[->,>=latex,dashed,color=orange,bend right=90] (b) to (e);
\draw[->,>=latex,dashed,color=orange] (a) to (b);
\draw[->,>=latex] (c) to (e);
\draw[->,>=latex] (b) to (d);
\draw[->,>=latex,dashed,color=orange] (e) to (d);
\draw[->,>=latex,bend left=10] (d) to (c);
\draw[->,>=latex,dashed,color=orange] (d) to (a);

\path (current bounding box.north west) -- (current bounding box.north east) coordinate[midway] (case 3 subcase 1 reallocated);
\end{tikzpicture}
\end{subfigure}

\tikz[overlay,remember picture]{\draw[-latex,thick] (initial) -- (case 1 end) node[midway,above,text width=2cm]{\textbf{Case 1:} $x_5\in\cO_1$};}
\tikz[overlay,remember picture]{\draw[-latex,thick] (initial) -- (case 2 end) node[midway,right,text width=2cm,yshift=-2mm]{\textbf{Case 2:} $x_5\in\cO_2$};}
\tikz[overlay,remember picture]{\draw[-latex,thick] (initial) -- (case 3 end) node[midway,above,text width=2cm,xshift=5mm]{\textbf{Case 3:} $x_5\in\cO_3$};}
\tikz[overlay,remember picture]{\draw[-latex,thick] (case 1 start) -- (case 1 reallocated) node[midway,right,text width=3cm]{reallocation with respect to the dotted cycle in blue};}
\tikz[overlay,remember picture]{\draw[-latex,thick] (case 3 start) -- (case 3 subcase 1 end) node[midway,above,xshift=-5mm]{$x_5\in\legal(1)$};}
\tikz[overlay,remember picture]{\draw[-latex,thick] (case 3 start) -- (case 3 subcase 2 end) node[midway,right]{$x_5\in\legal(2)$};}
\tikz[overlay,remember picture]{\draw[-latex,thick] (case 3 subcase 1 start) -- (case 3 subcase 1 reallocated) node[midway,right,text width=4.5cm]{reallocation with respect to the dotted cycle in blue};}

\caption{The graph at the top level is a subgraph of $G$, containing a uniquely assigned cycle $\C$ of length 3 and two items $x_4$ and $x_5$, as defined in Theorem~\ref{removable set in market with 4 players}. We consider 3 different cases depending on whether item $x_5$ is allocated to player 1, 2, or 3. After reallocating items with respect to the dotted cycle, if necessary, we always obtain some removable set, which are connected by the dashed edges.}
\label{C has length 3}
\end{figure}

\begin{figure}

\begin{subfigure}{\textwidth}
\centering
\begin{tikzpicture}[remember picture, every node/.style={draw,circle,label distance=-4mm,scale=0.9}]

\node[label=below:{$1,\{2\}$}] (a) at (-2,0) {$x_1$};
\node[label=below:{$2,\{1,3\}$}] (b) at (0,0) {$x_2$};
\node[label=below:{$3,\{?\}$},color=red] (c) at (2,0) {$x_3$};

\draw[->,>=latex,bend left=10] (a) to (b);
\draw[->,>=latex,bend left=10] (b) to (a);
\draw[->,>=latex] (c) to (b);

\path (current bounding box.south west) -- (current bounding box.south east) coordinate[midway] (initial);
\end{tikzpicture}
\end{subfigure}
\bigskip\bigskip\bigskip

\begin{subfigure}{\textwidth}
\centering
\begin{tikzpicture}[remember picture, every node/.style={draw,circle,label distance=-4mm,scale=0.9}]

\node[label=below:{$1,\{2\}$},color=black!30!green] (a) at (-2,0) {$x_1$};
\node[label=below:{$2,\{1,3\}$},color=black!30!green] (b) at (0,0) {$x_2$};
\node[label=below:{$3,\{1\}$},color=black!30!green] (c) at (2,0) {$x_3$};

\draw[->,>=latex,bend left=10] (a) to (b);
\draw[->,>=latex,bend left=10,dotted,color=black!30!green] (b) to (a);
\draw[->,>=latex,dotted,color=black!30!green] (c) to (b);
\draw[->,>=latex,bend left,dotted,color=black!30!green] (a) to (c);

\path (current bounding box.north west) -- (current bounding box.north east) coordinate[midway] (case 4 end);
\end{tikzpicture}
\hfill
\begin{tikzpicture}[remember picture, every node/.style={draw,circle,label distance=-4mm,scale=0.9}]

\node[label=below:{$1,\{2\}$}] (a) at (-2,0) {$x_1$};
\node[label=below:{$2,\{1,3\}$}] (b) at (0,0) {$x_2$};
\node[label=below:{$3,\{4\}$}] (c) at (2,0) {$x_3$};
\node[label=left:{$4,\{?\}\quad$},color=red] (d) at (1,1.5) {$x_4$};

\draw[->,>=latex,bend left=10] (a) to (b);
\draw[->,>=latex,bend left=10] (b) to (a);
\draw[->,>=latex] (c) to (b);
\draw[->,>=latex] (d) to (c);

\path (current bounding box.north west) -- (current bounding box.north east) coordinate[midway] (case 6 end);
\path (current bounding box.south west) -- (current bounding box.south east) coordinate[midway] (case 6 start);
\end{tikzpicture}
\hfill
\begin{tikzpicture}[remember picture, every node/.style={draw,circle,label distance=-4mm,scale=0.9}]

\node[label=below:{$1,\{2\}$},color=orange] (a) at (-2,0) {$x_1$};
\node[label=below:{$2,\{1,3\}$},color=orange] (b) at (0,0) {$x_2$};
\node[label=below:{$3,\{2\}$},color=orange] (c) at (2,0) {$x_3$};

\draw[->,>=latex,bend left=10,dashed,color=orange] (a) to (b);
\draw[->,>=latex,bend left=10,dashed,color=orange] (b) to (a);
\draw[->,>=latex,bend left=10,dashed,color=orange] (c) to (b);
\draw[->,>=latex,bend left=10,dashed,color=orange] (b) to (c);

\path (current bounding box.north west) -- (current bounding box.north east) coordinate[midway] (case 5 end);
\end{tikzpicture}
\end{subfigure}
\bigskip\bigskip\bigskip

\begin{subfigure}{\textwidth}
\centering
\begin{tikzpicture}[remember picture, every node/.style={draw,circle,label distance=-4mm,scale=0.9}]

\node[label=below:{$1,\{2\}$},color=orange] (a) at (-2,0) {$x_1$};
\node[label=below:{$2,\{1,3\}$},color=orange] (b) at (0,0) {$x_2$};
\node[label=below:{$3,\{4\}$},color=orange] (c) at (2,0) {$x_3$};
\node[label=left:{$4,\{1\}\quad$},color=orange] (d) at (1,1.5) {$x_4$};

\draw[->,>=latex,bend left=10] (a) to (b);
\draw[->,>=latex,bend left=10,dashed,color=orange] (b) to (a);
\draw[->,>=latex,dashed,color=orange] (c) to (b);
\draw[->,>=latex,dashed,color=orange] (d) to (c);
\draw[->,>=latex,dashed,color=orange] (a) to (d);

\path (current bounding box.north west) -- (current bounding box.north east) coordinate[midway] (x_4 legal to 1);
\path (current bounding box.south west) -- (current bounding box.south east) coordinate[midway] (x_4 legal to 1 bottom);
\end{tikzpicture}
\hfill
\begin{tikzpicture}[remember picture, every node/.style={draw,circle,label distance=-4mm,scale=0.9}]

\node[label=below:{$1,\{2\}$},color=orange] (a) at (-2,0) {$x_1$};
\node[label=below:{$2,\{1,3\}$},color=orange] (b) at (0,0) {$x_2$};
\node[label=below:{$3,\{4\}$},color=orange] (c) at (2,0) {$x_3$};
\node[label=left:{$4,\{2\}\quad$},color=orange] (d) at (1,1.5) {$x_4$};

\draw[->,>=latex,bend left=10,dashed,color=orange] (a) to (b);
\draw[->,>=latex,bend left=10,dashed,color=orange] (b) to (a);
\draw[->,>=latex,dashed,color=orange] (c) to (b);
\draw[->,>=latex,dashed,color=orange] (d) to (c);
\draw[->,>=latex,dashed,color=orange] (b) to (d);

\path (current bounding box.north west) -- (current bounding box.north east) coordinate[midway] (x_4 legal to 2);
\path (current bounding box.south west) -- (current bounding box.south east) coordinate[midway] (x_4 legal to 2 bottom);
\end{tikzpicture}
\hfill
\begin{tikzpicture}[remember picture, every node/.style={draw,circle,label distance=-4mm,scale=0.9}]

\node[label=below:{$1,\{2\}$}] (a) at (-2,0) {$x_1$};
\node[label=below:{$2,\{1,3\}$}] (b) at (0,0) {$x_2$};
\node[label=below:{$3,\{4\}$},color=orange] (c) at (2,0) {$x_3$};
\node[label=left:{$4,\{3\}\quad$},color=orange] (d) at (1,1.5) {$x_4$};

\draw[->,>=latex,bend left=10] (a) to (b);
\draw[->,>=latex,bend left=10] (b) to (a);
\draw[->,>=latex] (c) to (b);
\draw[->,>=latex,bend left=10,dashed,color=orange] (d) to (c);
\draw[->,>=latex,bend left=10,dashed,color=orange] (c) to (d);

\path (current bounding box.north west) -- (current bounding box.north east) coordinate[midway] (x_4 legal to 3 end);
\path (current bounding box.south west) -- (current bounding box.south east) coordinate[midway] (x_4 legal to 3 start);
\end{tikzpicture}
\end{subfigure}
\bigskip\bigskip\bigskip

\begin{tikzpicture}[overlay,remember picture]
\node[text width=2cm,xshift=1cm,yshift=2cm] at (case 4 end) {\textbf{Case 4:} $x_3\in\legal(1)$};
\draw[-latex,thick,bend right=20] (initial) to (case 4 end);
\draw[-latex,thick] (initial) -- (case 6 end) node[midway,right,text width=2cm,yshift=-2mm]{\textbf{Case 6:} $x_3\in\legal(4)$};
\node[text width=2cm,xshift=-1cm,yshift=2.4cm] at (case 5 end) {\textbf{Case 5:} $x_3\in\legal(2)$};
\draw[-latex,thick,bend left=20] (initial) to (case 5 end);

\draw[-latex,thick] (case 6 start) -- (x_4 legal to 1) node[midway,above,xshift=-5mm]{$x_4\in\legal(1)$};
\draw[-latex,thick] (case 6 start) -- (x_4 legal to 2) node[midway,right,yshift=-2mm]{$x_4\in\legal(2)$};
\draw[-latex,thick] (case 6 start) -- (x_4 legal to 3 end) node[midway,above,xshift=1.5cm]{$x_4\in\legal(3)\setminus(\legal(1)\cup\legal(2))$};
\end{tikzpicture}

\caption{The graph at the top level is a subgraph of $G$, containing a uniquely assigned cycle $\C$ of length 2, and another item $x_3$, as defined in Theorem~\ref{removable set in market with 4 players}. We consider 3 different cases depending on whether item $x_3$ is legal but not assigned to player 1, 2, or 4. We either obtain a cycle of length 3, which is marked with dotted edges, or some removable set, which is connected by dashed edges.}
\label{C has length 2}
\end{figure}
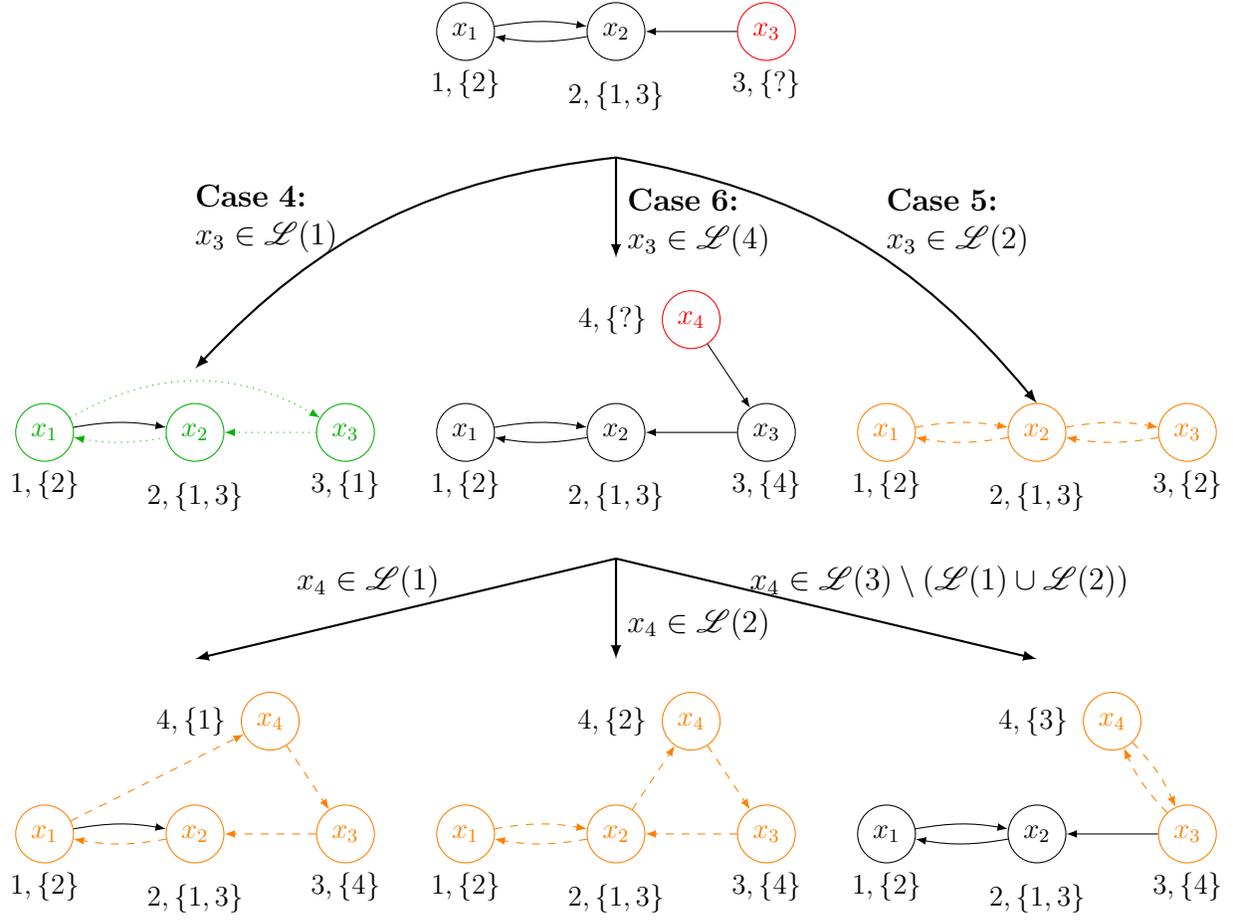

Given a multi-demand market with a set $\X$ of $m$ items and four players $I$, auppose each player $i\in I$ is $k_i$-demand and has multi-demand valuation function $v_i:2^{\X}\to\real_{\ge0}^m$. We summarize the procedure for finding a dynamic pricing in a multi-demand market with four players in Algorithm~\ref{algorithm for 4 players}. The correctness of the algorithm is proved in the following theorem.\\

\begin{theorem}\label{correctness of algorithm for 4 players}
A dynamic pricing exists in every multi-demand market with four players and it can be computed using Algorithm~\ref{algorithm for 4 players}.
\end{theorem}

\begin{proof}
We proceed by induction on the number of items $m$ in the market. The base cases are markets with at most three items, which implies that there are at most three players. Such markets have been studied in~\cite{Berger2020}. From the induction hypothesis, we may assume that a dynamic pricing can be found in any multi-demand market with less than $m$ items and at most four players.

We first compute the rough prices as in Algorithm~\ref{rough price} and let market $A:=(\X^2,I_2,(\hat{k}_i)_{i\in I_2},\legal)$ be the simplified induced market described in Section~\ref{Fine Price}. We now show that the price vector $\vp^F$ defined in Algorithm~\ref{algorithm for 4 players} is a dynamic pricing in market $A$.

If there exists an item $\hat{x}$ that is legal to all players, which corresponds to lines~\ref{lst:line:4 player alg. case 1 item legal to all players}-\ref{lst:line:4 player alg. case 1 item legal to all players p^F def} in Algorithm~\ref{algorithm for 4 players}, we can follow the ordering of any dynamic pricing in the submarket with $\hat{x}$ removed and assign $\hat{x}$ a ``high'' price as in line~\ref{lst:line:4 player alg. case 1 item legal to all players p^F def}, to obtain prices $\vp^F$. Then, Lemma~\ref{item legal to all players} and the induction hypothesis guarantee that $\vp^F$ is a dynamic pricing in market $A$.

If such $\hat{x}$ does not exist, which corresponds to lines~\ref{lst:line:4 player alg. case removable set I or II}-\ref{lst:line:4 player alg. case removable set II p^F def} of Algorithm~\ref{algorithm for 4 players}, the existence of a removable set is guaranteed by Theorem~\ref{removable set in market with 4 players}. So market $A$ can always be reduced to a submarket $B$, whose dynamic pricing $\vp^B$ exists by induction hypothesis. With the valuation functions defined as in line~\ref{lst:line:4 player alg. valuation def} of Algorithm~\ref{algorithm for 4 players}, the dynamic pricing of market $B$ must be in the range $\mathbf{0}<\vp^B<\Delta/2$ in order to yield positive utility for all items that are legal to a given player. So depending on whether the removable set is of type I or II, the definition of prices $\vp^F$ is given by either line~\ref{lst:line:4 player alg. case removable set I p^F def}, which agrees with Lemma~\ref{prices of removable set of type I}, or line~\ref{lst:line:4 player alg. case removable set II p^F def}, which agrees with Lemma~\ref{prices of removable set of type II}. In both cases, $\vp^F$ is a dynamic pricing in market $A$.

Finally, since $\mathbf{0}<\vp^F\le\Delta/2<\Delta$, by Definition~\ref{definition fine prices} of fine prices, $\vp^F$ is a vector of fine prices of the original market. From Section~\ref{Fine Price}, we know that the pricing $\vp^D$ defined in line~\ref{lst:line:4 player alg. p^D def} is the desired dynamic pricing.
\end{proof}

\begin{algorithm}
\caption{Finding dynamic pricing in multi-demand markets with 4 players.}\label{algorithm for 4 players}
\begin{algorithmic}[1]
\State Define $\vp^R$ to be the rough prices of the original market.
\State Define $\Delta>0$ as in Section $\ref{Fine Price}$.
\State Define $\X^2:=\X\setminus\left(\bigcup\limits_{i\in I}\R_i\right)$, $I_2:=\{i\in I:\R_i\neq\K_i\}$ and $\hat{k}_i:=k_i-|\R_i|$ for all $i\in I_2$.
\State Define market $A:=(\X^2,I_2,(\hat{k}_i)_{i\in I_2},\legal)$.
\State Define
\[\hat{v}_i(x):=\begin{cases}
\Delta/2\quad&\text{if $x$ is legal to $i$}\\
0\quad&\text{if $x$ is not legal to $i$}
\end{cases}\]
for all $i\in I_2$ and $x\in\X^2$.\label{lst:line:4 player alg. valuation def}
\If{there exists an item $\hat{x}\in\X^2$ that is legal to all players in $I_2$}\label{lst:line:4 player alg. case 1 item legal to all players}
    \State Define $\vp^B$ to be a dynamic pricing in the submarket of $A$  induced by  a legal allocation $\bcO$ with items $\X^2\setminus\{\hat{x}\}$.
    \State Define the fine prices $\vp^F$ by $p_x^F:=\begin{cases}
    p_x^B\quad&\text{if }x\in\X^2\setminus\{\hat{x}\}\\
    \Delta/2\quad&\text{if }x=\hat{x}
    \end{cases}$.\label{lst:line:4 player alg. case 1 item legal to all players p^F def}
\Else\label{lst:line:4 player alg. case removable set I or II}
    \State Find a removable set $\X_r$ of type I or II with respect to a legal allocation $\bcO$ in market $A$.
    \State Define $\vp^B$ to be a dynamic pricing in the submarket of $A$ induced by $\bcO$ with items $\X^2\setminus\X_r$.
    \If{$\X_r$ is a removable set of type I with central item $x_c$}
        \State Define $\hat{p}:=\min\{p_x^B:x\in\X^2\setminus\X_r\}$.
        \State Define the fine prices $\vp^F$ by $p_x^F:=\begin{cases}
        \hat{p}/2\quad&\text{if }x=x_c\\
        p_x^B\quad&\text{if }x\in\X^2\setminus\X_r\\
        \Delta/2\quad&\text{if }x\in\X_r\setminus\{x_c\}
        \end{cases}$.\label{lst:line:4 player alg. case removable set I p^F def}
    \Else
        \State Define the fine prices $\vp^F$ by $p_x^F:=\begin{cases}
        p_x^B\quad&\text{if }x\in\X^2\setminus\X_r\\
        \Delta/2\quad&\text{if }x\in\X_r
        \end{cases}$.\label{lst:line:4 player alg. case removable set II p^F def}
    \EndIf
\EndIf
\State Define the dynamic pricing $\vp^D$ by $p_x^D:=\begin{cases}
p_x^R+p_x^F\quad&\text{if }x\in\X^2\\
p_x^R\quad&\text{if }x\in\X\setminus\X^2
\end{cases}$.\label{lst:line:4 player alg. p^D def}
\State \Return $\vp^D$
\end{algorithmic}
\end{algorithm}

\section{Multi-Demand Market with Two Optimal Allocations}\label{Multi-Demand Market with Two Optimal Allocations}

The legality graph defined in Section~\ref{Legality Graph} can also be used for finding dynamic pricing in multi-demand markets with at most two optimal allocations. We consider a simplified market $A=(\X_A,I_A,(k_i)_{i\in I_A},\legal)$, legal allocation $\bcO$ and legality graph $G$ as defined in Section~\ref{Legality Graph}. In particular, every item is legal to at least two players. 

So in this section, let us make the assumption that market $A$ has at most two optimal allocations. This implies that  every item is legal to at most two players. In the following lemma we show that this also implies that, in a uniquely assigned cycle, every player has either two or no legal items. 

\begin{lemma}\label{0 or 2 legal items in UA cycle}
Given a uniquely assigned cycle $\C$ in $G$, every player $i$, $i\in I_A$ has exactly $0$ or $2$ legal items in $\C$.
\end{lemma}

\begin{proof}
Consider a uniquely assigned cycle $\C$ in $G$ and a player $i\in I_A$. 

Suppose there are at least three different items $x$, $y$ and $z$ in the cycle $\C$ that are legal to $i$. Without loss of generality, $x$ and $y$ are either both in $\cO_i$ or both not in $\cO_i$. If $x$ and $y$ are both in $\cO_i$, then we get a contradiction to the fact that $\C$ is a uniquely assigned cycle. Now suppose both $x$ and $y$ are not in $\cO_i$, then there are two edges $x^\p\to x$ and $y^\p\to y$ in the cycle $\C$. Since $x$ and $y$ are each legal to at most two players, the only player that they are legal but not allocated to is the player $i$. So both $x^\p$ and $y^\p$ are in $\cO_i$. We again get a contradiction to the fact that $\C$ is a uniquely assigned cycle. So, there are at most 2 items in $\C$ that are legal to player $i$.

Suppose there is a unique item $x$ in $\C$ that is legal to $i$. Then there exist two edges $x\to y$ and $z\to x$ in the cycle $\C$. Using the fact that every item is legal to at most two players, we have that $y$ is legal to $i$ if $x$ is in  $\cO_i$. Similarly,  we have that $z$ is legal to $i$ if $x$ is not in  $\cO_i$, leading to a contradiction.
\end{proof}

As in Section~\ref{Multi-Demand Market with Four Players}, we use induction on the number of items in the market. In this section, we define two new types of removable sets. Both are closely related to uniquely assigned cycles in $G$. In particular, we study how dynamic pricing can be done based on the parity of the length of such cycles. We first consider a uniquely assigned cycle of even length.

\begin{definition}
Suppose $\C$ is a uniquely assigned cycle in the legality graph $G$ of even length. Then, the items in $\C$ form a removable set of type III and $\C$ is called a removable cycle.
\end{definition}

Intuitively, we use a removable set $\X_r$ of type III formed by $\C$ in the following way. If no item is allocated to a player $i$, $i\in I_A$, then the choice of player $i$ is determined solely by the dynamic pricing in the submarket of $A$ with items $\X_A\setminus\X_r$. On the other hand, if a player $i$ has exactly two legal items in $\X_r$, then these two items must be adjacent in the cycle. To ensure that an appropriate compensation is possible, we alternately assign ``high'' and ``low'' prices to items in the cycle $\C$ so that player $i$ selects exactly one item from $\C$. This intuition leads to the following lemma.

\begin{lemma}\label{prices of removable set of type III}
Suppose $\X_r\subseteq\X_A$ is a removable set of type III formed by a cycle $\C$
\[x_0\to x_1\to x_2\to\ldots\to x_{l-1}\to x_l=x_0\,,\] 
i.e. $\X_r=\{x_j\,:\,j=1,\ldots, l\}$. Let $B=(\X_B,I_B,(k_i^{\p})_{i\in I_B},\legal)$ be the submarket of $A$ induced by $\bcO$ with $\X_B=\X\setminus\X_r$. Any pricing $\vp\in\real^m$ satisfying the following conditions is a dynamic pricing in market $A$.
\begin{enumerate}
\item For all $y_1\in\{x_j\,:\, j=1,\ldots, l,\,j\text{ is odd}\}$, $y_2\in\X_B$, and $y_3\in\{x_j\,:\,j=1,\ldots, l,\,j\text{ is even}\}$, \[p_{y_1}>p_{y_2}>p_{y_3}\,.\]
\item $(p_x)_{x\in\X_B}$ is a dynamic pricing in market $B$.
\end{enumerate}
\end{lemma}

\begin{proof}
Suppose $x_j\in\cO_{i_j}$ for $ j=0,\ldots, l$. Let us define $I_r:=\{i_j\,:\,j=0,\ldots, l\}\subseteq I_A$. For all $j=1,\ldots, l$, by construction of the legality graph $G$, we have that the item $x_j$ is legal to players $i_j$ and $i_{j-1}$. Since every item is legal to at most two players the item $x_j$ is only legal to players $i_j$ and $i_{j-1}$. By definition of uniquely assigned cycles, all $l$ players  in $I_r$ are distinct. By Lemma~\ref{0 or 2 legal items in UA cycle}, items $x_j$ and $x_{j+1}$ are the only two items in $\X_r$ that are legal to player $i_j$. Note that we identify $i_{l+1}$ with $i_1$ and $x_{l+1}$ with $x_1$ in this proof.

Note that
\[|\cO_{i}\cap\X_r|=\begin{cases}
1\quad&\text{if }i\in I_r\\
0\quad&\text{if }i\notin I_r
\end{cases}
\]
implies
\[
k_i^{\p}=\begin{cases}
k_i-1\quad&\text{if }i\in I_r\\
k_i\quad&\text{if }i\notin I_r
\end{cases}\,.\]

Consider a player $\hat{i}\in I_A$ and let us show that any bundle in demand for player $\hat{i}$ with respect to the prices $\vp$ can be extended to a legal allocation in market $A$. We consider the following three cases.

\textbf{Case $\hat{i}=i_{\hat{j}}\in I_r$, for some $\hat{j}=1,\ldots, l$ and $\hat{j}$ is even.}
In this case, $k_{\hat{i}}^{\p}=k_{\hat{i}}-1$ and the only two items that are legal to $\hat{i}$ in $\X_r$  are $x_{\hat{j}}$ and $x_{\hat{j}+1}$. Since $\hat{j}$ is even, for all items $x\in\X_B$, we have
\[
p_{x_{\hat{j}+1}}>p_x>p_{x_{\hat{j}}}.
\]
So for player $\hat{i}$ the bundle in demand in the market $A$ consists of the item $x_{\hat{j}}$ and a bundle $\X_{\hat{i}}$ of size $k_{\hat{i}}^{\p}$ that is in demand in market $B$. Since $(p_x)_{x\in\X_B}$ is a dynamic pricing in market $B$, there exists a legal allocation $\bcO^B$ in market $B$, such that $\cO_{\hat{i}}^B=\X_{\hat{i}}$. We then obtain a legal allocation $\bcO^A$ in market $A$ such that $\cO_{\hat{i}}^A=\X_{\hat{i}}\cup \{x_{\hat{j}}\}$ by defining
\[
\cO_i^A:=\begin{cases}
\cO_i^B\cup\{x_j\}\quad&\text{if }i\in I_r\text{ and }i=i_j\text{ for some }j=1,\ldots,l\\
\cO_i^B\quad&\text{if }i\notin I_r\,.
\end{cases}
\]

\textbf{Case  $\hat{i}=i_{\hat{j}}\in I_r$, for some $\hat{j}=1,\ldots, l$ and $\hat{j}$ is odd.}
This case is analogous to the previous case, except that for all items $x\in\X_B$, we have
\[
p_{x_{\hat{j}}}>p_x>p_{x_{\hat{j}+1}}.
\]
Accordingly, we need to define a legal allocation $\bcO^A$ in market $A$ as
\[
\cO_i^A:=\begin{cases}
\cO_i^B\cup\{x_{j+1}\}\quad&\text{if }i\in I_r\text{ and }i=i_j\text{ for some }j=1,\ldots,l\\
\cO_i^B\quad&\text{if }i\notin I_r\,.
\end{cases}
\]

\textbf{Case  $\hat{i}\notin I_r$.}
In this case, no item in $\X_r$ is legal to $\hat{i}$. So for player $\hat{i}$, the bundle in demand is a set $\X_{\hat{i}}$ of $k_{\hat{i}}^{\p}=k_{\hat{i}}$ items, where $\X_{\hat{i}}$ is in demand in market $B$. Since $(p_x)_{x\in\X_B}$ is a dynamic pricing in market $B$, there exists a legal allocation $\bcO^B$ in market $B$ such that $\cO_{\hat{i}}^B=\X_{\hat{i}}$. We can extend it to a legal allocation in market $A$ in the same way as in either of the previous cases.
\end{proof}

This strategy of alternately assigning ``high'' and ``low'' prices in a uniquely assigned cycle  is possible only when the cycle has an even length. Therefore, it would be desirable to have such cycles in $G$ of even length. It turns out that even when such cycles do not exist, as long as we can find two appropriate intersecting odd cycles, it is possible to construct a removable set of type III by reallocating $\bcO$ with respect to one of the two odd cycles. 

\begin{definition}\label{odd cycle pair definition}
Suppose there are two odd length cycles $\C_1$
\[x_0\to x_1\to x_2\to\ldots\to x_{l_1-1}\to x_{l_1}=x_0\] and $\C_2$
\[y_0\to y_1\to y_2\to\ldots\to y_{l_2-1}\to y_{l_2}=y_0\] in the legality graph $G$, with $x_s\in\cO_{i_{x,s}}$ for all $s=1,\ldots, l_1$ and $y_t\in\cO_{i_{y,t}}$ for all $ t=1,\dots,l_2$. The pair $(\C_1,\C_2)$ is called an odd cycle pair if $\C_1$ and $\C_2$ satisfy the following conditions
\begin{enumerate}
\item$x_{l_1-j}=y_{l_2-j}$ for all $j=0,\ldots, r$, for some $r=0,\ldots,\min\{l_1,l_2\}$.\label{odd cycle pair definition.1}
\item $\{i_{x,s}\,:\,s=1,\ldots, l_1\}\cap\{i_{y,t}\,:\,t=1,\ldots, l_2\}$ has exactly $l_1+l_2-(r+2)$ players.\label{odd cycle pair definition.2}
\end{enumerate}

\end{definition}

\begin{lemma}\label{odd cycle pair implies removable cycle}
If there is an odd cycle pair in $G$, then there exists a removable set of type III in market $A$ when the legality graph is constructed with respect to another optimal allocation.
\end{lemma}

\begin{proof}
Let $(\C_1,\C_2)$ be an odd cycle pair. Let us use the same notation as in Definition~\ref{odd cycle pair definition}. From~\ref{odd cycle pair definition.1} of Definition~\ref{odd cycle pair definition}, we know that $i_{x,l_1-j}=i_{y,l_2-j}$ for $j=0,\ldots,r$. Moreover, since $x_{l_1-r}=y_{l_2-r}$, we know that $x_{l_1-r-1}$ and $y_{l_2-r-1}$ must be allocated to the same player. So, we also have $i_{x,l_1-r-1}=i_{y,l_2-r-1}$. By~\ref{odd cycle pair definition.2} of Definition~\ref{odd cycle pair definition}, these are exactly the $r+2$ players that are common in $\C_1$ and $\C_2$. All other players involved in the two cycles are necessarily distinct.

Let $\bcO^2$ be the reallocation $\bcO$ with respect to $\C_2$. By Corollary $\ref{reallocation along cycle}$, $\bcO^2$ is also an optimal allocation in market $A$. Let $G_2$ be the legality graph corresponding to allocation $\bcO^2$. In the graph $G_2$, the path $x_0\to x_1\to \ldots\to x_{l_1-r-2}\to x_{l_1-r-1}$ remains unchanged in comparison to $G_1$. However, the cycle $y_0\to y_1\to \ldots\to y_{l_2-1}\to y_{l_2}$ in $G_1$ becomes the cycle $y_{l_2}\to y_{l_2-1}\to \ldots \to y_1\to y_0$ in $G_2$; and item $y_t$ is allocated to player $i_{y,t-1}$ by $\bcO^2$ for all $t=1,\ldots,l_2$.

In particular, $y_{l_2-r-1}$ is no longer allocated to $i_{y,l_2-r-1}=i_{x,l_1-r-1}$; and $y_1$ is now allocated to $i_{y,0}=i_{y,l_2}=i_{x,l_1}=i_{x,0}$. So, edges $x_{l_1-r-1}\to y_{l_1-r-1}$ and $y_1\to x_1$ are in $G_2$. Therefore, there is a cycle
\[x_1\to x_2\to \ldots\to x_{l_1-r-1}\to y_{l_2-r-1}\to y_{l_2-r-2}\to\ldots\to y_2\to y_1\to x_1\]
in the graph $G_2$ and let us  denote this cycle by $\C$.
Notice that $\C$ is a uniquely assigned cycle in $G_2$ of length
\[(l_1-r-1)+(l_2-r-1)=(l_1+l_2)-2r-2,\]
which is even since $l_1$ and $l_2$ are both odd. By definition, $\C$ is a removable cycle and the items in $\C$ form a removable set of type III in the graph $G_2$.
\end{proof}

Lastly, we need to consider the case when there is no removable set formed by a uniquely assigned cycle of even length and no odd cycle pair in the market. So every uniquely assigned cycle has odd length. In this case we can get an additional property about the items in uniquely assigned cycles. In particular, if we remove all but one item $x$ in  an odd length cycle $\C$ from the market, we can show that in the resulting submarket the item $x$ is legal to only one player. Using this property, we define a fourth type of removable set that contains all items in $\C$ except for one item $x$ and show that we can again alternately assign ``high'' and ``low'' prices to the items in such a removable set as we did in Lemma~\ref{prices of removable set of type III}.

\begin{definition}\label{removable set of type IV definition}
Let $\C$ be a uniquely assigned odd length cycle in $G$ of the form 
\[x_0\to x_1\to x_2\to\ldots\to x_{l-1}\to x_l=x_0\,.\]
If there does not exist a uniquely assigned even cycle or an odd cycle pair in the graph $G$, then $\{x_j\,:\,j=1,\ldots, l-1\}$ is called a removable set of type IV.
\end{definition}

\begin{lemma}\label{prices of removable set of type IV}
 Let $\C$ be a cycle in $G$ 
 \[x_0\to x_1\to x_2\to\ldots\to x_{l-1}\to x_l=x_0\]
 such that $\X_r\coloneqq\{x_j\,:\,j=1,\ldots, l-1\}\subseteq\X_A$ is a removable set of type IV in $G$. Let $B=(\X_B,I_B,(k_i^{\p})_{i\in I_B},\legal)$ be the submarket of $A$ induced by $\bcO$ with items $\X_B=\X_A\setminus\X_r$. Any pricing $\vp\in\real^m$ satisfying the following conditions is a dynamic pricing in market $A$.
\begin{enumerate}
\item For all $y_1\in\{x_j\,:\, j=1,\ldots, l-1,\,j\text{ is odd}\}$, $y_2\in\X_B$, and $y_3\in\{x_j\,:\,j=1,\ldots, l-1,\,j\text{ is even}\}$ \[p_{y_1}>p_{y_2}>p_{y_3}\,.\]
\item $(p_x)_{x\in\X_B}$ is a dynamic pricing in market $B$.
\end{enumerate}
\end{lemma}

\begin{proof}
Suppose $x_j\in\cO_{i_j}$ for $j=1,\ldots, l$ and let $I_r=:\{i_j\,:\, j=1,\ldots, l-1\}\subseteq I_A$. Note that 
\[|\cO_{\hat{i}}\cap\X_r|=\begin{cases}
1\quad&\text{if }i\in I_r\\
0\quad&\text{if }i\notin I_r
\end{cases}
\]
implies
\[
k_i^{\p}=\begin{cases}
k_i-1\quad&\text{if }i\in I_r\\
k_i\quad&\text{if }i\notin I_r
\end{cases}\,.\]

Our proof relies on the following claim.

\begin{claim}The item $x_l$ is assigned to player $i_l$ in every optimal allocations of market $B$.
\end{claim}

\begin{proof}
Let $\bcO^B$ be the restriction of $\bcO$ to market $B$ given by $\cO_i^B:=\cO_i\cap\X_B$ for all $i\in I_B$. The allocation $\bcO^B$ is an optimal allocation in market $B$ and $x_l$ is in $\cO_{i_l}^B$. Thus the item $x_l$ is legal to $i_l$ in market $B$. Let $G_B$ be the legality graph for $B$ corresponding  to allocation $\bcO^B$. 

Since $x_l$ cannot be legal to players other than $i_{l-1}$ and $i_l$ in market $B$ by Remark~\ref{legality in B implies legality in A}, it remains to show that $x_l$ is not legal to $i_{l-1}$ in market $B$.

Suppose for contradiction that $x_l$ is legal to $i_{l-1}$ in market $B$.
By Corollary~\ref{cycle of items allocated to different players}, there is a uniquely assigned cycle $\C_y$ 
\[x_l\to y_1\to y_2\to\ldots\to y_{l^{\p}-1}\to y_{l^{\p}}=x_l\] in $G_B$ that contains $x_l$. Suppose $y_t\in\cO_{i_{y,t}}^B$ for $t=1,\ldots, l^{\p}$. By Corollary $\ref{submarket implies subgraph}$, $\C_y$ is also a cycle in the legality graph $G$. If $l$ is even, then $\C_y$ is a removable cycle in market $A$, leading to a contradiction. So $l$ must be odd. Let 
\[t^*:=\min\{t:t=1,\ldots, l^{\p},\,i_{y,t}=i_j\text{ for some }j=1,\ldots, l\}\,.\]
Note that since $y_{l^{\p}}=x_l$, we have $i_{y,l^{\p}}=i_l$. So, the definition of $t^*$ is valid. Suppose $i_{y,t^*}=i_{j^*}$ for some $ j^*=1,\ldots, l$. Then, the edge $y_{t^*}\to x_{j^*+1}$ exists in $G$ and so the cycle $\C_y^*$
\[x_l\to y_1\to y_2\to\ldots\to y_{t^*}\to x_{j^*+1}\to x_{j^*+2}\to\ldots\to x_{l-1}\to x_l\] exists in $G$.  By definition of $t^*$,  the cycle $\C_y^*$ is a uniquely assigned cycle.
If $\C_y^*$ has odd length, then the cycles $\C$ and $\C_y^*$ form an odd cycle pair in $G$ with $r=l-j^*-1=l^{\p}-t^*-1$. If $\C_y^*$ has even length, then $\C_y^*$ is a removable cycle in $G$. So in both cases, we reach a contradiction.
\end{proof}

Consider a player $\hat{i}\in I_A$ and let us show that any bundle in demand for player $\hat{i}$ with respect to the prices $\vp$  can be extended to a legal allocation in market $A$. We consider the following three cases.

\textbf{Case $\hat{i}=i_{\hat{j}}\in I_r$ for some $\hat{j}=1,\ldots, l-1$ and $\hat{j}$ is even.}
This case is analogous to the first case in the proof of Lemma~\ref{prices of removable set of type III}.

\textbf{Case  $\hat{i}=i_{\hat{j}}\in I_r$ for some $\hat{j}=1,\ldots, l-1$  and $\hat{j}$ is odd.}
We use the same notation as in the second case in the proof of Lemma~\ref{prices of removable set of type III}. The argument is the same except that we no longer have $x_l\in\X_r$. By the claim above we have $x_l\in\cO_{i_l}^B$. So, similar to the definition of $\bcO^A$ in the second case in the proof of Lemma~\ref{prices of removable set of type III}, we can define a legal allocation in market $A$ as
 \[
\cO_i^A:=\begin{cases}
\cO_i^B\cup\{x_{j+1}\}\quad&\text{if }i\in I_r\text{ and }i=i_j\text{ for some }j=1,\ldots,l-1\\
\cO_{i_l}^B\setminus\{x_l\}\cup\{x_1\}\quad&\text{if }i= i_l\\
\cO_i^B\quad&\text{if }i\notin I_r\text{  and  }i\neq i_l\,.
\end{cases}
\]

\textbf{Case  $\hat{i}\notin I_r$}
This case is analogous to the third case in the proof of Lemma~\ref{prices of removable set of type III}.\end{proof}

Given a multi-demand market with a set $\X$ of $m$ items and a set $I$ of $n$ players. Suppose each player $i\in I$ is $k_i$-demand and has multi-demand valuation function $v_i:2^{\X}\to\real_{\ge0}^m$. Assume that the market has only two optimal allocations. We summarize the procedure for finding a dynamic pricing in a multi-demand market with two optimal allocations in Algorithm~\ref{algorithm for 2 optimal allocations}. The correctness of the algorithm is proved in the following theorem.

\begin{theorem}\label{correctness of algorithm for 2 optimal allocations}
A dynamic pricing exists in a multi-demand market with at most two optimal allocations. It can be computed using Algorithm~\ref{algorithm for 2 optimal allocations}.
\end{theorem}

\begin{proof}
We proceed by induction on the number of items $m$ in market. As in Theorem~\ref{correctness of algorithm for 4 players}, the base cases are markets with at most three items. From the induction hypothesis, we may assume that a dynamic pricing can be found in any multi-demand market with less than $m$ items and at most two optimal allocations.

We first compute the rough prices as in Algorithm~\ref{rough price} and let market $A:=(\X^2,I_2,(\hat{k}_i)_{i\in I_2},\legal)$ be the induced simplified market described in Section~\ref{Fine Price}. We now show that the price vector $\vp^F$ defined in line~\ref{lst:line:2 opt. allocs. alg. p^F def} of Algorithm~\ref{algorithm for 2 optimal allocations} is a dynamic pricing in market $A$.

Let $G$ be the legality graph with respect to some legal allocation $\bcO$ of market $A$. By Lemma~\ref{cycle of items allocated to different players} and Definition~\ref{removable set of type IV definition}, we know that one of the following three situations takes place: $G$ has a removable cycle, $G$ has an odd cycle pair, or there is a removable set of type IV in market $A$.

If there does not exist a removable cycle in $G$ and there is an odd cycle pair in $G$, which corresponds to line~\ref{lst:line:2 opt. allocs. alg. check odd cycle pair} in Algorithm~\ref{algorithm for 2 optimal allocations}, we can reallocate items with respect to one of the odd cycles to obtain a new legality graph. Lemma~\ref{odd cycle pair implies removable cycle} guarantees that the new graph contains a removable cycle. So, after running lines~\ref{lst:line:2 opt. allocs. alg. check odd cycle pair}-\ref{lst:line:2 opt. allocs. alg. realloc wrt. odd cycle pair} of Algorithm~\ref{algorithm for 2 optimal allocations}, we have that either $G$ contains a removable cycle, which means that market $A$ has a removable set of type III, or market $A$ has a removable set of type IV.

Then, we can define a removable set $\X_r$ and remove it from market $A$ to obtain a submarket $B$, whose dynamic pricing $\vp^B$ exists by induction hypothesis. With the valuation functions defined as in line~\ref{lst:line:2 opt. allocs. alg. valuation def} of Algorithm~\ref{algorithm for 2 optimal allocations}, the dynamic pricing of market $B$ must be in the range $\mathbf{0}<\vp^B<\Delta/2$ in order to yield positive utility for all items that are legal to a given player. By alternately assigning ``high'' and ``low'' prices to items in $\X_r$, Lemmas~\ref{prices of removable set of type III} and \ref{prices of removable set of type IV} ensure that the prices $\vp^F$ defined in line~\ref{lst:line:2 opt. allocs. alg. p^F def} of Algorithm~\ref{algorithm for 2 optimal allocations} is a dynamic pricing in market $A$.

Finally, since $\mathbf{0}<\vp^F\le\Delta/2<\Delta$ by Definition~\ref{definition fine prices} of fine prices, $\vp^F$ is a vector of fine prices of the original market. From Section~\ref{Fine Price}, we know that the pricing $\vp^D$ defined in line~\ref{lst:line:2 opt. allocs. alg. p^D def} is a desired dynamic pricing.
\end{proof}

\begin{algorithm}[H]
\caption{Finding dynamic pricing in multi-demand markets with at most two optimal allocations.}\label{algorithm for 2 optimal allocations}
\begin{algorithmic}[1]
\State Define $\vp^R$ to be the rough prices of the market.
\State Define $\Delta>0$ as in Section $\ref{Fine Price}$.
\State Define $\X^2:=\X\setminus\left(\bigcup\limits_{i\in I}\R_i\right)$, $I_2:=\{i\in I:\R_i\neq\K_i\}$, and $\hat{k}_i:=k_i-|\R_i|$ for all $i\in I_2$.
\State Define market $A=(\X^2,I_2,(\hat{k}_i)_{i\in I_2},\legal)$.
\State Define
$$\hat{v}_i(x):=\begin{cases}
\Delta/2\quad&\text{if $x$ is legal to $i$}\\
0\quad&\text{if $x$ is not legal to $i$}
\end{cases}$$
for all $i\in I_2$ and $x\in\X^2$.\label{lst:line:2 opt. allocs. alg. valuation def}
\State Define $G$ to be the legality graph with respect to a legal allocation $\bcO$ of market $A$.
\If{there does not exist a removable cycle in $G$ and there is an odd cycle pair $(\C_1,\C_2)$ in $G$}\label{lst:line:2 opt. allocs. alg. check odd cycle pair}
    \State Redefine $\bcO$ to be the reallocation of $\bcO$ with respect to cycle $\C_2$.
    \State Redefine $G$ to be the legality graph with respect to the new legal allocation $\bcO$.\label{lst:line:2 opt. allocs. alg. realloc wrt. odd cycle pair}
\EndIf
\If{there is a removable cycle $x_0\to x_1\ldots x_{l-1}\to x_l=x_0$ in $G$.}
    \State Define $\X_r:=\{x_j:1\le j\le l\}$.
\Else
    \State Define $x_0\to x_1\ldots x_l\to x_{l+1}=x_0$ to be a cycle in $G$ such that $\X_r=\{x_j:1\le j\le l\}$ is a removable set of type IV.
\EndIf
\State Define $\vp^B$ to be a dynamic pricing in the submarket $B$ of $A$ induced by $\bcO$ with items $\X^2\setminus\X_r$.
\State Define $\hat{p}:=\min\{p_x^B:x\in\X^2\setminus\X_r\}$.
\State Define the fine prices $\vp^F$ by $p_x^F:=\begin{cases}
\hat{p}/2\quad&\text{if $x=x_j$ for some $1\le j\le l$, $j$ even}\\
p_x^B\quad&\text{if }x\in\X^2\setminus\X_r\\
\Delta/2\quad&\text{if $x=x_j$ for some $1\le j\le l$, $j$ odd}
\end{cases}$.\label{lst:line:2 opt. allocs. alg. p^F def}
\State Define the dynamic pricing $\vp^D$ by $p_x^D:=\begin{cases}
p_x^R+p_x^F\quad&\text{if }x\in\X^2\\
p_x^R\quad&\text{if }x\in\X\setminus\X^2
\end{cases}$.\label{lst:line:2 opt. allocs. alg. p^D def}
\State \Return $\vp^D$
\end{algorithmic}
\end{algorithm}

\section{Market with Tri-Demand Players}\label{Market with 3-Demand Players}

We consider a saturated market $A=(\X_A,I_A,\mathbf{k},\legal\})$. Let us again use the usual notation for the number of items and players in the market $m:=|\X_A|$ and $n:=|I_A|$. Following the ideas in Section~\ref{Fine Price}, it is enough to consider markets with the following valuation functions for $i\in I_A$, $x\in\X_A$
\[v_i(x)=\begin{cases}
1\quad\text{if }x\in\legal(i)\\
0\quad\text{if }x\notin\legal(i)\,.
\end{cases}\]

 In this section, we focus on the special case where the demand  of every player is at most three, i.e.\ $k_i$ is in $\{1,2,3\}$ for all players $i\in I_A$. We define the following property:
\begin{equation}\tag{P}
\text{An item $x\in\X_A$ is legal to a player $i\in I_A$ if and only if $v_i(x)=1$.}\label{legality-valuation equivalence}
\end{equation}
Notice that market $A$ as a whole satisfies property~\eqref{legality-valuation equivalence}, but since the converse of Remark~\ref{legality in B implies legality in A} is not necessarily true, it is not always the case that every submarket of $A$ satisfies property~\eqref{legality-valuation equivalence}. For saturated tri-demand markets with property~\eqref{legality-valuation equivalence}, instead of simply showing the existence of a dynamic pricing, we prove a stronger statement: For every item $x\in\X_A$, there exists a dynamic pricing in market $A$ such that item $x$ is priced ``the lowest''.

\begin{definition}
Given item $x\in\X_A$ and dynamic pricing $\vp\in\real_{>0}^m$ of market $A$, we say that $\vp$ is a dynamic pricing fixed at $x$ if $p_y>p_x$ for all $y\in\X_A\setminus\{x\}$.
\end{definition}

In this section, we work extensively with ordering of items in terms of prices, so we introduce the following notation. Given a price vector $\vp\in\real^m$, for two disjoint subsets $\X_1,\X_2\subseteq\X_A$, we write $\X_1<_p\X_2$ to indicate that for all $x\in\X_1$ and $y\in\X_2$ we have $p_x<p_y$.

We often need to define the lowest-priced items in submarkets of market $A$ and use them to combine dynamic pricings in different submarkets in a meaningful way. To ensure that such items are well-defined, we make the
assumption that given a dynamic pricing in a submarket of $A$, no two items have the same price. This assumption is
valid because of the following reasoning. Consider a dynamic pricing $\vp$ where the price $p_x$ of an item $x$ is the same for at least two different items in the market. For all players and all items $y$ with strictly lower utility than $x$, we define a positive $\varepsilon$ such that $\varepsilon$ is less than the smallest difference in utility and $p_x+\epsilon$ do not coincide with any existing price in the market. Then, we
redefine the price of $x$ to be $p_x+\varepsilon$. Notice that any item that used to have higher (lower) utility
than $x$ still has higher (lower) utility. Consider a player $i$ for whom some item $y$ previously had
the same utility as $x$. This means that $i$ was indifferent between $x$ and $y$ and $i$ is free to pick either
item. So, increasing the price of $x$ by $\varepsilon$ yields another dynamic pricing. We can repeat this process until every item in the market has a different price.

We would like to use induction on the number of players plus the number of items in the markets. We first prove two base cases, namely a market with two players and a market with two item.

\begin{lemma}\label{base case of 2 players}
If $n=2$ then for every $x\in\X_A$ there exists a dynamic pricing fixed at $x$.
\end{lemma}

\begin{proof}
It is enough to select prices $\vp$ such that $\legal(1)\setminus\legal(2)<_p\left(\legal(1)\cap\legal(2)\right)\setminus \{x\}$ and
$\legal(2)\setminus\legal(1)<_p\left(\legal(1)\cap\legal(2)\right)\setminus \{x\}$ and $\{x\}<_p\X_A\setminus \{x\}$. 
\end{proof}

\begin{corollary}\label{base case of 2 items}
If $m=2$ then for every $x\in\X_A$ there exists a dynamic pricing fixed at $x$.
\end{corollary}

\begin{proof}
This is immediate from Lemma~\ref{base case of 2 players} since $m=2$ implies that $n\le2$.
\end{proof}

With the two base cases proved, for the rest of this section, we make the following assumption. In any submarket $B=(\X_B,I_B,(k_i)_{i\in I_B},\legal)$ of $A$ with $\X_B\subset\X_A$ or $I_B\subset I_A$, such that property~\eqref{legality-valuation equivalence} holds for $B$, there exists a dynamic pricing fixed at $x$ for any $x\in\X_B$. We will refer to this assumption as the induction hypothesis. 
 
In this section, we adopt a slightly different notion of submarket, which we use for the rest of this paper. 

\begin{definition}
A submarket of market $A$ is a market $B=(\X_B,I_B,(k_i)_{i\in B},\legal)$ such that $\X_B\subseteq\X_A$ and $I_B\subseteq I_A$.
\end{definition}

We also extend the notion of the legality function to subsets of players as follows.

\begin{definition}
We define $\legal:2^{I_A}\times2^{\X_A}\to2^{\X_A}$ to be the function given by 
\[\legal(I,\X):=\{x\in\X:x\text{ is legal to }i\text{ for some }i\in I\}\]
for all $I\subseteq I_A$ and $\X\subseteq\X_A$.
\end{definition}

For simplicity, when $I=\{i\}$ or $\X=\{x\}$ is a singleton, we can omit the set brackets and write $\legal(I,\X)=\legal(i,\X)$ or $\legal(I,\X)=\legal(I,x)$.
Notice that $\legal(i,\X_A)=\legal(i)$ for all $i\in I_A$. So, this new definition is indeed a generalization of the original legality function.

Note that by construction, the legality function $\legal$ satisfies the following properties:
\begin{enumerate}[ref={\thedefinition.\arabic*}]
\item For all $I\subseteq I_A$ and $\X\subseteq\X_A$, we have $\legal(I,\X)\subseteq\X$.\label{L-prop.1}
\item For all $I_1$, $I_2\subseteq I_A$ and $\X_1,\X_2\subseteq\X_A$, we have \label{L-prop.2}
\[\legal(I_1\cup I_2,\X_1\cup \X_2)=\legal(I_1,\X_1)\cup\legal(I_1,\X_2)\cup\legal(I_2,\X_1)\cup\legal(I_2,\X_2)\,.\]
\item For all $I_1$, $I_2\subseteq I_A$, such that $I_1\subseteq I_2$, and $\X\subseteq \X_A$, we have $\legal(I_1,\X)\subseteq\legal(I_2,\X)$.\label{L-prop.3}
\item For $I\subseteq I_A$ and $\X_1$, $\X_2$, such that $\X_2\subseteq \X_A$, $\X_1\subseteq \X_2$, we have $\legal(I,\X_1)\subseteq\legal(I,\X_2)$.\label{L-prop.4}
\end{enumerate}

A \emph{submarket pair} consists of two related submarkets and can be viewed as a partition of market~$A$ with additional properties specified in Defintion~\ref{submarket pair definition}. 

\begin{definition}\label{submarket pair definition}
Suppose $B=(\X_B,I_B,(k_i)_{i\in I_B},\legal)$ and $C=(\X_C,I_C,(k_i)_{i\in I_C},\legal)$ are submarkets of $A$, such that the following conditions are satisfied.
\begin{enumerate}
\item $\X_B\cup\X_C=\X_A$, $\X_B\cap\X_C=\varnothing$, and $\X_B,\X_C\neq\varnothing$.\label{UL-cond.1}
\item $I_B\cup I_C=I_A$, $I_B\cap I_C=\varnothing$ and $I_B,I_C\neq\varnothing$.\label{UL-cond.2}
\item $\legal(I_B,\X_A)\subseteq\X_B$.\label{UL-cond.3}
\item $|\X_B|-1=\sum\limits_{i\in I_B}k_i$ and $|\X_C|+1=\sum\limits_{i\in I_C}k_i$.\label{UL-cond.4}
\item $|\legal(I_C,\X_B)|\ge2$ and for all $\hat{x}\subseteq\legal(I_C,\X_B)$, there is a legal allocation $\bcO$ in market $A$ such that $ \bigcup\limits_{i\in I_B}\cO_i=\X_B\setminus\{\hat{x}\}$ and $\bigcup\limits_{i\in I_C}\cO_i=\X_C\cup\{\hat{x}\}$\label{UL-cond.5}
\end{enumerate}
Then, we say that $(B,C)$ form a submarket pair in market $A$.
\end{definition}

By definition of a legal item, we know that any assignment of a single item can be extended to a legal allocation as long as legality is respected. We define such an assignment as a \emph{legal assignment}.

\begin{definition}\label{legal assignment}
A legal assignment of a subset $\X\subseteq\X_A$ is a way of assigning items in $\X$ to players $I_A$ so that the items received by each player are legal for that player. Such assignment is also called a legal assignment of size $\hat{m}$ when $\hat{m}=|\X|$. \end{definition}

We could think of a legal assignment as a vector $\mathbfcal{A}=(\mathcal{A}_i)_{i\in I_A}$ of disjoint subsets of $\X$, indexed by the players. When the number of items received by every player $i\in I_A$ does not exceed the demand $k_i$, we say that the legal assignment \emph{respects demand}. If there is a legal allocation $\bcO$ where $\mathcal{A}_i\subseteq\cO_i$ for all $i\in I_A$, we say the legal assignment is extendable. Otherwise, we say that the legal assignment is not extendable.

Note that a legal assignment of more than one item is not necessarily extendable. Let us consider a legal assignment of size two that is not extendable. 

\begin{lemma}\label{legal assignment of size 2 non extendable implies submarket pair}
Suppose there exist two items $x_1,x_2\in\X_A$ such that some legal assignment of $\{x_1,x_2\}$ that respects demand is not extendable, then there is a submarket pair $(B,C)$ in market $A$.
\end{lemma}

\begin{proof}
Suppose $x_1$ and $x_2$ are assigned to player $i_1$ and $i_2$, respectively. Note that it is possible that $i_1=i_2$. Consider a bipartite graph $H$, with the two parts being the set of items $\X_A$ and the set of players $I_A$. For every player $i\in I_A$, the player $i$ is represented by $k_i$ copies in $H$. There is an edge between a player $i$ and an item $x$ if and only if $x$ is legal to $i$. Thus, every perfect matching in $H$ corresponds to a legal allocation, and every legal allocation can be represented by some perfect matching in $H$.

 Let us remove $x_1$, $x_2$ and one copy of each of $i_1$, $i_2$  from $H$ to get a subgraph $H^*$. Since the above assignment of $x_1$ and $x_2$ is not extendable,  we have that $H^*$ has no perfect matchings. Thus using Hall's marriage theorem, there exists $I_B\subseteq I_A$ such that $|\X^*|<|I_B^*|$, where $\X^*:=\legal(I_B,\X_A\setminus\{x_1,x_2\})$ and $|I_B^*|$ denotes the total number of copies of players in $I_B$ that are in $H^*$. That is
\[|I_B^*|=\begin{cases}
\sum\limits_{i\in I_B}k_i\quad&\text{if }i_1,i_2\notin I_B\\
\sum\limits_{i\in I_B}k_i-1\quad&\text{if exactly one of $i_1,i_2$ is in $I_B$}\\
\sum\limits_{i\in I_B}k_i-2\quad&\text{if }i_1\in I_B\text{ and }i_2\in I_B\,.
\end{cases}\]
We can show that the last two cases are not possible. Suppose $i_1\in I_B$ and $i_2\notin I_B$. Then,
\begin{equation}\label{i_1 in I_B leads to contradiction}
|\legal(I_B,\X_A\setminus\{x_2\})|=|\X^*\cup\{x_1\}|=|\X^*|+1<|I_B^*|+1=\sum\limits_{i\in I_B}k_i\,.
\end{equation}
This means that there is no perfect matching in $H$ when $x_2$ is assigned to $i_2$. This is a contradiction since a legal assignment of a single item is always extendable. The other cases of $i_1\notin I_B,i_2\in I_B$ and $i_1,i_2\in I_B$ are similar. Thus, the set $I_B$ is a strict subset of $I_A$.

Define $\X_B=\X^*\cup\{x_1,x_2\}$, $\X_C=\X_A\setminus\X_B$ and $I_C=I_A\setminus I_B$. Let $B=(\X_B,I_B,(k_i)_{i\in I_B},\legal)$ and $C=(\X_C,I_C,(k_i)_{i\in I_C},\legal)$ so that $B$ and $C$ are submarkets of market A. Notice that $i_1$ is in $I_C$.  It is straightforward to check that $(B,C)$ forms a submarket pair.
\end{proof}

To make use of the property~\eqref{legality-valuation equivalence}, we would like the submarket $B$ in the above lemma to be inclusion-wise maximal in the following sense.

\begin{definition}\label{maximal submarket pair definition}
Let $(B,C)$ be a submarket pair of $A$ where $B$ has items $\X_B$. We say the pair $(B,C)$ is maximal if for all submarket pairs $(B^*,C^*)$ where market $B^*$ has items $\X_B^*$, we have that $\X_B\subseteq\X_B^*$ implies $\X_B=\X_B^*$.
\end{definition}

The next lemma reveals connections between maximal submarket pairs and the property~\eqref{legality-valuation equivalence}.

\begin{lemma}\label{B' and C' satisfies legal iff valuation = 1 property}
Suppose $(B,C)$ is a maximal submarket pair of $A$, with $B=(\X_B,I_B,(k_i)_{i\in I_B},\legal)$ and $C=(\X_C,I_C,(k_i)_{i\in I_C},\legal)$. Then the following holds.
\begin{enumerate}[ref={\thedefinition.\arabic*}]
\item Let $i^\p$ be an artificial unit-demand player such that $\legal(i^{\p})\coloneqq\legal(I_C,\X_B)$\footnote{Note that $\legal$ is a function that defines legality for players in market $A$, here we formally extend the definition of $\legal$ to one artificial player. So, $\legal(i^{\p})$ simply defines a $\{0,1\}$-valued valuation function on $\X_B$. $\legal$ provides no information about the legality of items with respect to market $B$.} and \[v_{i^\p}(x):=\begin{cases}
1\quad\text{if }x\in\legal(i^\p)\\
0\quad\text{if }x\notin\legal(i^\p)\,.
\end{cases}\]
Let $I_B^{\p}:=I_B\cup\{i^{\p}\}$. Then, the market $B^{\p}\coloneqq(\X_B,I_B^{\p},(k_i)_{i\in I_B^{\p}},\legal)$ satisfies property~\eqref{legality-valuation equivalence}.\label{legal in B' iff valuation is 1}
\item For all $x_1\in\legal(I_C,\X_B)$, market $C^{\p}:=(\X_C\cup\{x_1\},I_C,(k_i)_{i\in I_C},\legal)$ satisfies property~\eqref{legality-valuation equivalence}.\label{legal in C' iff valuation is 1}

\end{enumerate}
\end{lemma}

\begin{proof}
By~\ref{UL-cond.4} in Definition~\ref{submarket pair definition}, both $B^{\p}$ and $C^{\p}$ are saturated markets. For any allocation $\bcO^B$ in market $B^{\p}$ and allocation $\bcO^C$ in market $C^{\p}$, since the valuation function takes a value of either 0 or 1, we have an upper bounds for the social welfare
$$\SW(\bcO^B)=\sum\limits_{i\in I_B^{\p}}v_i(\cO_i^B)\le\sum\limits_{i\in I_B^{\p}}k_i\quad\text{and}\quad\SW(\bcO^C)=\sum\limits_{i\in I_C}v_i(\cO_i^C)\le\sum\limits_{i\in I_C}k_i\,.$$

\begin{claim}\label{legal in A implies legal in B'}
For every player $i\in I_B$ and item $x\in\X_B$, if $x\in\legal(i)$ then $x$ is legal to $i$ in market~$B^{\p}$.
\end{claim}
\begin{proof}
Consider a player $i^*\in I_B$ and an item $x\in\legal(i^*)$. Since $x$ is legal to $i^*$ in market $A$, there exists an allocation $\bcO^A$ in market $A$ such that $x\in\cO_{i^*}^A$. Thus  by~\ref{UL-cond.3}  in Definition~\ref{submarket pair definition} we have
\[\bigcup\limits_{i\in I_B}\cO_i^A\subseteq\legal(I_B,\X_A)\subseteq\X_B\]
and by~\ref{UL-cond.4}  in Definition~\ref{submarket pair definition} we have
\[|\bigcup\limits_{i\in I_B}\cO_i^A|=\sum\limits_{i\in I_B}k_i=|\X_B|-1.\]
So,
\[|\X_B\setminus\bigcup\limits_{i\in I_B}\cO_i^A|=1\,.\]

Let $x^{\p}$ be the single item in $\X_B\setminus\bigcup\limits_{i\in I_B}\cO_i^A$, then $x^{\p}$ must be assigned to some player in $I_C$ by $\bcO^A$. So, we have $x^{\p}\in\legal(I_C,\X_B)=\legal(i^{\p})$. Let $\bcO^B$ be the allocation such that
$$\cO_i^B:=\begin{cases}
\cO_i^A\quad&\text{if }i\in I_B\\
\{x^{\p}\}\quad&\text{if }i=i^{\p}\,.
\end{cases}$$
Note that $x\in\cO_{i^*}^B$ and $\bcO^B$ is an optimal allocation in market $B^{\p}$ since $\SW(\bcO^B)=\sum\limits_{i\in I_B^{\p}}k_i$. 
\end{proof}

\begin{claim} For every player $i\in I_B$ and item $x\in\X_B$, if  $x$ is legal to $i$ in market $B^{\p}$ then $x\in\legal(i)$.
\end{claim}
\begin{proof}
Consider a player $i^*\in I_B$ and item $x\in\X_B$. Suppose $x$ is legal to $i^*$ in market $B^{\p}$, then there exists an optimal allocation $\bcO^B$ in market $B^{\p}$ such that $x\in\cO_{i^*}^B$. Let $x^{\p}$ be the single item in $\cO_{i^{\p}}^B$, then $x^{\p}\in\legal(i^{\p})=\legal(I_C,\X_B)$. By property~\ref{UL-cond.5}  in Definition~\ref{submarket pair definition}, there exists a legal allocation~$\bcO^A$ in market $A$ such that \[\bigcup\limits_{i\in I_C}\cO_i^A=\X_C\cup\{x^{\p}\}\,.\] The allocation given by $\cO_i^B$ for every player $i\in I_B$ and $\cO_i^A$ for every player $i\in I_C$ is a legal allocation in market $A$. So,  we have $x\in\legal(i^*)$.
\end{proof}

\begin{claim}\label{legal in A iff legal in B'}
For every item $x^{\p}\in\X_B$, $x^{\p}\in\legal(i^\p)$ if and only if $x^{\p}$ is legal to $i^\p$ in market $B^{\p}$.
\end{claim}
\begin{proof}
Consider an item $x^{\p}\in\legal(i^\p)=\legal(I_C,\X_B)$. By property~\ref{UL-cond.5} in Definition~\ref{submarket pair definition}, there exists a legal allocation $\bcO^A$ in $A$ such that \[\bigcup\limits_{i\in I_B}\cO_i^A=\X_B\setminus\{x^{\p}\}\,.\] We define $\bcO^B$ in the same way as in Claim~\ref{legal in A implies legal in B'}, then $x^{\p}\in\cO_{i^{\p}}^B$ and $\bcO^B$ is again a legal allocation in market $B^{\p}$. Conversely, suppose an item $x^{\p}\in\X_B$ is legal to player $i^{\p}$ in market $B$. Then there exists an optimal allocation $\bcO^B$ in market $B^{\p}$ such that $x^{\p}\in\cO_{i^{\p}}^B$. Since the social welfare is optimal, we must have $v_i(\cO_i)=k_i$ for all $i\in I_B^{\p}$. In particular, $v_{i^{\p}}(x^{\p})=v_{i^{\p}}(\cO_{i^{\p}})=k_{i^{\p}}=1$, so we have $x\in\legal(i^{\p})$.
\end{proof}

Therefore, we conclude that an item $x\in\X_B$ is legal to a player $i\in I_B^{\p}$ in market $B^{\p}$ if and only if $x\in\legal(i)$ if and only if $v_i(x)=1$. Thus, market $B^{\p}$ satisfies property~\eqref{legality-valuation equivalence}, proving the first part of the statement.

Next, let us consider an item ${x_1}\in\legal(I_C,\X_B)$ and consider the market $C^{\p}:=(\X_C^{\p},I_C,(k_i)_{i\in I_C},\legal)$ where $\X_C^{\p}=\X_C\cup\{x_1\}$.

\begin{claim}\label{legal in C' implies legal in A}
For every player $i\in I_C$ and item $x\in\X_C^{\p}$, if $x$ is legal to $i$ in market $C^{\p}$ then $x\in\legal(i)$.
\end{claim}

\begin{proof}
Consider a player $i\in I_C$ and item $x\in\X_C^{\p}$. By property~\ref{UL-cond.5} in Definition~\ref{submarket pair definition} there is a legal allocation $\bcO^A$ in market $A$ such that \[\bigcup\limits_{i\in I_C}\cO_i^A=\X_C^{\p}\,,\] so $C^{\p}$ is a submarket of $A$ induced by the optimal allocation $\bcO$. It then follows from Corollary~\ref{legality in B implies legality in A} that $x\in\legal(i)$.
\end{proof}

\begin{claim}\label{legal assignment in A is extendable in C'}
For every $x_2\in\X_C$, any legal assignment with respect to market $A$ of $\{x_1,x_2\}$ to $I_C$ that respects demand is extendable with respect to market $C^{\p}$.
\end{claim}

\begin{proof}
Suppose for contradiction that there exists an item $x_2\in\X_C$, such that a legal assignment of $\{x_1,x_2\}$ to $I_C$ with respect to market $A$ that respects demand is not extendable to a legal allocation in market $C^{\p}$. Suppose $x_1$ and $x_2$ are assigned to players $i_1$ and $i_2$ respectively. This implies that there also does not exist a legal allocation in market $A$ with $x_1$ and $x_2$ assigned to $i_1$ and $i_2$. Indeed, if such a legal allocation $\bcO^A$ exists in $A$, then  $(\cO_i^A)_{i\in I_C}$ is a legal allocation in market~$C^{\p}$. 

By Lemma~\ref{legal assignment of size 2 non extendable implies submarket pair}, there exists a submarket pair $(\hat{B},\hat{C})$ in $A$. We write $\hat{B}=(\hat{\X}_B,\hat{I}_B,(k_i)_{i\in\hat{I}_B},\legal)$ and $\hat{C}=(\hat{\X}_C,\hat{I}_C,(k_i)_{i\in\hat{I}_C},\legal)$. From the proof of Lemma~\ref{legal assignment of size 2 non extendable implies submarket pair}, we can assume that $x_1,x_2\in\hat{\X}_B$ and $i_1,i_2\in\hat{I}_C$. Let us define two new submarkets.
\begin{align*}
&B^*:=(\X_B^*,I_B^*,(k_i)_{i\in I_B^*},\legal)\quad\text{where} \quad\X_B^*=\X_B\cup\hat{\X}_B,I_B^*=I_B\cup\hat{I}_B,\\
&C^*:=(\X_C^*,I_C^*,(k_i)_{i\in I_C^*},\legal)\quad\text{where} \quad\X_C^*=\X_C\cap\hat{\X}_C,I_C^*=I_C\cap\hat{I}_C.
\end{align*}
Since $x_2\in\hat{\X}_B\setminus\X_B\subseteq\X_B^*\setminus\X_B$, we have that $\X_B\subset\X_B^*$. Now, we would like to show that $(B^*,C^*)$ form a submarket pair. This would imply that $(B,C)$ is not maximal and we have a contradiction.

First of all, since $(B,C)$ and $(\hat{B},\hat{C})$ satisfies properties~\ref{UL-cond.1} and~\ref{UL-cond.2} in Definition~\ref{submarket pair definition}, we have the following.
\begin{align*}
&\X_B^*\cup\X_C^*=\X_A,\quad\X_B^*\cap\X_C^*=\varnothing,\quad\X_B^*\neq\varnothing,\\
&I_B^*\cup I_C^*=I_A,\quad I_B^*\cap I_C^*=\varnothing,\quad I_B^*\neq\varnothing.
\end{align*}
Since $i_1,i_2\in I_C\cap\hat{I}_C=I_C^*$, we also have $I_C^*\neq\varnothing$.

Since $(B,C)$ and $(\hat{B},\hat{C})$ satisfies property~\ref{UL-cond.3} in Definition~\ref{submarket pair definition} we have
$$\legal(I_B^*,\X_A)=\legal(I_B,\X_A)\cup\legal(\hat{I}_B,\X_A)\subseteq\X_B\cup\hat{\X}_B=\X_B^*$$
and so, $(B^*,C^*)$ satisfies property~\ref{UL-cond.3} as well. 
 
Now, to show property~\ref{UL-cond.4} in Definition~\ref{submarket pair definition}, let us compute the size of $\X_B^*$ in two ways.
\begin{align}
|\X_B^*|=|\X_B\cup\hat{\X}_B|&=|\X_B|+|\hat{\X}_B|-|\X_B\cap\hat{\X}_B|=\sum\limits_{i\in I_B}k_i+1+\sum\limits_{i\in\hat{I}_B}k_i+1-|\X_B\cap\hat{\X}_B|\nonumber\\
&=\sum\limits_{i\in I_B^*}k_i+\sum\limits_{i\in I_B\cap\hat{I}_B}k_i+2-|\X_B\cap\hat{\X}_B|\label{size of X_B^*}\\
|\X_B^*|=|\X_B\cup\hat{\X}_B|&\ge|\legal(I_B,\X_A)\cup\legal(\hat{I}_B,\X_A)|=|\legal(I_B\cup\hat{I}_B,\X_A)|\ge\sum\limits_{i\in I_B\cup\hat{I}_B}k_i=\sum\limits_{i\in I_B^*}k_i\label{lower bound on X_B union X_B hat}\\
|\X_B\cap\hat{\X}_B|&\ge|\legal(I_B,\X_A)\cap\legal(\hat{I}_B,\X_A)|\geq|\legal(I_B\cap\hat{I}_B,\X_A)|\ge\sum\limits_{i\in I_B\cap\hat{I}_B}k_i\label{lower bound on X_C intersect X_C hat}
\end{align}
If ($\ref{lower bound on X_B union X_B hat}$) holds with equality, then in any legal allocations of market $A$, items in $\X_B^*$ are all allocated to $I_B^*$, so none of the items in $\X_B^*$ can be legal to players in $I_C^*$. This contradicts the fact that $x_1\in\legal(i_1)$ where $x_1\in\X_B\cap\hat{\X}_B\subseteq\X_B^*$ and $i_1\in I_C^*$. So, ($\ref{lower bound on X_B union X_B hat}$) is a strict inequality. Similarly, ($\ref{lower bound on X_C intersect X_C hat}$) is also a strict inequality. Combine these two inequalities with equation ($\ref{size of X_B^*})$, we conclude that
$$|\X_B^*|=\sum\limits_{i\in I_B^*}k_i+1\,.$$
It is then straightforward to check that $(B^*,C^*)$ satisfies property~\ref{UL-cond.4}.

Since $i_1,i_2\in I_C^*$, we have
\begin{equation}\label{X_C^* is nonempty}
|\X_C^*|=\sum\limits_{i\in I_C^*}k_i-1\ge2-1=1\end{equation} and so $$\quad\X_C^*\neq\varnothing\,.
$$
So $(B^*,C^*)$ satisfies properties~\ref{UL-cond.1} and~\ref{UL-cond.2} in Definition~\ref{submarket pair definition}.
 
Finally, since $\{x_1,x_2\}\subseteq\legal(\{i_1,i_2\},\X_B^*)\subseteq\legal(I_C^*,\X_B^*)$ we have that $(B^*,C^*)$ satisfies property~\ref{UL-cond.5}. So, $(B^*,C^*)$ is indeed a submarket pair, leading to the desired contradiction.
\end{proof}

\begin{claim}\label{legal in market A implies legal in market C'}
For every player $i\in I_C$ and item $x\in\X_C^{\p}$, if $x\in\legal(i)$ then $x$ is legal to $i$ in market~$C^{\p}$.
\end{claim}

\begin{proof}
First, let us consider the case where there exists a unique $i_1\in I_C$ such that $x_1\in\legal(i_1)$ and $i_1$ is unit-demand. For every player $i_2\in I_C$ such that $i_2\neq i_1$ and for every item $x_2\in\legal(i_2,\X_C)$, the previous claim implies that assigning $x_1$ to $i_1$ and $x_2$ to $i_2$ respects demand and extends to a legal allocation in market $C^{\p}$. So, $x_2$ is legal to $i_2$ in market $C^{\p}$.

Otherwise, consider a player $i_2\in I_C$ and item $x_2\in\legal(i_2,\X_C^{\p})$. Consider a player $i_1\in I_C$ such that $x_1\in\legal(i_1)$ and either $i_1\neq i_2$ or $k_{i_1}>1$. Using the same extendable legal assignment as above, we have that $x_1$ is legal to $i_1$ and $x_2$ is legal to $i_2$ in market $C^{\p}$.
\end{proof}

Therefore, we conclude that an item $x\in\X_C^{\p}$ is legal to a player $i\in I_C$ in market $C^{\p}$ if and only if $x\in\legal(i)$ if and only if $v_i(x)=1$. So, market $C^{\p}$ satisfies property~\eqref{legality-valuation equivalence}. This proves the second part of the statement.

\end{proof}

Now, we can show that when there is a legal assignment of size two that respects demand but is not extendable, there is a dynamic pricing that is fixed at $x$ for any item $x\in\X_A$.

\begin{lemma}\label{legal assignment of size 2}
Suppose there exist two items $x_1,x_2\in\X_A$ such that a legal assignment of $\{x_1,x_2\}$ that respects demand is not extendable, then for every $x\in\X_A$, there is a dynamic pricing that is fixed at $x$.
\end{lemma}

\begin{proof}
By Lemma~\ref{legal assignment of size 2 non extendable implies submarket pair}, we can select $(B,C)$ to be a maximal submarket pair in $A$. Take an item $x$ and let us consider the submarket that item $x$ belongs and the legality of $x$.

\textbf{Case $x\in\legal(I_C,\X_B)$.}
Define artificial player $i^{\p}$ and market $B^{\p}$ as in Lemma~\ref{legal in B' iff valuation is 1}. Define $C^{\p}:=(\X_C\cup\{x\},I_C,(k_i)_{i\in I_C},\legal)$. By property~\ref{UL-cond.4} in Definition~\ref{submarket pair definition}, $B^{\p}$ and $C^{\p}$ are saturated markets. By Lemma~\ref{B' and C' satisfies legal iff valuation = 1 property}, $B^{\p}$ and $C^{\p}$ satisfies property~\eqref{legality-valuation equivalence}. By property~\ref{UL-cond.1} and \ref{UL-cond.2} in Definition~\ref{submarket pair definition}, we have $|\X_B|<|\X_A|$ and $|I_C|<|I_A|$. By induction hypothesis, there is a dynamic pricing $\vp^B$ in market $B^{\p}$ that is fixed at $x$ and there is a dynamic pricing $\vp^C$ in market $C^{\p}$ that is fixed at $x$. Then, we order the items so that $\X_C\cup\{x\}<_p\X_B\setminus\{x\}$. Note that the items in $\X_C\cup\{x\}$ (or $\X_B\setminus\{x\}$) are ordered according to $\vp^C$ (or $\vp^B$). Then such a pricing is a dynamic pricing in market $A$ and it is fixed at $x$.

\textbf{Case 
$x\in\X_B\setminus\legal(I_C,\X_B)$.}
Define artificial player $i^{\p}$ and market $B^{\p}$ as in Lemma~\ref{legal in B' iff valuation is 1}. As argued in the previous case, there is a dynamic pricing $\vp^B$ in market $B^{\p}$ that is fixed at $x$. By property~\ref{UL-cond.5}  in Definition~\ref{submarket pair definition}, we can define
\[y:=\arg\min\{p_{\hat{x}}^B:\hat{x}\in\legal(I_C,\X_B)\}\,.\]
Let $C^{\p}:=(\X_C\cup\{y\},I_C,(k_i)_{i\in I_C},\legal)$. As argued in the previous case, there is a dynamic pricing $\vp^C$ in market $C^{\p}$ that is fixed at $y$. Define
\begin{align*}
\X_{B,1}&:=\{\hat{x}\in\X_B\setminus\{x\}:p_{\hat{x}}^B<p_y^B\}\\
\X_{B,2}&:=\{\hat{x}\in\X_B\setminus\{x\}:p_{\hat{x}}^B>p_y^B\}\,.
\end{align*}
We define a price vector $\vp^A$ with the prices of items in the following order
\[\{x\}<_p\X_{B,1}<_p\{y\}<_p\X_C<_p\X_{B,2}\,.\]
The items in the sets $\X_{B,1}$, $\X_C$, and $\X_{B,2}$ are ordered according to $\vp^B$, $\vp^C$, and $\vp^B$, respectively. By definition of $y$, we have $\legal(I_C,\X_{B,1})=\varnothing$. Then $\vp^A$ is a dynamic pricing in market $A$ and it is fixed at $x$.

\textbf{Case $x\in\X_C$.}
By property~\ref{UL-cond.5} in Definition~\ref{submarket pair definition}, we can select an item $y\in\legal(I_C,\X_B)$. Define artificial player $i^{\p}$ and market $B^{\p}$ as in Lemma~\ref{legal in B' iff valuation is 1}. Let $C^{\p}:=(\X_C\cup\{y\},I_C,(k_i)_{i\in I_C},\legal)$. As argued in the first case above, there is a dynamic pricing $\vp^B$ in market $B^{\p}$ that is fixed at $y$ and a dynamic pricing $\vp^C$ in market $C^{\p}$ that is fixed at $x$. 

We define a price vector $\vp^A$ with the prices of items in the following order
\[\{x\}<_p\X_C\setminus\{x\}\cup\{y\}<_p\X_B\setminus\{y\}\,.\]
The items in the sets $\X_C\setminus\{x\}\cup\{y\}$ and $\X_B\setminus\{y\}$ are ordered according to $\vp^C$ and $\vp^B$, respectively. Then $\vp^A$ is a dynamic pricing in market $A$ and it is fixed at $x$.

\end{proof}

By Lemma~\ref{legal assignment of size 2}, we know that whenever there is a non-extendable legal assignment of two items that respects demand, we are able to reduce the problem to submarkets. So, from now on, we assume that any legal assignment of two items that respects demand is extendable. With this assumption, we can show that whenever there is a unit-demand player in market $A$, we can always find a dynamic pricing fixed at any item $x\in\X_A$.

\begin{lemma}\label{unit-demand, pricing fixed at x}
Suppose $\hat{i}\in I_A$ is unit-demand, then for every item $x\in\X_A$, there is a dynamic pricing fixed at $x$.
\end{lemma}

\begin{proof}
Consider an item $x\in\X_A$ and let us consider two cases below.

\textbf{Case  $\legal(\hat{i})=\{x\}$.}
Then, item $x$ is only legal to player $\hat{i}$. Let $B:=(\X_A\setminus\{x\},I_A\setminus\{\hat{i}\},(k_i)_{i\in I_A\setminus\{\hat{i}\}},\legal)$, so $B$ is saturated. Since $x$ is assigned to $\hat{i}$ in any legal allocation in market $A$, we have that an item $y\in\X_A$ is legal to player $i\in I_A\setminus\{\hat{i}\}$ in market $A$ if and only if $y$ is legal to $i$ in market $B$. So, $B$ satisfies property~\eqref{legality-valuation equivalence}. By induction hypothesis, there is a dynamic pricing $\vp^B$ in market $B$. We can extend $\vp^B$ to prices $\vp^A$  by pricing $x$ so that $\{x\}<_p\X_A\setminus\{x\}$. This yields a dynamic pricing in $A$ that is fixed at $x$.

\textbf{Case  $\legal(\hat{i})\setminus\{x\}\neq\varnothing$.}
Consider $\hat{x}\in\legal(\hat{i})\setminus\{x\}$. We define a new market $B:=(\X_A\setminus\{\hat{x}\},I_A\setminus\{\hat{i}\},(k_i)_{i\in I_A\setminus\{\hat{i}\}},\legal)$ so $B$ is saturated. If there exists an item $x^*\in\X_A\setminus\{\hat{x}\}$ and player $i^*\in I_A\setminus\{\hat{i}\}$ such that a legal assignment of $\hat{x}$ to $\hat{i}$ is not extendable, then a dynamic pricing in $A$ fixed at $x$ exists by Lemma~\ref{legal assignment of size 2}. Otherwise, if all such legal assignments are extendable, then $B$ satisfies property~\eqref{legality-valuation equivalence}. By induction hypothesis, there is a dynamic pricing $\vp^B$ in market $B$ that is fixed at~$x$. We can extend $\vp^B$ to prices $\vp^A$  by pricing $x$ so that $\X_A\setminus\{\hat{x}\}<_p\{\hat{x}\}$. This yields a dynamic pricing in $A$ that is fixed at $x$.
\end{proof}

We now assume that every player is either bi-demand or tri-demand and consider legal assignments of three items. Notice that we do not need to go further than three, since every player has demand at most three. As long as there is a pricing such that every legal assignment of three items is extendable, we are sure that any legal bundle chosen by a player can be extended. Instead of considering the legal assignment of any three items, we assume that one of the three items is the item at which we want to fix the dynamic pricing. We will later show that this is sufficient. We again use the concept of a submarket pair, with slight modifications as in the following definition.

\begin{definition}\label{UL-cond new}
We say that submarkets $B=(\X_B,I_B,(k_i)_{i\in I_B},\legal)$ and $C=(\X_C,I_C,(k_i)_{i\in I_C},\legal)$ form a generalized submarket pair $(B,C)$ in market $A$ if the following conditions are satisfied.
\begin{enumerate}[ref={\arabic*}]
\item $(B,C)$ satisfies conditions \ref{UL-cond.1}, \ref{UL-cond.2} and \ref{UL-cond.3} in Definition~\ref{submarket pair definition}.
\item $|\X_B|-2=\sum\limits_{i\in I_B}k_i$ and $|\X_C|+2=\sum\limits_{i\in I_C}k_i$.\label{UL-cond.4 new}
\item $|\legal(I_C,\X_B)|\ge3$ and for all $\hat{x}_1,\hat{x}_2\in\legal(I_C,\X_B)$, there is a legal allocation $\bcO$ in market $A$ such that $ \bigcup\limits_{i\in I_B}\cO_i=\X_B\setminus\{\hat{x}_1,\hat{x}_2\}$ and $\bigcup\limits_{i\in I_C}\cO_i=\X_C\cup\{\hat{x}_1,\hat{x}_2\}$.\label{UL-cond.5 new}
\end{enumerate}
\end{definition}

\begin{lemma}\label{legal assignment of size 3}
Suppose in $A$ any legal assignment of two items can be extended to a legal allocation.
For an item $x_1\in\X_A$, suppose there exist two items $x_2,x_3\in\X_A$ such that a legal assignment of $\{x_1,x_2,x_3\}$ that respects demand is not extendable, then there is a dynamic pricing fixed at $x_1$.
\end{lemma}

\begin{proof}
Suppose $x_1,x_2,x_3$ are assigned to players $i_1,i_2,i_3\in I_A$, respectively, where $i_1,i_2,i_3$ are not necessarily distinct. The following claim and proof is analogous to Lemma~\ref{legal assignment of size 2 non extendable implies submarket pair}.

\begin{claim}\label{legal assignment of size 3 non extendable implies generalized submarket pair}
There is a generalized submarket pair $(B,C)$ in market $A$.
\end{claim}

\begin{proof}
We consider the same bipartite graph $H$ as defined in Lemma~\ref{legal assignment of size 2 non extendable implies submarket pair}. We remove $x_1,x_2,x_3$ and one copy of $i_1,i_2,i_3$ from $H$ to obtain graph $H^*$, which has no perfect matching. Using Hall's marriage theorem, there exists $I_B\subseteq I_A$ such that $|\X^*|<|I_B^*|$, where $|\X^*|=|\legal(I_B,\X_A\setminus\{x_1,x_2,x_3\})|$ and $|I_B^*|$ denotes the total number of copies of players in $I_B$ that are in $H^*$. By a similar argument as in Lemma~\ref{legal assignment of size 2 non extendable implies submarket pair}, since we assume that any legal assignment of two items can be extended to a legal allocation, we have
\[i_1,i_2,i_3\notin I_B\quad\text{and}\quad|I_B^*|=\sum\limits_{i\in I_B}k_i.\]
Moreover, if we only assign $x_2$ and $x_3$ to $i_2$ and $i_3$, then by Hall's marriage theorem, we have
\[
|\X^*|<|I_B^*|\le|\legal(I_B,\X_A\setminus\{x_2,x_3\})|=|\X^*\cup\{x_1\}|
\]
and so
\begin{equation}|\X^*\cup\{x_1\}|=|I_B^*|=\sum\limits_{i\in I_B}k_i\,.\label{submarket B has 2 extra items}
\end{equation}
We define $\X_B:=\X^*\cup\{x_1,x_2,x_3\}$, $\X_C:=\X_A\setminus\X_B$, and $I_C:=I_A\setminus I_B$. Let $B:=(\X_B,I_B,(k_i)_{i\in I_B},\legal)$ and $C:=(\X_C,I_C,(k_i)_{i\in I_C},\legal)$ so that $B$ and $C$ are submarkets of market $A$. Notice that properties~\ref{UL-cond.1},~\ref{UL-cond.2} and~\ref{UL-cond.3} in Definition~\ref{submarket pair definition}  are still satisfied. We can also show that the property~\ref{UL-cond.4 new} in Definition~\ref{UL-cond new} holds in a similar way as in the proof of Lemma~\ref{legal assignment of size 2 non extendable implies submarket pair}.\\
From
\[\{x_1,x_2,x_3\}\subseteq\legal(\{i_1,i_2,i_3\},\X_B)\subseteq\legal(I_C,\X_B),\]
we get $|\legal(I_C,\X_B)|\ge3$. Consider two items $\hat{x}_1,\hat{x}_2\in\legal(I_C,\X_B)$. Suppose $\hat{x}_1\in\legal(\hat{i}_1)$ and $\hat{x}_2\in\legal(\hat{i}_2)$. Since any legal assignment of two items can be extended to a legal allocation, there is a legal allocation $\bcO$ in market $A$ such that $\hat{x}_1\in\cO_{\hat{i}_1}$ and $\hat{x}_2\in\cO_{\hat{i}_2}$. Thus, we can also show that the property~\ref{UL-cond.5 new} in Definition~\ref{UL-cond new} is satisfied.
By definition, $(B,C)$ is a generalized submarket pair.
\end{proof}

Now, consider a maximal generalized submarket pair $(B^*,C^*)$, where $B^*=(\X_B^*,I_B^*,(k_i)_{i\in I_B^*},\legal)$, $C^*=(\X_C^*,I_C^*,(k_i)_{i\in I_C^*},\legal)$, and $\X_B\subseteq\X_B^*$. Here maximal is in the sense of Definition~\ref{maximal submarket pair definition}. Note, that $\X_B\subseteq\X_B^*$ guarantees $x_1\in\X_B^*$. We now prove the following claim that is analogous to Lemma~\ref{B' and C' satisfies legal iff valuation = 1 property}.

\begin{claim}\label{generalized B' and C' satisfies legal iff valuation = 1 property}
The following holds.
\begin{enumerate}
\item Let $i^\p$ be an artificial bi-demand player such that $\legal(i^{\p})\coloneqq\legal(I_C,\X_B)$ and \[v_{i^\p}(x)=\begin{cases}
1\quad\text{if }x\in\legal(i^\p)\\
0\quad\text{if }x\notin\legal(i^\p)\,.
\end{cases}\]
Let $I_B^{\p}:=I_B^*\cup\{i^{\p}\}$. Then, the market $B^{\p}\coloneqq(\X_B^*,I_B^{\p},(k_i)_{i\in I_B^{\p}},\legal)$ satisfies property~\eqref{legality-valuation equivalence}.
\item Let $y_1,y_2$ be items in $\legal(I_C^*,\X_B^*)$ such that $y_1\neq y_2$, define $C^{\p}:=(\X_C^{\p},I_C^*,(k_i)_{i\in I_C^*},\legal)$ where $\X_C^{\p}=\X_C^*\cup\{y_1,y_2\}$. Then, the market $C^{\p}$ satisfies property~\eqref{legality-valuation equivalence}.
\end{enumerate}
\end{claim}

\begin{proof}
The first part of the claim follows from an argument similar to Claims~\ref{legal in A implies legal in B'} to \ref{legal in A iff legal in B'}. The only difference is that Definition~\ref{UL-cond new} should be used instead of Definition~\ref{submarket pair definition}. We now focus on the second part of the claim.

Similar to Claim~\ref{legal in C' implies legal in A}, using property~\ref{UL-cond.5 new} of Definition~\ref{UL-cond new}, one can show that if $x\in\X_C^{\p}$ is legal to $i\in I_C^*$ in $C^\p$ then $x\in\legal(i)$.

Similar to Claim~\ref{legal assignment in A is extendable in C'}, for all $y_3\in\X_C^*$, any legal assignment of $\{y_1,y_2,y_3\}$ to $I_C^*$ that respects demand with respect to market $A$ is extendable in market $C^{\p}$. Notice that the constant term ``2'' in equation~\eqref{size of X_B^*} would become ``4''. Since every legal allocation of size two is extendable, the right-hand sides of equations~\eqref{lower bound on X_B union X_B hat} and~\eqref{lower bound on X_C intersect X_C hat} now have an extra $+1$. The proof is otherwise analogous to the proof of Claim~\ref{legal assignment in A is extendable in C'}.

Lastly, let us consider the analogous of Claim~\ref{legal in market A implies legal in market C'}. First, let us consider the case where both $y_1$ and $y_2$ are only legal to the same bi-demand player $i_1\in I_C^*$. For every player $i_3\in I_C^*$ such that $i_3\neq i_1$ and for every item $y_3\in\legal(i_3,\X_C^*)$, results from the previous paragraph imply that assigning $y_1,y_2$ to $i_1$ and $y_3$ to $i_3$ respects demand and extends to a legal allocation in $C^\p$. So, $y_3$ is legal to $i_3$ in market $C^\p$.

Otherwise, consider a player $i_3\in I_C^*$ and item $y_3\in\legal(i_3,\X_C^{\p})$. Consider players $i_1,i_2\in I_C^*$ such that $y_1\in\legal(i_1)$, $y_2\in\legal(i_2)$ and either $|\{i_1,i_2,i_3\}|\ge2$ or $i_1$ is 3-demand. Again, assigning $y_1,y_2,y_3$ to $i_1,i_2,i_3$ is a legal assignment that respects demand and hence extends a legal allocation in market $C^{\p}$. So, $y_1,y_2,y_3$ are legal to $i_1,i_2,i_3$, respectively, in market $C^\p$.

Therefore, we conclude that an item $y\in X_C^\p$ is legal to a player $i\in I_C^*$ in market $C^\p$ if and only if $y\in\legal(i)$ if and only if $v_i(y)=1$. So, market $C^\p$ satisfies property~\eqref{legality-valuation equivalence}.
\end{proof}

\begin{claim}\label{dynamic pricing for legal assignment of size 3}
There is a dynamic pricing fixed at $x_1$ in market $A$.
\end{claim}

\begin{proof}
By induction hypothesis, there is a dynamic pricing $\vp^B$ in $B^{\p}$ that is fixed at $x_1$.  We define the two lowest priced items in $\legal(I_C^*,\X_B)$ as follows.
\begin{align*}
y_1&=\arg\min\{p_x^B:x\in\legal(I_C,\X_B)\}\\
y_2&=\arg\min\{p_x^B:x\in\legal(I_C,\X_B)\setminus\{y_1\}\}.
\end{align*}
Let $C^{\p}=(\X_C^*\cup\{y_1,y_2\},I_C^*,(k_i)_{i\in I_C^*},\legal)$ so that by the previous claim, $C^{\p}$ is a saturated market and satisfies property~\eqref{legality-valuation equivalence}. By induction hypothesis, there is a dynamic pricing $\vp^C$ in market $C^{\p}$ that is fixed at $y_1$. We define the following sets by comparing the prices of items to the price of $y_1$ and $y_2$.
\begin{align*}
\X_{B,1}&=\{x\in\X_B^*:p_x^B\le p_{y_1}^B\},\\
\X_{B,2}&=\{x\in\X_B^*:p_{y_1}<p_x^B\le p_{y_2}^B\},\\
\X_{B,3}&=\{x\in\X_B^*:p_x^B>p_{y_2}^B\},\\
\X_{C,1}&=\{x\in\X_C^*:p_x^C<p_{y_2}^C\},\\
\X_{C,2}&=\{x\in\X_C^*:p_x^C>p_{y_2}^C\}.
\end{align*}
We define a price vector $\vp^A$ with prices of items in the following order
$$\X_{B,1}<_p\X_{C,1}\setminus\{y_1\}<_p\X_{B,2}<_p\X_{C,2}<_p\X_{B,3}.$$
The items in the sets $\X_{B,1},\X_{C,1},\X_{B,2},\X_{C,2},\X_{B,3}$ are ordered according to $\vp^B,\vp^C,\vp^B,\vp^C,\vp^B$, respectively. By definition of $y_1$ and $y_2$, we have $\legal(I_C^*,\X_{B,1}\cup\X_{B,2})=\{y_1,y_2\}$ and by property~\ref{UL-cond.3} of Definition~\ref{submarket pair definition}, $\legal(I_B^*,\X_C)=\varnothing$. It is then straightforward to check that $\vp^A$ is a dynamic pricing in market $A$. Moreover, since $\vp^B$ is fixed at $x_1$, we always have $x_1\in\X_{B,1}$. So, $\vp^A$ is fixed at $x_1$.

\end{proof}
\end{proof}

In the above lemma, one of the hypotheses is that the item at which we would like to fix the dynamic pricing is in a non-extendable legal assignment of size 3. We now consider the case when this is not true. Notice that we now work with the following three assumptions regarding market $A$ and some item $x\in\X_A$. There is no unit-demand player. All legal assignments of size two are extendable. All legal assignments of size three that contains $x$ are extendable. We show in the following lemma that there exists a dynamic pricing fixed at $x$. The idea of the proof is that by removing $x$ from the market, we can drop to a submarket that satisfies property~\eqref{legality-valuation equivalence}. Then, the hypothesis that all legal assignments of size three that contains $x$ are extendable ensures that no non-extendable legal assignments are introduced when we assign the lowest price to $x$.

\begin{lemma}\label{dynamic pricing fixed at an item outside non-extendable legal assignments of size 3}
Suppose there is no unit-demand player in market $A$ and all legal assignments of size two are extendable. Consider an item $x\in\X_A$. If every legal assignment of size three is extendable as soon as it contains $x$ and respects demand, then there is a dynamic pricing fixed at $x$.
\end{lemma}

\begin{proof}
Consider a player $i_x\in I_A$ such that $x\in\legal(i_x)$. Let us define $k^{\p}_i:=k_i$ for all  $i\in I_A$, $i\neq i_x$ and $k^\p_{i_x}:=k_{i_x}-1$. Notice that $k_{i_x}^{\p}>0$ as there are no unit-demand players in $A$. Define a saturated submarket $B=(\X_B,I_A,\mathbf{k}^{\p},\legal)$ with $\X_B=\X_A\setminus\{x\}$. For all $i\in I_A$ and $y\in\legal(i,\X_B)$, the legal assignment of $x$ and $y$ to $i_x$ and $i$ extends to some legal allocation $\bcO$ in market $A$. The restriction of $\bcO$ to market $B$ is a legal allocation of market $B$. It follows that $B$ satisfies property~\eqref{legality-valuation equivalence}.

By induction hypothesis, there is a dynamic pricing $\vp^B$ in market $B$. We define a price vector $\vp^A$ with the prices of items in the following order
$$\{x\}<_p\X_B.$$
The items in the sets $\X_B$ are ordered according to $\vp^B$. Since any legal assignment of three items that contains $x$ is extendable in market $A$ and any legal assignment of three items that does contains $x$ is extendable in market $B$, it is straightforward to check that $\vp^A$ is a dynamic pricing in market $A$ and it is fixed at $x$.

\end{proof}

Given a multi-demand market with a set $\X$ of items and a set $I$ of players. Let us suppose that each player $i\in I$ is $k_i$-demand where $k_i\le3$, and has multi-demand valuation function. We summarize the procedure for finding a dynamic
pricing in a tri-demand market in Algorithm~\ref{algorithm for tri-demand markets with legality function}. The correctness of the algorithm is proved in the following theorem.

\begin{theorem}\label{correctness of algorithm for tri-demand market}
A dynamic pricing exists in a tri-demand market. It can be computed using Algorithm~\ref{algorithm for tri-demand markets with legality function}.
\end{theorem}

\begin{proof}
We first compute the rough prices as in Algorithm~\ref{rough price} and let market $A:=(\X_A,I_A,(\hat{k}_i)_{i\in I_A},\legal)$ be the induced market described in Section~\ref{Fine Price}. We define the $\{0,1\}$-valued valuation function on market $A$, as in line~\ref{lst:line:tri-demand alg. valuation def}, so that market $A$ is a tri-demand saturated market that satisfies property~\eqref{legality-valuation equivalence}. We now proceed by induction on the number of items $m$ plus the number of players $n$ in market $A$.

Base cases of $n\le2$ or $m\le2$ are covered by Lemma~\ref{base case of 2 players} and Corollary~\ref{base case of 2 items}, which corresponds to lines~\ref{lst:line:tri-demand alg. base cases start}-\ref{lst:line:tri-demand alg. base cases end} in Algorithm~\ref{algorithm for tri-demand markets with legality function}. As the induction hypothesis, we assume that in any saturated tri-demand market that satisfies property\eqref{legality-valuation equivalence} and has less than $m$ items or less than $n$ players, there is a dynamic pricing fixed at any item. We now pick an item $x_f$ and consider the following cases to show that the price vector $\vp^F$ defined in each case of Algorithm~\ref{algorithm for tri-demand markets with legality function} is a dynamic pricing fixed at~$x_f$ in market $A$.

First, we consider the case where there is a unit-demand player $\hat{i}\in I_A$. This corresponds to lines~\ref{lst:line:tri-demand alg. case unit-demand player}-\ref{lst:line:tri-demand alg. case unit-demand player case 2 p^F def} in Algorithm~\ref{algorithm for tri-demand markets with legality function}. If $x_f$ is the only legal item for $\hat{i}$, then we can remove $\hat{i}$ and $x_f$ and find a dynamic pricing in the remaining market, which exists by the induction hypothesis. Lemma~\ref{unit-demand, pricing fixed at x} ensures that by assigning $x_f$ the lowest price, as in line~\ref{lst:line:tri-demand alg. case unit-demand player case 1 p^F def}, we obtain a dynamic pricing fixed at $x_f$ in market $A$. Otherwise, we can remove player $\hat{i}$ and some item $\hat{x}$ in $\legal(\hat{i})\setminus\{x_f\}$ and find a dynamic pricing in the remaining market, which again exists by the induction hypothesis. Lemma~\ref{unit-demand, pricing fixed at x} ensures that by assigning $\hat{x}$ the highest price, as in line~\ref{lst:line:tri-demand alg. case unit-demand player case 2 p^F def}, we obtain a dynamic pricing fixed at $x_f$ in market $A$.

Next, we may assume that there is no unit-demand player. Consider the case where there exists a non-extendable legal assignment of size two that respects demand. This corresponds to lines~\ref{lst:line:tri-demand alg. case non-ext. asgmt. of size 2}-\ref{lst:line:tri-demand alg. case non-ext. asgmt. of size 2 case in C p^F def} in Algorithm~\ref{algorithm for tri-demand markets with legality function}. Lemma~\ref{legal assignment of size 2 non extendable implies submarket pair} guarantees the existence of a maximal submarket pair $(B,C)$ defined in line~\ref{lst:line:tri-demand alg. maximal submarket pair}. Lemma~\ref{B' and C' satisfies legal iff valuation = 1 property} ensures that by introducing a unit-demand artificial player with legal items $\legal(I_C,\X_B)$ and $\{0,1\}$-valued valuation function, we obtain a saturated market $B^\p$ that satisfies property~\eqref{legality-valuation equivalence}. Lemma~\ref{B' and C' satisfies legal iff valuation = 1 property} also ensures that by adding any item in $\legal(I_C,\X_B)$ to market $C$, we obtain a saturated market $C^\p$ that satisfies property~\eqref{legality-valuation equivalence}. By induction hypothesis, both markets $B^\p$ and $C^\p$ have dynamic pricings fixed at any item in the respective market. Then, depending on whether $x_f$ is in market $B$ and whether $x_f$ is legal for any player in $I_C$, Algorithm~\ref{algorithm for tri-demand markets with legality function} defines the prices $\vp^F$ differently in lines~\ref{lst:line:tri-demand alg. case non-ext. asgmt. of size 2 case in B and legal to C p^F def}, \ref{lst:line:tri-demand alg. case non-ext. asgmt. of size 2 case in B but not legal to C p^F def} and \ref{lst:line:tri-demand alg. case non-ext. asgmt. of size 2 case in C p^F def}, each of these lines corresponds to one of the three cases in Lemma~\ref{legal assignment of size 2}.

Next, we may assume that there is no unit-demand player in the market and that all legal allocations of two items that respect demand are extendable. Consider the case where there exists a non-extendable legal assignment of three items that respects demand and contains $x_f$. This corresponds to lines~\ref{lst:line:tri-demand alg. case non-ext. asgmt. of size 3}-\ref{lst:line:tri-demand alg. case non-ext. asgmt. of size 3 case not legal to C p^F def} in Algorithm~\ref{algorithm for tri-demand markets with legality function}. Claim~\ref{legal assignment of size 3 non extendable implies generalized submarket pair} guarantees the existence of a maximal generalized submarket pair $(B,C)$ with $x_f\in\X_B$, as defined in line~\ref{lst:line:tri-demand alg. maximal generalized submarket pair}. Lemma~\ref{generalized B' and C' satisfies legal iff valuation = 1 property} ensures that by introducing a bi-demand artificial player with legal items $\legal(I_C,\X_B)$ and $\{0,1\}$-valued valuation function, we obtain a saturated market $B^\p$ that satisfies property~\eqref{legality-valuation equivalence}. By the induction hypothesis, market $B^\p$ has a dynamic pricing fixed at $x_f$. Lemma~\ref{generalized B' and C' satisfies legal iff valuation = 1 property} then ensures that by adding the two lowest-priced items $y_1$ and $y_2$ from $\legal(I_C,\X_B)$ to market $C$, we obtain a saturated market $C^\p$ that satisfies property~\eqref{legality-valuation equivalence}. By induction hypothesis, market $C^\p$ has dynamic pricing fixed at $y_1$. Then, we can apply Claim~\ref{dynamic pricing for legal assignment of size 3} to get a dynamic pricing fixed at $x_f$ in market $A$, as defined in line \ref{lst:line:tri-demand alg. case non-ext. asgmt. of size 3 case not legal to C p^F def} of Algorithm~\ref{algorithm for tri-demand markets with legality function}.

Lastly, consider the case where there is no unit-demand players, all legal assignment of size two are extendable and all legal assignment of size three that contain $x_f$ are extendable. This corresponds to lines~\ref{lst:line:tri-demand alg. case Else}-\ref{lst:line:tri-demand alg. case Else p^F def} in Algorithm~\ref{algorithm for tri-demand markets with legality function}. According to Lemma~\ref{dynamic pricing fixed at an item outside non-extendable legal assignments of size 3}, we can remove $x_f$ and adjust the demands accordingly to obtain a saturated tri-demand submarket $B$ that satisfies property~\eqref{legality-valuation equivalence}. By induction hypothesis, there is a dynamic pricing in market $B$. Lemma~\ref{dynamic pricing fixed at an item outside non-extendable legal assignments of size 3} ensures that assigning $x_f$ the lowest price yields a dynamic pricing fiexed at $x_f$ in market $A$.

In all cases, we have that $\vp^F$ that is a dynamic pricing in market $A$. With the valuation functions defined as in line~\ref{lst:line:tri-demand alg. valuation def} of Algorithm~\ref{algorithm for 2 optimal allocations}, the dynamic pricing of submarket $B^\p$, $C^\p$ or $B$ as defined in Algorithm~\ref{algorithm for tri-demand markets with legality function} must be in the range $\mathbf{0}<\vp^B<\mathbf{1}$ and $\mathbf{0}<\vp^C<\mathbf{1}$ in order to yield positive utility for all items that are legal to a given player. It is straightforward to check that in all cases $\mathbf{0}<\vp^F<\mathbf{1}$. So, $\Delta\vp^F$ preserves the order of a dynamic pricing in market $A$ and $\mathbf{0}<\Delta\vp^F<\Delta$. By Definition~\ref{definition fine prices} of fine prices, $\Delta\vp^F$ is a vector of fine prices of the original market. From Section~\ref{Fine Price}, we know that the pricing $\vp^D$ defined in line~\ref{lst:line:tri-demand alg. p^D def} is the desired dynamic pricing.
\end{proof}

\begin{algorithm}
\caption{Finding dynamic pricing in tri-demand markets.}\label{algorithm for tri-demand markets with legality function}
\begin{algorithmic}[1]
\State Define $\vp^R$ to be the rough prices of the market.
\State Define $\Delta>0$ as in Section $\ref{Fine Price}$.
\State Define $\X_A:=\X\setminus\left(\bigcup\limits_{i\in I}\R_i\right)$, $I_A:=\{i\in I:\R_i\neq\K_i\}$, and $\hat{k}_i:=k_i-|\R_i|$ for all $i\in I_A$.
\State Define market $A:=(\X_A,I_A,(\hat{k}_i)_{i\in I_A},\legal)$.
\State Define
$$\hat{v}_i(x):=\begin{cases}
1\quad&\text{if $x$ is legal to $i$}\\
0\quad&\text{if $x$ is not legal to $i$}
\end{cases}$$
for all $i\in I_A$ and $x\in\X_A$.\label{lst:line:tri-demand alg. valuation def}
\State Define $m:=|\X_A|$ and $n:=|I_A|$.
\State Define $x_f$ to be an item in $\X_A$.
\If{$n=2$}\label{lst:line:tri-demand alg. base cases start}

    \State Define the dynamic pricing $\vp^F$ fixed at $x_f$ by 
    $$p_x^F:=\begin{cases}
    0.2\quad&\text{if }x=x_f\\
    0.5\quad&\text{if }x\in(\legal(i_1)\setminus\legal(i_2))\cup(\legal(i_2)\setminus\legal(i_1))\setminus\{x_f\}\\
    0.8\quad&\text{if }x\in(\legal(i_1)\cap\legal(i_2))\setminus\{x_f\}
    \end{cases}.$$
\ElsIf{$m\le2$}
    \State Define the dynamic pricing $\vp^F$ fixed at $x_f$ by
    $$p_x^F:=\begin{cases}
    0.2\quad\text{if }x=x_f\\
    0.8\quad\text{if }x\neq x_f
    \end{cases}.$$\label{lst:line:tri-demand alg. base cases end}

\ElsIf{$\hat{k}_{\hat{i}}=1$ for some $\hat{i}\in I_A$}\label{lst:line:tri-demand alg. case unit-demand player}
    \If{$\legal(\hat{i})=\{x_f\}$}
        \State Define $\vp^B$ to be a dynamic pricing in market $(\X_A\setminus\{x_f\},I_A\setminus\{\hat{i}\},(\hat{k}_i)_{i\in I_A\setminus\{\hat{i}\}},\legal)$.
        \State Define the dynamic pricing $\vp^F$ fixed at $x_f$ by
        $$p_x^F:=\begin{cases}
        0.2\quad&\text{if }x=x_f\\
        0.2+0.8p_x^B\quad&\text{if }x\in\X_A\setminus\{x_f\}
        \end{cases}.$$\label{lst:line:tri-demand alg. case unit-demand player case 1 p^F def}
    \Else
        \State Define $\hat{x}$ to be an item in $\legal(\hat{i})\setminus\{x_f\}$.
        \State Define $\vp^B$ to be a dynamic pricing fixed at $x_f$ in market $(\X_A\setminus\{\hat{x}\},I_A\setminus\{\hat{i}\},(\hat{k}_i)_{i\in I_A\setminus\{\hat{i}\}},\legal)$.
        \State Define the dynamic pricing $\vp^F$ fixed at $x_f$ by
        $$p_x^F:=\begin{cases}
        0.8p_x^B\quad&\text{if }x\in\X_A\setminus\{\hat{x}\}\\
        0.8\quad&\text{if }x=\hat{x}\\
        \end{cases}.$$\label{lst:line:tri-demand alg. case unit-demand player case 2 p^F def}
    \EndIf
\algstore{3-demand alg p1}
\end{algorithmic}
\end{algorithm}

\begin{algorithm}[H]
\begin{algorithmic}
\algrestore{3-demand alg p1}
\ElsIf{there exists a non-extendable legal assignment of size 2 that respects demand}\label{lst:line:tri-demand alg. case non-ext. asgmt. of size 2}
    \State Define $(B,C)$ to be a maximal submarket pair with $B:=(\X_B,I_B,(\hat{k}_i)_{i\in I_B},\legal)$ and $C:=(\X_C,I_C,(\hat{k}_i)_{i\in I_C},\legal)$.\label{lst:line:tri-demand alg. maximal submarket pair}
    \State Define a unit-demand artificial player $i^{\p}$ with $\legal(i^{\p}):=\legal(I_C,\X_B)$ and
    $$v_{i^\p}(x):=\begin{cases}
    1&\text{if }x\in\legal(i^\p)\\
    0&\text{if }x\notin\legal(i^\p)
    \end{cases}\,.$$
    \State Define $B^{\p}:=(\X_B,I_B^{\p},(\hat{k}_i)_{i\in I_B^{\p}},\legal)$ where $I_B^{\p}:=I_B\cup\{i^{\p}\}$.
    \If{$x_f\in\legal(I_C,\X_B)$}
        \State Define $\vp^B$ to be a dynamic pricing fixed at $x_f$ in market $B^{\p}$.

        \State Define $\vp^C$ to be a dynamic pricing fixed at $x_f$ in market $(\X_C\cup\{x_f\},I_C,(\hat{k}_i)_{i\in I_C},\legal)$.
        \State Define the dynamic pricing $\vp^F$ fixed at $x_f$ by
        $$p_x^F:=\begin{cases}
        0.5p_x^C\quad&\text{if }x\in\X_C\cup\{x_f\}\\
        0.5+0.5p_x^B\quad&\text{if }x\in\X_B\setminus\{x_f\}
        \end{cases}.$$\label{lst:line:tri-demand alg. case non-ext. asgmt. of size 2 case in B and legal to C p^F def}
    \ElsIf{$x_f\in\X_B\setminus\legal(I_C,\X_B)$}
        \State Define $\vp^B$ to be a dynamic pricing fixed at $x_f$ in market $B^{\p}$.
        \State Define $y:=\arg\min\{p_x^B:x\in\legal(I_C,\X_B)\}$.

        \State Define $\vp^C$ to be a dynamic pricing fixed at $y$ in market $(\X_C\cup\{y\},I_C,(\hat{k}_i)_{i\in I_C},\legal)$.
        \State Define the dynamic pricing $\vp^F$ fixed at $x_f$ by
        $$p_x^F:=\begin{cases}
        0.4p_x^B\quad&\text{if }x\in\{x\in\X_B:p_x^B\le p_y^B\}\\
        0.4+0.4p_x^C\quad&\text{if }x\in\X_C\\
        0.8+0.2p_x^B\quad&\text{if }x\in\{x\in\X_B:p_x^B>p_y^B\}
        \end{cases}.$$\label{lst:line:tri-demand alg. case non-ext. asgmt. of size 2 case in B but not legal to C p^F def}
    \Else
        \State Define $y$ to be an item in $\legal(I_C,\X_B)$.
        \State Define $\vp^B$ to be a dynamic pricing fixed at $y$ in market $B^{\p}$.

        \State Define $\vp^C$ to be a dynamic pricing fixed at $x_f$ in market $(\X_C\cup\{y\},I_C,(\hat{k}_i)_{i\in I_C},\legal)$.
        \State Define the dynamic pricing $\vp^F$ fixed at $x_f$ by
        $$p_x^F:=\begin{cases}

        0.5p_x^C\quad&\text{if }x\in\X_C\cup\{y\}\\
        0.5+0.5p_x^B\quad&\text{if }x\in\X_B\setminus\{y\}
        \end{cases}.$$\label{lst:line:tri-demand alg. case non-ext. asgmt. of size 2 case in C p^F def}
    \EndIf
\algstore{3-demand alg p2}
\end{algorithmic}
\end{algorithm}

\begin{algorithm}[H]
\begin{algorithmic}
\algrestore{3-demand alg p2}
\ElsIf{there exists a non-extendable legal assignment of size 3 that respects demand and contains $x_f$}\label{lst:line:tri-demand alg. case non-ext. asgmt. of size 3}
    \State Define $(B,C)$ to be a maximal generalized submarket pair with $B:=(\X_B,I_B,(\hat{k}_i)_{i\in I_B},\legal)$ and $C:=(\X_C,I_C,(\hat{k}_i)_{i\in I_C},\legal)$, such that $x_f\in\X_B$.\label{lst:line:tri-demand alg. maximal generalized submarket pair}
    \State Define a bi-demand artificial player $i^{\p}$ with $\legal(i^{\p}):=\legal(I_C,\X_B)$ and
    $$v_{i^\p}(x):=\begin{cases}
    1&\text{if }x\in\legal(i^\p)\\
    0&\text{if }x\notin\legal(i^\p)
    \end{cases}.$$
    \State Define $B^{\p}:=(\X_B,I_B^{\p},(\hat{k}_i)_{i\in I_B^{\p}},\legal)$ where $I_B^{\p}=I_B\cup\{i^{\p}\}$.
    \State Define $\vp^B$ to be a dynamic pricing fixed at $x_f$ in market $B^{\p}$.
    \State Define $y_1:=\arg\min\{p_x^B:x\in\legal(I_C,\X_B)\}$.
    \State Define $y_2:=\arg\min\{p_x^B:x\in\legal(I_C,\X_B)\setminus\{y_1\}\}$.

        \State Define $\vp^C$ to be a dynamic pricing fixed at $y_1$ in market $(\X_C\cup\{y_1,y_2\},I_C,(\hat{k}_i)_{i\in I_C},\legal)$.
        \State Define the dynamic pricing $\vp^F$ fixed at $x_f$ by
        $$p_x^F:=\begin{cases}
        0.2p_x^B\quad&\text{if }x\in\{x\in\X_B:p_x^B\le p_{y_1}^B\}\\
        0.2+0.2p_x^C\quad&\text{if }x\in\{x\in\X_C:p_x^C<p_{y_2}^C\}\\
        0.4+0.2p_x^B\quad&\text{if }x\in\{x\in\X_B:p_{y_1}^B<p_x^B\le p_{y_2}^B\}\\
        0.6+0.2p_x^C\quad&\text{if }x\in\{x\in\X_C:p_x^C>p_{y_2}^C\}\\
        0.8+0.2p_x^B\quad&\text{if }x\in\{x\in\X_B:p_x^B>p_{y_2}^B\}
        \end{cases}\,.$$\label{lst:line:tri-demand alg. case non-ext. asgmt. of size 3 case not legal to C p^F def}
\Else\label{lst:line:tri-demand alg. case Else}
    \State Take $i_x\in I_A$ such that $x_f\in\legal(i_x)$.
    \State Let $\vp^B$ be a dynamic pricing in market $(\X_A\setminus\{x_f\},I_A,\hat{\mathbf{k}}-\chi(\{i_x\}),\legal)$.
    \State Define the dynamic pricing $\vp^F$ fixed at $x_f$ by
        $$p_x^F:=\begin{cases}
        0.2\quad&\text{if }x=x_f\\
        0.2+0.8p_x^B\quad&\text{if }x\in\X_A\setminus\{x_f\}
        \end{cases}\,.$$\label{lst:line:tri-demand alg. case Else p^F def}
\EndIf
\State Define the dynamic pricing $\vp^D$ by $p_x^D:=\begin{cases}
\vp^R+\Delta\vp^F\quad&\text{if }x\in\X_A\\
\vp^R\quad&\text{if }x\in\X\setminus\X_A
\end{cases}$.\label{lst:line:tri-demand alg. p^D def}
\State \Return $\vp^D$
\end{algorithmic}
\end{algorithm}

\newpage

\bibliographystyle{unsrt}
\bibliography{Literature}

\end{document}